%% file: Multisymplectic_bensoam_bauge_article.tex
\documentclass[%
jmp,%
onecolumn,%
notitlepage,%
showpacs,
]{revtex4-1}

\include{Packages}

\include{Commandes}

\begin{document}
\title[Multisymplectic geometry with symmetry]{Multisymplectic geometry and covariant formalism for mechanical systems with a Lie group as configuration space: application to the Reissner beam}

\author{Joël Bensoam}
\email{bensoam@ircam.fr}
\homepage[]{http://recherche.ircam.fr/equipes/instruments/bensoam/perso/?page=1}
\author{Florie-Anne Baugé}
\email{florie-anne.bauge@ircam.fr}
\affiliation{IRCAM, UMR 9912 STMS (IRCAM/CNRS/UPMC),\\
1 place I. Stravinsky 75004 Paris, France}

\date{\today}% It is always \today, today,

\begin{abstract}
Many physically important mechanical systems may be described with a Lie group $G$ as configuration space. According to the well-known Noether's theorem, underlying symmetries of the Lie group may be used to considerably reduce the complexity of the problems. However, these reduction techniques, used without care for general problems (waves, field theory), may lead to uncomfortable infinite dimensional spaces. As an alternative, the \emph{covariant} formulation allows to consider a finite dimensional configuration space by increasing the number of independent variables. But the geometric elements needed for reduction, adapted to the specificity of covariant problems which admit Lie groups as configuration space, are difficult to apprehend in the literature (some are even missing to our knowledge). To fill this gap, this article reconsiders the historical geometric construction made by E. Cartan in this particular "covariant Lie group" context. Thus, and it is the main interest of this work, the Poincaré-Cartan and multi-symplectic forms are obtained for a principal $G$ bundle. It allows to formulate the \emph{Euler-Poincaré equations} of motion and leads to a Noether's current form defined in the dual Lie algebra.
\end{abstract}

\pacs{Valid PACS appear here}% PACS, the Physics and Astronomy
%Use showkeys class option if keyword
\maketitle

\setcounter{tocdepth}{3}
\tableofcontents

%%%%%%%%%%%%%%%%%%%%%%%%%%%%%%%%%%%%%%%%%%%%%%%%%%%%

\section{Introduction}
%%%%%%%%%%%%%%%%%%%%%%%%%%%%%%%%%%%%%%%%%%%%%%%%%%%%

C.-M. Marle~\cite{Marle2003} gives several physically important mechanical systems
 which configuration space may be identified with a Lie group.
 These systems share remarkable properties which derive from the underlying symmetries of the Lie group. According to the well-known Noether's theorem, symmetries may be related to conserved quantities, named in the modern concept of Lie Groups and Lie algebras: 
\emph{momentum maps}. The property of invariance of these maps allows for considerable 
reduction of the dimension of the phase-space and divides the computational complexity 
by several orders of magnitude, according to the ``symmetry dimension''. Since their 
discovery by Euler \cite{Euler1765} around 1765, momenta have been well illustrated in the 
scientific literature (see Arnold~\cite{Arnold:1966} for example) and are now well settled
 in the symplectic geometry, Hamiltonian theory or Poisson formalism. 

Even if these reduction techniques have been studied thoroughly in the literature (see for example~\cite{marsden:1999}),
their usage for more general problems (waves, field theory) may lead to infinite dimensional manifolds as configuration space.
In this \emph{dynamical} approach,  geodesic curves have to be considered in an infinite dimensional function space.

As an alternative, the \emph{covariant} formulation allows to consider a finite dimensional configuration space (the dimension of the symmetry group itself in our case). This can be achieved by increasing the number of independent variables since the validity of the calculus of variations and of the Noether's theorem is not limited to the previous one-variable setting. Although its roots go back to De Donder~\cite{de1930theorie}, Weyl~\cite{weyl1935geodesic}, Carathéodory~\cite{caratheodory1999calculus}, after J.-M. Souriau in the seventies~\cite{Souriau:1970}, the classical field theory has  been only well understood in the late 20th century (see for example~\cite{kanatchikov1998canonical} for an extension from symplectic to multi-symplectic theory). It is therefore not surprising that, in this covariant (or jet formulation) setting, the geometric constructions needed for reduction have been presented even more recently.

In this context, the derivation of the conserved quantities from the symmetries of the Lie group is not so easy to establish. It can be summarized as follow: the momentum map is no longer a function but must be defined, more generally, as  a \emph{Noether's current}. This form is the interior product of the Poincaré-Cartan form by the fundamental vector field of the Lie group. This definition (difficult to apprehend for a non-specialist audience) can be found in the work of M. Castrill\'on~L\'opez in~\cite{Lopez:2001} for example. Noether 's theorem is then provided by the invariance of this current form along the critical sections (solutions of the problem).

Nevertheless, this definition must be adapted to the specificity of covariant problems which admit Lie groups as configuration space. To our knowledge, it is difficult to find this derivation in the literature - the specific Poincaré-Cartan form is only given, as is, by F. Demoures~\cite{demoures2013multisymplectic}, p. 16, without mentioning the Maurer-Cartan form which however appears in the formula. The specific Noether's current form, obtained by interior product, is also missing in this reference. So the goal of this article is to fill this gap in the literature by reconsidering the historical construction of the Poincaré-Cartan form in this particular "covariant Lie group" context, leading to an adapted current form to express the Noether's theorem.

\subsection{Article organization}
In order to state the subject, the study made by Elie Cartan on variational problems is related in a brief subsection. He has explained, in his "Leçon sur les invariants intégraux"~\cite{Cartan:1922}, the way the Poincaré-Cartan form is obtained and what are its properties. The differential of the Poincaré-Cartan form (called pre-symplectic form after Souriau~\cite{Souriau:1970}) gives rise to the Hamilton's equations of motion and can be related to the Poisson formalism.

After this historical introduction, section \ref{sec:LF} is dedicated to extend the discussion to a more general jet-bundle and leads to what is called now \emph{multi-symplectic geometry}. By the way, it is worth mentioning that Cartan in 1933 in~\cite{Cartan:1933} thought about doing a geometry where geodesics would be replaced by (hyper)surfaces. In this covariant context, the proofs are more laborious but follow the main ideas of Cartan. This leads to a general formalism where theories on reduction by Lie group action can be handled with confidence in the second part of this article (section \ref{sec:III}). The covariant Hamiltonian formalism, successively without and with Lie group considerations, finishes this paper sketching the way to the Lie-Poisson bracket adapted to the \emph{multi-symplectic geometry}.

\section{Historical background}
\subsection{Cartan's lesson}
Let's summarize Elie Cartan's lesson where Hamiltonian formalism is obtained by introducing in a natural way the Poincaré-Cartan form. This form comes from considering variations of the action functional along real trajectories with variable boundary conditions. The action functional is written $  \mathcal{A}  = \int_{t_0(\Eps)}^{t_1(\Eps)}  \Ld ( q, \dot q, t) \, \dd t$, (where $ \Ld( q, \dot q, t)$ is the Lagrangian density of the system and $\dot q$ denotes $\derp{q}{t}$). 

Cartan expresses the action variation $\delta \mathcal{A} =\dd   \mathcal{A} ( Z )$ using a variation vector field (say $ Z(\Eps, t)$ as in appendix~\ref{app:00}). With his "magic rule" for the Lie derivative of any differential form $\alpha$, that is $\Lie{Z} \alpha=Z\contr \dd \alpha+\dd(Z\contr \alpha)$, he may compute (integrating by part)
\begin{eqnarray}\label{eq:dA1}
\delta \mathcal{A} &=&
\left[ \dfrac{\partial  \Ld }{\partial \dot{ q}} \left. \dfrac{\partial  q_\Eps}{\partial \Eps}\right|_{\Eps=0}+  \Ld \, \dd t ( Z)\right]_{t_0(\Eps)}^{t_1(\Eps)} %\\ \nonumber
%&-& 
- \int_{t_0(\Eps)}^{t_1(\Eps)} \left( \dfrac{\dd }{\dd t}\left(\dfrac{\partial  \Ld}{\partial \dot { q}}\right)  -  \dfrac{\partial  \Ld}{\partial q} \right)\;\left. \dfrac{\partial  q_\Eps}{\partial \Eps}\right|_{\Eps=0} \,  \dd t,
 \end{eqnarray}
where $q_\Eps = q(\Eps, t)$ is the varied curve. By choosing a vector field of variation, $Z$, vanishing on the boundaries $t_0$, $t_1$, the first term disappears. The principle of least action then leads to the well known Euler-Lagrange equation of motion 
 \begin{equation}\label{eq:ELc}
\dfrac{\dd }{\dd t}\left(\dfrac{\partial  \Ld}{\partial \dot { q}}\right)  -  \dfrac{\partial  \Ld}{\partial q}=0.
\end{equation}

Furthermore, Elie Cartan continues his discussion by drawing the consequences of the variable boundary conditions (the boundaries $t_i(\Eps)$ depend on the parameter of variation $\Eps$). With this assumption and for real trajectories verifying the Euler-Lagrange equation~(\ref{eq:ELc}), the variation of action reduces to the first term of~(\ref{eq:dA1}) since the second integral vanishes. On the variable boundaries, $t_i(\Eps)$, he then uses the relation (see appendix~\ref{app:00})
\begin{equation}\label{eq:ContactCartan}
\left. \dfrac{\partial  q_\Eps}{\partial \Eps}\right|_{\Eps=0} = \dd  q( Z) -  \dot{ q} \dd t( Z),
\end{equation}
in~(\ref{eq:dA1}) to obtain 
\begin{eqnarray*}
\delta \mathcal{A}&=&\left[ \dfrac{\partial  \Ld }{\partial \dot{ q}} \left(\dd  q( Z) -  \dot{ q} \dd t( Z)\right)+  \Ld \, \dd t ( Z)\right]_{t_0(\Eps)}^{t_1(\Eps)} 
=\left[ \dfrac{\partial  \Ld }{\partial \dot{ q}}\dd q- \left(\dfrac{\partial  \Ld }{\partial \dot{ q}}\dot{ q}-\Ld  \right)\dd t \right]_{t_0(\Eps)}^{t_1(\Eps)}( Z).
\end{eqnarray*}
 This can be written $\dd \mathcal{A}=\left[\PC\right]_{t_0}^{t_1}$ upon introducing the Poincaré-Cartan form
\begin{equation}\label{eq:PCHc}
\PC= p\, \dd  q -\Ha dt,
\end{equation}
with new variables $ p = \derp{ \Ld}{\dot q}$ and $ \Ha=\derp{\Ld}{\dot q}\dot q - \Ld $. Thus, the Legendre transform appears very naturally.

Then, Cartan shows that this computation leads to an integral invariant along critical sections (solutions) of the variational problem. To do so, he considers a collection of real trajectories labeled by a parameter $\Eps $ (see fig.~\ref{fig:tube}). 

\emph{"Finally, suppose that we consider a tube of trajectories, \ie a closed continuous linear
collection of trajectories, each of which is limited to a time interval $[t_0, t_1]$ that varies also with $\Eps $. The formula which gives the variation of the action along these variable trajectories reduces to
\begin{equation}\label{eq:VA}
\delta \mathcal{A}=\left[\PC\right]_{t_0}^{t_1}=(\PC)_{1}-(\PC)_{0}.
\end{equation}
When one returns to the initial trajectory the total variation of the action is obviously zero, in such a way that, if one integrates with respect to $\Eps $ then one will have
\begin{equation*}
\oint\delta\mathcal{A}=0\Leftrightarrow
\oint(\PC)_{1}=\oint(\PC)_{0}
\end{equation*}
$[...]$ given an arbitrary tube of trajectories, if the integral $\oint\PC$ is taken along a closed curve around the tube then that integral will be independent of that curve and will depend only upon the tube..."
}
\begin{figure}[h!]
\centering
            \includegraphics[scale=.6]{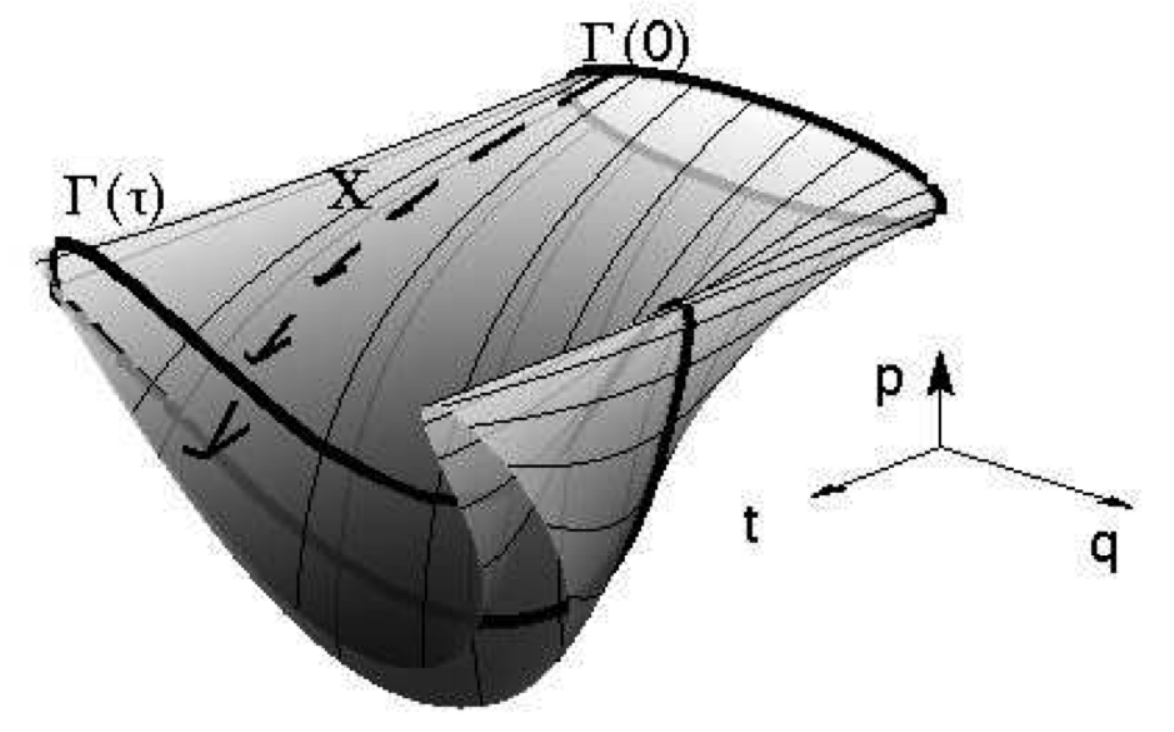}
            \caption{Tube of real trajectories: the integral $I=\oint_{\Gamma}\PC$ does not depend on the choice of the closed loop ${\Gamma}$ around the tube. This quantity is thus an integral invariant that depends only on the choice of the tube.}
            \label{fig:tube}
\end{figure}

So, the quantity $I=\oint\PC$ is invariant \ie does not depend on the closed curve $\Gamma$ along the tube of trajectories. In a modern language, $I$ is invariant along the vector field $X$ tangent to the critical sections, that is
\begin{equation*}
0 = \dd I(X) = \oint_{\Gamma}  \Lie{X}\, \PC 
= \oint_{\Gamma}  X \contr \dd \PC + \oint_{\partial\Gamma} X \contr \PC.
\end{equation*}
The last integral taken over a closed loop is obviously null. Since this property is true for any choice of the tube, one must have $X \contr \dd \PC=0$. So, differentiating~(\ref{eq:PCHc}), Cartan introduces the Poincaré-Cartan form
\begin{equation}\label{eq:PS1}
\PS =-\dd\PC=\dd q\wedge \dd p+\dd\Ha\wedge \dd t
\end{equation}
that verifies ${X}\contr  \PS=0$, for all vector field $X$ tangent to real trajectories. In other words, since $\PS$ is also a closed form, the Lie derivative of the (pre)-symplectic form vanishes
  \begin{equation}\label{symplecticX}
\Lie{X} \PS=0,\quad \text{ $\forall X$ tangent to real trajectories}.
 \end{equation}
He then notices that each coefficient of \mbox{$\delta t=\dd t(X)$},  \mbox{$\delta  q=\dd q(X)$} and $\delta  p=\dd p(X)$ in ${X}\contr  \PS=0$, \ie
\begin{equation*}
(\dd\Ha -\derp{\Ha}{t}\dd t)\delta t-(\dd p +\derp{\Ha}{ q}\dd t)\delta  q+(\dd q-\derp{\Ha}{ p}\dd t)\delta  p=0,
\end{equation*}
must necessarily be canceled. Doing so, he recovers the Hamilton canonical equations of motion
\begin{equation}
\begin{cases}
&\dd\Ha -\derp{\Ha}{t}\dd t=0\\
&\dd p +\derp{\Ha}{ q}\dd t=0\\
&\dd q-\derp{\Ha}{ p}\dd t=0,
\end{cases}
 \text{more commonly written}\quad
 \begin{cases}
 	\dd \Ha( X) = \derp{\Ha}{t}\\
	\dot{p} = -\derp{  \Ha}{ q}\\	
	\dot { q} = \derp{\Ha}{p}.
\end{cases}
\end{equation}
This last formulation is obtained evaluating these forms along a "normalized" vector field $X=\bpm1 &\dot{ q} & \dot { p}\epm^T$$=\bpm \dd t( X) & \dd  q ( X) & \dd  p ( X) \epm^T$, tangent to the trajectories $\bpm t &  q(t) &  p(t) \epm^T$. The two last equations are the well-known Hamilton's equations of motion. In the first equation, the quantity $\dd \Ha(X)$ expresses the variation of the Hamiltonian function along the trajectories. This quantity is often written as $\der{\Ha}{t} = \dot \Ha$ in the literature. If the Hamiltonian function does not depend explicitly on time, \ie $\derp{\Ha}{t} = 0$, then $\Ha$ is invariant along the trajectories. This is known as the conservation of energy (conservative systems).

\subsection{Poisson formalism}

More generally, the invariance~(\ref{symplecticX}) of the pre-symplectic form $\PS$ may be written in another way introducing the canonical symplectic structure 
\begin{equation}\label{eq:canonical}
\Ocanon=\dd q\wedge \dd p.
\end{equation}
Without loss of generality, the vector field $X$ along real trajectories can be normalized to verify $dt(X)=1$. Then for conservative systems with $\dd \Ha(X)=0$ (\scalebox{0.75}{or non-conservative ones} \scalebox{0.75}{$\dd \Ha(X)=\derp{\Ha}{t}$}), the invariance~(\ref{symplecticX}) yields 
\begin{equation}\label{eq:XHam}
X\contr \Ocanon=\dd \Ha\quad\scalebox{0.75}{(or $X\contr \Ocanon=\dd \Ha-\derp{\Ha}{t}\dd t$),}
\quad dt(X)=1.
\end{equation}
The real trajectories obtained from the Hamiltonian $\Ha$ materialize a canonical transformation (an infinitesimal transformation that preserves the pre-symplectic structure $\PS$). Thus, using the canonical symplectic structure $\Ocanon$, the notion of Hamiltonian vector field $Y_F$ associated to a canonical transformation arising from a function $F$ is defined by
\begin{equation*}
Y_F\contr \Ocanon= \dd F\quad\scalebox{0.75}{(or $Y_F\contr \Ocanon= \dd F-\derp{F}{t}\dd t$)},\quad \dd t(Y_F)=1.
\end{equation*}
This last statement can also be written $\Ocanon(Y_F,.)=\dd F(.)$. Thus, considering the Hamiltonian vector field $X_\Ha$ associated to a physical problem, the dynamic of any observable $F$ may be computed by $\Ocanon(Y_F,X_\Ha)=\dd F(X_\Ha)$ where the last term is related to the variation of $F$ along real trajectories: that is $dF(X_\Ha)=\frac{\dd F}{\dd t}=\dot{F}$. This leads, by introducing the Poisson bracket $\{F,\Ha\}=\Ocanon(Y_F,X_\Ha)$, to the Poisson equation
\begin{equation}\label{Poisson}
\dot{F}=\{F,\Ha\}\quad\scalebox{0.75}{(or $\dot{F}=\{F,\Ha\}+\derp{F}{t}$)}.
\end{equation}
In particular in the conservative case, the Poisson bracket of any conserved quantity $J$ with the Hamiltonian $\Ha$ must vanish: $\{J,\Ha\}=0$.

\section{Lagrangian formalism of first-order field theories}\label{sec:LF}

This discussion has started in 1922 in a context where only one independent variable "$t$" was taken into account in variational problems. In modern language, one would speak about \emph{symplectic geometry} where the considered space is a jet-bundle with unidimensional base (\ie endowed with a volume form $\vol=dt$).

This section extends this formalism to a more general jet-bundle. In the context of field theory, the state space $(t, q, v)$ (time, position, velocity) originally conceived by Cartan to study the geometry of the trajectories (curves) according to the optimization principle, is extended to fiber bundles. To be more precise, let $M$ be an orientable manifold and $\pi:E\rightarrow M$ a differentiable fiber bundle with typical fibre $F$ ($\dim F = N$). In the fiber bundle context, more than one independent variables ($\dim M=n+1$) are allowed for (space-time) parametrization and then, the concept of curves is generalized to the concept of sections ($\sect:M\rightarrow E$ in the sequel).

To take into account the velocities, one of the possibilities is to introduce the one-jet fiber bundle $\JE$ (velocity space, for short) characterized by its projection $\jpi:\JE\rightarrow M$. In this way, sections are lifted to the extended jet-bundle $\JE$ to holonomic sections $\prolonge{\sect}:M\rightarrow \JE$ that encode the velocity (see fig.~\ref{fig:HS}). The holonomic concept is used to that purpose. This concept is materialized with the help of a \emph{contact form} $\fcv$ which might be related to the form~(\ref{eq:ContactCartan}) in Cartan's lesson. 

Furthermore, an optimization process seeks the critical section among a family of varied sections generated by an infinitesimal transformation. The vector field tangent to this transformation (say $Z$) needs also to be lifted by a holonomic process ($j^1Z$ in the sequel).
\begin{figure}[htbp]
	\begin{center}
		\includegraphics[width=\textwidth]{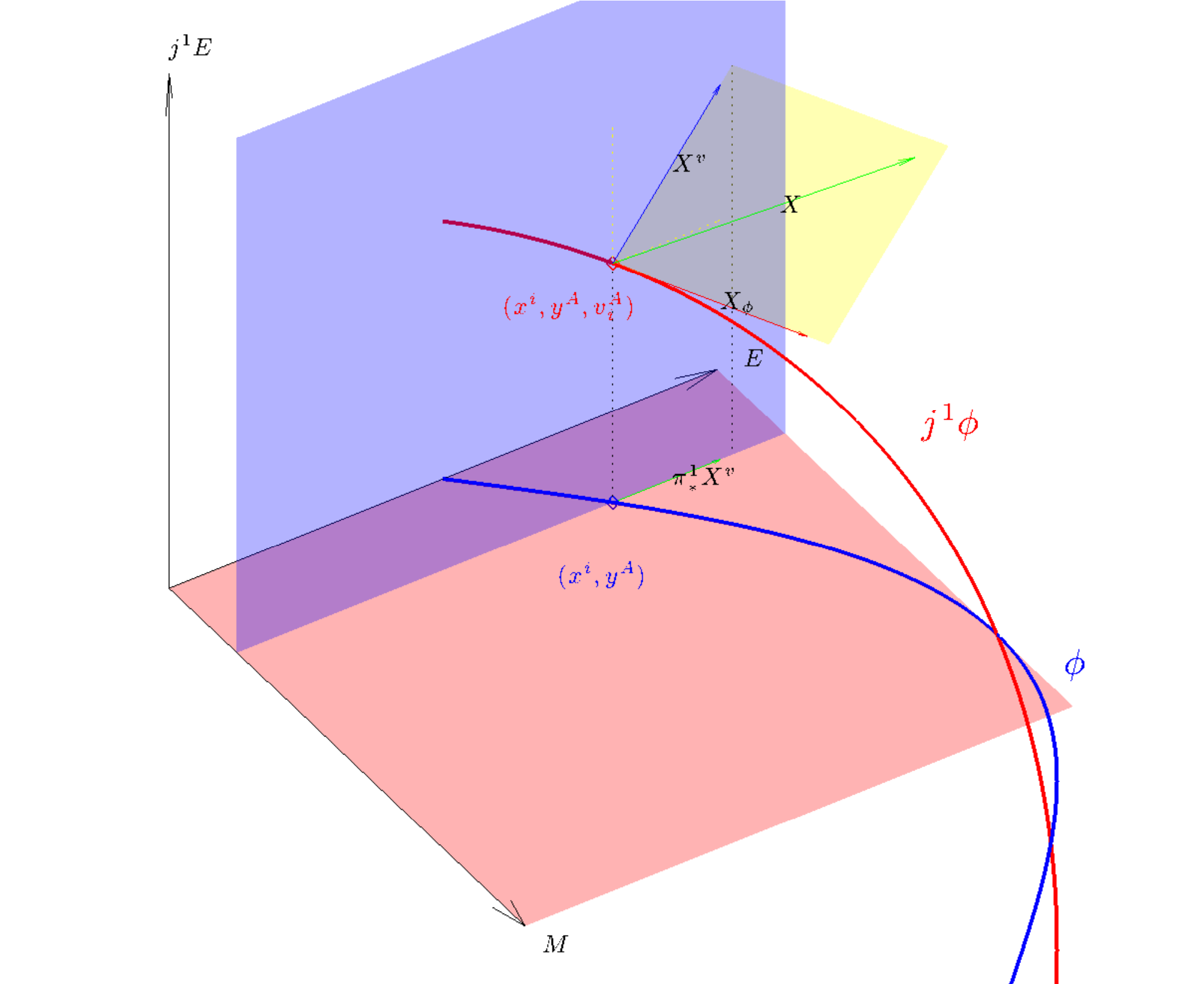}
	\end{center}
	\caption{Sketch of a one-jet fiber-bundle $\JE$: the section $\prolonge{\sect}$ is called the canonical lifting or the canonical prolongation of $\sect$ to $\JE$. A section of $\jpi $ which is the canonical extension of a section of $\pi$ is called a \textbf{holonomic section}. Any vector is a sum of a tangent vector to the section $\prolonge{\sect}$ and a vertical vector $X=X_\sect+X^{v}$.}
\label{fig:HS}
\end{figure}

\subsection{Geometrical structures of first-order jet bundles}\label{app:01}
For all the article, the framework of Arturo Echeverria-Enriquez \& al~\cite{Echeverria00} is adopted as a system of notation. It is quickly related here for convenience of the reader.
Let $M$ be an orientable manifold and $\pi:E\rightarrow M$ a differentiable fiber bundle with typical fibre F.  Let  $\dim M = n + 1$,  $\dim F = N$. The bundle of 1-jets of sections of $\pi$, or 1-jet bundle, is denoted by $\JE$, which is endowed with the natural projection $\pi^1: \JE \rightarrow E$. For every $\ptE \in E$, the fiber of $\JE$ is denoted $\JEpt{\ptE}$ and its elements by $\ptJE$. If $\sect :\mathcal{U}\subset M \rightarrow E$ is a representative of $\ptJE\in \JEpt{\ptE}$, we write  $\ptJE = T_{\pi(\ptE)}\sect$.

In addition, the map $\jpi  = \pi \circ \pi^1: \JE \rightarrow M$ defines another structure of differentiable bundle. The vertical bundle associated with $\pi$ is denoted by $V(\pi)$, that is $V(\pi) = \ker T\pi$, and the vertical bundle associated with $\pi^1$ is written $V(\pi^1) $, that is $V(\pi^1) = \ker T\pi^1$. The corresponding vertical vector fields will be denoted by $\ChpV{\pi}{E}$ and $\ChpV{\pi^1}{\JE}$.

Let $\x{\mu}$, $\mu = 1, \ldots , n+1$, be a local system in $M$ and $\y{A}$, $A = 1, . . . ,N$ a local system in the fibers; that is, $\{\x{\mu}, \y{A}\}$ is a coordinate system adapted to the bundle. In these coordinates, a local section $\sect :\mathcal{U} \subset M\rightarrow E$ is written as $\sect(x) = (\x{\mu}, \sect^A(x))$, that is, $\sect(x)$ is given by functions $\y{A} = \sect^A(x)$. This local system $(\x{\mu},\, \y{A})$ allows to construct a local system $(\x{\mu}, \y{A}, \v{\mu}{A} )$ in $\JE$, where $\v{\mu}{A}$ are defined as follows: if $\ptJE \in \JE$, with $\pi^1(\ptJE) = \ptE$ and $\pi(\ptE) = \b x$, let $\sect :\mathcal{U} \subset M \rightarrow E$, $\y{A} = \sect^A$, be a representative of $\ptJE$, then
\begin{equation*}
\v{\mu}{A} (\ptJE) =\left.\left(\derp{\sect^A}{\x{\mu}}\right)\right|_{\b x}.
\end{equation*}
These systems are called natural local systems in $\JE$. In one of them, we have
\begin{equation*}
\prolonge{\sect} (\b x)=(\x\mu(\b x),\sect^A(\b x),\derp{\sect^A}{\x{\mu}}(\b x)).
\end{equation*}

Regarding to the tangent spaces basis,  $\basex{\mu}$ denotes the basis on $TM$, and $\basey{A}$ the one on the tangent space of the fiber, $TF$. That is $\{\basex{\mu}, \basey{A}\}$ is a basis of $TE$. The dual basis are written $\dx{\mu}$ and $\dy{A}$ and are defined by $\dx{\mu}(\basex{\nu}) = \delta_\nu^\mu$ and $\dy{A}(\basey{B})=\delta_{B}^A$.
If the point $\ptJE=(x^{\mu},\y{A},\v{\mu}{A})$ of the 1-jet bundle $\JE$ is over $\ptE=(x^{\mu},\y{A})$ then it exists a section $\sect$ representative of that point such that $\v{\mu}{A}=\dy{A}(\derp{\sect}{x^{\mu}})$.  The basis associated with the velocity coordinates, $\v{\mu}{A}$, will be denoted $\basev{A}{\mu}$ and its dual basis $\dv{\mu}{A}$.

Given a section $\prolonge\sect$ in the one-jet fiber bundle $\JE$,
any vector $X$ of the tangent space $T_{\ptJE} \JE$ at point  $\ptJE$  may be expressed as a sum of a tangent vector to the section $\prolonge\sect$ and a vertical vector
\begin{equation*}
X=X_\sect+X^{v}, \text{ where } X_\sect=T_\ptE(\prolonge\sect \circ \jpi ) (X) \text{ and }
 X^{v}=X-X_\sect.
\end{equation*}
The vector $X_\sect$ is obtained by the tangent map $X_\sect=T_\ptE \prolonge\sect (X_M)$ where the vector $X_M$  is the projection on the base $M$ of  $X$ given by $X_M=T_\ptE  \jpi ( X)$ (see fig. \ref{fig:HS}).

\subsection{Canonical form and holonomy}\label{subsect:holonomy}
The bundle $\JE$ is endowed with a canonical geometric structure, $\fcv$. To be defined, this canonical form needs the concept of vertical differentiation:
\begin{definition}[Vertical differential]
Let $\sect:M \rightarrow E$ be a section of $\pi$, $ \b x \in M$ and $ \b y= \sect (\b x)$. The vertical differential of
the section $\sect$ at the point $\ptE \in E$ is the map
\begin{eqnarray*}
\dd^V_\ptE \sect : T_\ptE E &\rightarrow & V_\ptE(\pi)\\
u  &\mapsto  &u - T_\ptE(\sect \circ \pi ) u.
\end{eqnarray*}
\end{definition}
If $(\x{\mu}, \y{A})$ is a natural local system of $E$ and $\sect = (\x{\mu}, \sect^A(\x{\mu}))$, then
\begin{align*}	
\dd^V_\ptE \sect(\basex{\mu})=-\derp{\sect^A}{\x\mu}\bigg|_\ptE \basey{A},\quad \quad
\dd^V_\ptE \sect(\basey{A})=\basey{A}.
\end{align*}
As $\dd^V_\ptE \sect$ depends only on $\prolonge{\sect}(\pi(\ptE))$, the vertical differential can be lifted to $\JE$ in the following way:
\begin{definition}[Canonical form (or contact form)]\label{def:canForm}
Consider $\ptJE \in \JE$ with $\ptJE \xrightarrow{\pi^1} \ptE \xrightarrow{\pi} \b x$ and $\bar{u} \in T_{\ptJE}\JE$. The structure canonical form of $\JE$ is a vectorial 1-form $\fcv$ in $\JE$ with values on $V(\pi)$ defined by
\begin{align*}
\fcv(\ptJE; \bar{u})= \dd^V_\ptE \sect (T_{\ptJE}\pi^1(\bar{u}))
\end{align*}
where the section $\sect$ is a representative of $\ptJE$.
\end{definition}
This expression is well defined and does not depend on the representative $\sect$ of $\ptJE$.
On a natural local coordinate system, the contact form $\fcv(\ptJE; .)$ at point $\ptJE=(x^{\mu},\y{A},\v{\mu}{A})$ may be evaluated on the basis vectors $\basex{\mu}$, $\basey{A}$ and $\basev{A}{\mu}$ to give
\begin{align*}
\fcv(\ptJE; \basex{\mu})=-\v{\mu}{A}\basey{A},\quad \quad
\fcv(\ptJE; \basey{A})=\basey{A},\quad \quad
\fcv(\ptJE; \basev{A}{\mu})=0.
\end{align*}
Then, for an arbitrary vector $\bar{u}=\alpha^{\mu}\basex{\mu}+\beta^A\basey{A}+\gamma^A_{\mu}\basev{A}{\mu}$ of $\Chp{\JE}$
\begin{equation*}
\fcv(\ptJE; \bar{u})=
\alpha^{\mu}\fcv(\ptJE; \basex{\mu})+\beta^A\fcv(\ptJE; \basey{A})+\gamma^A_{\mu}\fcv(\ptJE; \basev{A}{\mu})
=\left(\beta^A-\v{\mu}{A}\alpha^{\mu}\right)\basey{A}.
\end{equation*}
From the above calculations, it is clear that $\fcv$ is differentiable and is given by
\begin{align}\label{anx:eqFC}
\fcv_{\ptJE}= (\dy{A}-\v{\mu}{A} \dx{\mu}) \otimes \basey{A}
\end{align}
in a natural local system.

Holonomic sections can be characterized using this canonical form as follows (see~\cite{Echeverria00} p. 8)
\begin{proposition}[Holonomic section]\label{prop:holsec}
Let $\psi:M \rightarrow \JE$ be a section of $\jpi $. The necessary and sufficient condition for $\psi$
to be a holonomic section is that $\psi^*\fcv = 0.$
\end{proposition}
This proposition can also be interpreted in another way
\begin{proposition}[Tangent space to holonomic section]\label{prop:holsec2}
Any vector $ \bar X$ tangent to a holonomic section $\psi$ belongs to the kernel of the contact form
\begin{equation*}
\fcv\big|_{\ptJE}(\bar X) = 0,\quad \forall \bar X\in T_{\ptJE} \psi
\end{equation*}
\end{proposition}
This can be proved by using the pullback definition of the contact form: for any $ u\in T_xM$, $(\psi^*\fcv)\Big|_x (u) = \fcv\Big|_{\ptJE} (T\psi (u))=0$. This shows the result since the tangential map, $\bar X = T\psi(u)$, maps to a tangent vector to the section $\psi$. In particular, if $u$ is one of the basis vector $\basex{\mu}$ of the base $M$, we have $\fc{A}(X_\mu)=0$ with $X_\mu=T\psi(\basex{\mu})$.

\subsection{Principle of least action}\label{sec:PLA}

From now on, $M$ is an oriented manifold and $\vol \in \Lambda^{n+1} (M)$
 is a fixed volume $(n + 1)$-form on $M$
 \begin{equation}\label{eq:volume}
\vol=dx^1\wedge dx^2\wedge\ldots\wedge dx^{n+1}.
\end{equation}
 With these elements well-defined (see Arturo Echeverria-Enriquez \& al~\cite{Echeverria00}) the Lagrangian formalism is used. The  Lagrangian form is written as
 \begin{equation*}
\Lag=\Ld (x^{\mu}, y^A, v^A_{\mu} )\, \vol,\quad \Ld \in \C^\infty (\JE),
\end{equation*}
in a natural local system $(x^{\mu}, y^A, v^A_{\mu} )$ on $\JE$.
So, we can define
\begin{definition}[The Hamilton principle]\label{def:HP}
Let $((E,M; \pi),\Lag)$ be a Lagrangian system. Let $\Gamma_c(M,E)$ be the set of compactly
supported sections of $\pi$ and consider the (action) map 
\begin{eqnarray*}
\mathcal{A}  : \Gamma_c(M,E) &\rightarrow & \R\\
\sect  &\mapsto&\int_{\mathcal{U}}\pullback{\sect}\Lag,\quad \mathcal{U}\subset M
\end{eqnarray*}
The variational problem posed by the Lagrangian form $\Lag$ is the problem of searching for the critical (or stationary) sections of the functional $\mathcal{A}$.
\end{definition}

Following Cartan’s idea, the Poincaré-Cartan $(n+1)$-form is obtained from the Hamilton principle. To do so, we look for “stationary” action with respect to a diffeomorphism that preserves the contact module (see next paragraph) - that is to say, with respect to variations given by a one-parameter transformation group associated to an arbitrary vector field $Z\in \Chp{E}$ (see appendix \ref{anx:deltaA} for details). So, we compute
\begin{equation*}
\delta \mathcal{A} =\left.\der{}{\Eps}\right|_{\Eps = 0}\mathcal{A}(\sect_\Eps)
= \lim_{\Eps\rightarrow 0}
\frac{\mathcal{A}(\sect_\Eps)- \mathcal{A}(\sect)}{\Eps}.
\end{equation*}
Indeed, even in this more complicated circumstance, a similar expression to the Cartan's lesson equation~(\ref{eq:dA1}) may be obtained as 
\begin{equation}\label{eq:dA2}
\delta\mathcal{A}=\int_{\partial\mathcal{U}}\pullback{\sect}(j^1Z\contr \Theta_{\Lag})
-
\int_{\mathcal{U}}\pullback{\sect}\left[ \z{A} \left(
\dd \left(\derp{\Ld}{\v{\mu}{A}}\right)-
\frac{1}{n+1}\derp{\Ld}{\y{A}}\dx{\mu}\right)\wedge\dnx{n}{\mu}\right]
\end{equation}
if the (Lagrangian) Poincaré-Cartan $(n+1)$-form
\begin{equation}\label{eq:PC}
\Theta_{\Lag}=\derp{\Ld}{v^A_{\mu}}\dy{A}\wedge d^{n}x_{\mu}-
\left(\derp{\Ld}{v^A_{\mu}}v^A_{\mu}-
\Ld\right)\vol,
\end{equation}
is introduced with the n-form $d^{n}x_{\mu}=\basex{\mu}\contr \vol$ (see appendix \ref{app:OUF}). The Legendre transformation, with new variables 
\begin{equation}\label{eq:legendre}
p^{\mu}_A=\derp{\Ld}{v^A_{\mu}},\quad \Ha=\derp{\Ld}{v^A_{\mu}}v^A_{\mu}-\Ld,
\end{equation}
follows in a natural way, as it appears clearly in~(\ref{eq:PC}). 
%%%%%%%%%%%%%%%%%%%%

Looking at the preceding calculus, $Z=\alpha^{\mu}\basex{\mu}+\beta^A\basey{A}$ represents the (arbitrary) vector field tangent to a variation. It reflects a local one-parameter transformation group, ${\tau_{\Eps}^Z}$, used to search for the critical (or stationary) sections of the functional $\mathcal{A}$. As it can be seen, the vector field $Z$ needs to be extended to $j^1Z$ in the whole jet-bundle $\JE$. 

The definition generally used to compute the lifted vector field $j^1Z$ (see appendix \ref{sect:jet-pContactInv}) is to say that its flow leaves invariant the contact module $\fcv_{\ptJE}=\fc{A}\otimes \basey{A}$ given in coordinates at point $\ptJE=(x^{\mu},\y{A},\v{\mu}{A})$ by 
\begin{align}\label{eq:contact}
\fc{A}=\dy{A}-\v{\mu}{A} \dx{\mu}.
\end{align}
We rather prefer, in order to study particular configuration spaces, a more geometric definition (see definition \ref{def:JetProl} of appendix \ref{sect:jet-pGeom}) saying that the holonomic lift of a section stay holonomic under the action of a contact transformation. For short, since the Hamilton principle searches an optimum between holonomic sections, the variation process must generate only this type of sections. Both techniques (\ref{sect:jet-pContactInv} and \ref{sect:jet-pGeom}) give the same final result:
%%%%%%%%%%%%%%PROPOSITION%%%%%%%%%%%%%%%%%%%%%%%%%%%%%%%%%%%%
\begin{proposition}[One-jet prolongation (or lift) of vector fields]\label{anx:liftVF:geom}
Considering the natural basis $\bpm\basex{\mu}&\basey{A}&\basev{A}{\mu}\epm$ of $T\JE$, let  $Z=\alpha^{\mu}\basex{\mu}+\beta^A\basey{A}$ be a vector field of $\Chp{E}$. Its one-jet prolongation on $T\JE$, at point $\ptJE=(x^{\mu},\y{A},\v{\mu}{A})$, is the vector field
\begin{equation*}
j^1Z = \alpha^{\mu}\basex{\mu}+\beta^A\basey{A}+
\left(\derp{\z{A}}{\x{\mu}}+ \v{\mu}{B} \derp{\z{A}}{\y{B}}\right)
\basev{A}{\mu}
 ,\quad\text{with} \quad\z{A}=j^1Z\contr \fc{A}=\beta^A- \v{\nu}{A}\alpha^{\nu}.
\end{equation*}
\end{proposition}
%%%%%%%%%%%%%%PROPOSITION%%%%%%%%%%%%%%%%%%%%%%%%%%%%%%%%%%%%
In section \ref{sec:III}, the geometric definition \ref{def:JetProl} of appendix \ref{sect:jet-pGeom} will be more convenient to study the case where the configuration space is identified with a Lie group.

\subsection{Euler-Lagrange equations}\label{parag:EL}

On one hand, choosing  in~(\ref{eq:dA2}) a vector field $Z$ that vanishes on the boundary $\partial {\mathcal{U}}$ gives the Euler-Lagrange field equations $\pullback{\sect}\left[ 
\left(\dd \left(\derp{\Ld}{\v{\mu}{A}}\right)-
\frac{1}{n+1}\derp{\Ld}{\y{A}}\dx{\mu}\right)\wedge\dnx{n}{\mu}\right]=0$, in other words
\begin{equation}\label{eq:EL2}
\derp{}{x_\mu}\derp{\Ld}{v^A_{\mu}}\bigg|_{\prolonge{\sect}}-
\derp{\Ld}{y^A}\bigg|_{\prolonge{\sect}}=0,\quad A=1,\ldots,N.
\end{equation}
\paragraph*{Remark} The Euler-Lagrange 1-form 
\begin{equation}\label{eq:ELF}
\mathcal{T}_{A}^{\mu}=\dd \left(\derp{\Ld}{\v{\mu}{A}}\right)-
\frac{1}{n+1}\derp{\Ld}{\y{A}}\dx{\mu}
\end{equation}
may be related to its Hamiltonian version given later by~\refp{eq:1forms}(b) in the section~\ref{sec:CHF} dedicated to the Hamiltonian formalism: $\mathcal{T}^{A}_{\mu}=dp^{\mu}_A + \frac{1}{n+1}\derp{\Hd}{y^A}dx^{\mu}$.$\diamond$
\subsection{The variation theorem}\label{app:VT}
On the other hand, the Cartan's discussion, leading to the integral invariant, similar to the expression \refp{eq:VA}, may also be obtained by choosing sections $\sect$ verifying the Euler-Lagrange equations~(\ref{eq:EL2}) in order to cancel the second integral of~(\ref{eq:dA2}).
The Cartan's lesson is traduced in modern language by the variation theorem using the Lagrangian pre-multisymplectic $(n+2)$-form on $\JE$, $\PSL =-\dd \PCL $. Starting from the
 convenient expression $\PCL=\derp{\Ld}{\v{\mu}{A}}\fc{A}\wedge\dnx{n}{\mu}+\Lag$, it may be written (see end of appendix \ref{anx:deltaA})
 \begin{equation}\label{PSL}
 \PSL =\fc{A}\wedge\mathcal{T}_{A}^{\mu}\wedge\dnx{n}{\mu},\quad\scalebox{0.75}{or in coordinates $
  \PSL =\dy{A}\wedge \dd\left(\derp{\Ld}{\v{\mu}{A}}\right)\wedge d^{n}x_{\mu}+
 \left(\v{\mu}{A}\dd\left(\derp{\Ld}{\v{\mu}{A}}\right)-\derp{\Ld}{\y{A}}\dy{A}\right)\wedge\vol
 $}.
 \end{equation}
 The variation theorem is as follows
\begin{theorem}[Variation theorem]\label{th:var}
The following assertions regarding a section $\sect$ of the bundle $\pi: E\rightarrow M$ are equivalent
\begin{enumerate}[(i)]
\item $\sect$ is a stationary point of $\mathcal{A}=\int_{\mathcal{U}}\pullback{\sect}\Lag$;
\item the Euler–Lagrange equations~(\ref{eq:EL2}) hold in coordinates;
\item for any vector field $W$ on $\JE$ 
\begin{equation}\label{eq:C}
\pullback{\sect}(W\contr \Omega_{\Lag}) =0;
\end{equation}
\end{enumerate}
where $\prolonge{\sect}$ is a holonomic section.
\end{theorem} 
%%%%%%%%%%%%%%%%%%%%%%%%%%%%%%%%%%%%%%%%%

The followings aim at partially proving the theorem \ref{th:var}. Actually only the needed implications of the theorem are discussed here, that is $(i)\Rightarrow (ii)$ and $(i) \Rightarrow (iii)$.

Let $\sect$ be a section of the bundle $\pi : E\rightarrow M$.

\begin{itemize}
\item $(i) \Rightarrow (ii)$: If $\sect$ is a stationary point of $\mathcal A = \int_{\mathcal{U}} \pullback{\sect} \mathcal L$, then the Euler-Lagrange equations hold in coordinates.
\begin{proof}
The computation of the variation of the action can be found in appendix \ref{anx:deltaA} and leads to the Euler-Lagrange equations~\refp{eq:EL2}.
\end{proof}

\item $(i) \Rightarrow (iii)$: If $\sect$ is a stationary point of $\mathcal A = \int_{\mathcal{U}} \pullback{\sect} \mathcal L$, then $\pullback{\sect} (W\contr \PSL) = 0$ for any vector field $W$ on $\JE$.\\

To prove this implication, one needs the following lemma
\begin{lemma}[]\label{lemmeXOmega}
If $\sect$ is a section of $\pi$ and if either $W$ is tangent to the image of $\prolonge{\sect}$ or $W$ is $\pi^1$-vertical, then
\begin{equation*}%\label{eq:tg}
\pullback{\sect}(W\contr \PSL) =0.
\end{equation*}
\end{lemma}
\begin{proof}

Let us consider a $(n+1)$-vector field $\mvec{X} = (X_1, \cdots, X_{n+1})$ tangent to the section $\prolonge{\sect}$ where each $X_{\alpha}$ is given by the tangential map $X_{\alpha}=T{\prolonge{\sect}}(\basex{\alpha})$ of the basis vector $\basex{\alpha}$. The pulling back of $W\contr \PSL$ is by definition
\begin{equation*}
\pullback{\sect}(W\contr \PSL)(\basex{1},\ldots,\basex{n+1}) =(W\contr \PSL)(\mvec{X})=\PSL(W, \mvec{X}).
\end{equation*}
A calculation using the multi-symplectic $(n+2)$-form $\PSL$ given by~(\ref{PSL}), lemma~\ref{lemma:M}  and its corollary (appendix~\ref{app:HF}) shows that
\begin{eqnarray*}
\PSL(W, \mvec{X})=\fc{A}(W)\left(\mathcal{T}^{\mu}_{A}\wedge\dnx{n}{\mu}\right)(\mvec{X})
+
(-1)^\alpha\fc{A}(X_\alpha)\left(\mathcal{T}^{\mu}_{A}\wedge\dnx{n}{\mu}\right)(W,\hat{\mvec{X}}_{\alpha})
\end{eqnarray*}
where $\mathcal{T}^{\mu}_{A}$ is the Euler-Lagrange 1-form given by~\refp{eq:ELF}. According to the holomonic criteria (proposition \ref{prop:holsec2}), the term $\fc{A}(X_{\alpha})$ obviously vanishes, so we have finally
\begin{equation}\label{eq:PSLproof}
\PSL(W, \mvec{X})=\fc{A}(W)\left(\mathcal{T}^{\mu}_{A}\wedge\dnx{n}{\mu}\right)(\mvec{X})
,\quad \forall W\in \Chp{\JE}.
\end{equation} 
\begin{itemize}
\item[$\diamond$] First, assume that $W$ is tangent to the image of $\prolonge{\sect}$ in $\JE$, that is, $ W=W_{\sect} = T{(\prolonge{\sect})} w$, for some vector $w$ on $TM$. For this choice, $\fc{A}(W)=0$ and then $ \PSL(W, \mvec{X})=0$.
\item[$\diamond$]Second, the contact form $\fcv$ vanishes on $\pi^1$-vertical vectors: $\fc{A}(W)=0$, for $W\in\chi^{V(\pi^1)} (\JE)$ (no component in $\dx{\mu}$ nor in $\dy{A}$).
\end{itemize}
\end{proof}
\begin{proof}
Taking into account lemma~\ref{lemmeXOmega}, the $(i) \Rightarrow (iii)$ implication is partially proved. Finally, for the weaker assumption where $W\in \Chp{\JE}$, \ie  for any vector $W\in TJE$, the section $\sect$ has to be a stationary point of $\mathcal{A}$ in order to cancel, in equation~\refp{eq:PSLproof}, the term $\left(\mathcal{T}^{\mu}_{A}\wedge\dnx{n}{\mu}\right)(\mvec{X})
=\pullback{\sect}( \mathcal{T}^{\mu}_{A}\wedge \dnx{n}{\mu})(\basex{1},\ldots,\basex{n+1})$, that is to say, to satisfy the Euler-Lagrange equations~\refp{eq:EL2}.
\end{proof}
%%%%%%%%%%%%%%%%%%%%%%%%%%%%
\end{itemize}

The variation theorem is really useful especially in presence of symmetry. It allows to demonstrate the first Noether's theorem and to obtain a conserved quantity named \emph{current}.

\subsection{Lagrangian symmetries and Noether’s theorem}

In Mechanics, a symmetry of a Lagrangian dynamical system is a diffeomorphism in the phase
space of the system (the tangent bundle) which leaves the Lagrangian invariant. 
It can be thought of being generated by a vector field (say $S$). This leads to the following definition
\begin{definition}[An infinitesimal natural symmetry]
Let $((E,M; \pi),\Lag)$ be a Lagrangian system. 
An infinitesimal natural symmetry of the Lagrangian system is a vector field $S \in \chi(E)$ such
that its canonical prolongation leaves $\Lag$ invariant (vanishing Lie derivative)
\begin{equation*}
\Lie{(j^1S)}\Lag = 0.
\end{equation*}
\end{definition}
If the vector field $S \in \chi(E)$  is an infinitesimal natural symmetry of the Lagrangian system $((E,M; \pi),\Lag)$ then the Poincaré-Cartan form is also invariant, i.e $\Lie{j^1S} \Theta_{\Lag}=0$. According to the first Noether's theorem, the presence of symmetries leads to conserved quantities. The main
result is
\begin{theorem}[The first Noether's theorem]\label{th:Noether2}
Let $S \in \chi(E)$  be an infinitesimal natural symmetry of the Lagrangian system $((E,M; \pi),\Lag)$. Then,
the n-form \mbox{$J(S) := (j^1S)\contr\Theta_{\Lag}$} is a constant (closed) form on the critical sections of the variational problem posed by $\Lag$.
\end{theorem}
\begin{proof}
Let $\sect: M \rightarrow E$  be a critical section of the variational problem, that is, according to the variation theorem~(\ref{eq:C})
\begin{equation*}
\pullback{\sect} \left(W\contr d\Theta_{\Lag}\right)=0,\quad \forall W\in \chi(\JE).
\end{equation*}
So, in one hand, for the choice $W=j^1S$, we have $\pullback{\sect} \left(j^1S\contr d\Theta_{\Lag}\right)=0$.
Since $\Lag$ is invariant under $S$, we have, on the other hand, the invariance of the Poincaré-Cartan form
\begin{equation*}
0=\Lie{j^1S}\Theta_{\Lag}=d \left( j^1S\contr\Theta_{\Lag}\right)+ j^1S\contr d\Theta_{\Lag}.
\end{equation*}
Therefore, on a critical section
\begin{eqnarray*}
0&=&\pullback{\sect}\Lie{j^1S}\Theta_{\Lag}\\
&=&\pullback{\sect}\left( d \left( j^1S\contr\Theta_{\Lag}\right)\right)+ \cancel{\pullback{\sect}\left( j^1S\contr d\Theta_{\Lag}\right)}\\
&=&\pullback{\sect}\left( d \left( j^1S\contr\Theta_{\Lag}\right)\right)
=d \left[\pullback{\sect}\left( j^1S\contr\Theta_{\Lag}\right)\right]
\end{eqnarray*}
and the result, $\dd \pullback{\sect} J(S)=0$, follows.
\end{proof}
For critical sections $\sect$ of the variational
problem posed by $\Lag$, the expression $\pullback{\sect}J(S)$ is called \emph{Noether’s current} associated with $S$.

%%%%%%%%%%%%%%%%%%%%%%%%%%%%%%%%%%%%%%%%%%%%%%%%%%%%

%%%%%%%%%%%%%%%%%%%%%%%%%%%%%%%%%%%%%%%%%%%%%%%%%%%%
\section{Lagrangian formalism with Lie group}\label{sec:III}
%%%%%%%%%%%%%%%%%%%%%%%%%%%%%%%%%%%%%%%%%%%%%%%%%%%%
To obtain the preceding results, the vector field of variation (say $Z$) needed to be lifted ($j^1Z$). This was handled by considering  the canonical contact form given in coordinates by $\b\fcv{}= (\dy{A}-\v{\mu}{A} \dx{\mu}) \otimes \basey{A}$ (see subsection~\ref{subsect:holonomy} and appendix~\ref{app:jet-p}).

What happens if we now consider a principal bundle where the fiber is a Lie group? What are the expressions of the contact form $\b\fcv{}$ and the lift of a vector field in this context? What does become of the Euler-Lagrange equations of motion and what about the Poincaré-Cartan form $\Theta_{\Lag}$ ? If all these questions had an answer then the Noether's theorem would give a new expression of an invariant current along solutions.

To every Lie group $G$, a Lie algebra $\mathfrak{g}$ is associated, whose underlying vector space is the tangent space of G at the identity element, which completely captures the local structure of the group. Since each velocity can then be translated to  $\mathfrak{g}$, the \emph{Lie algebra} is used in the 1-jet bundle definition and the specific canonical contact form is expressed using the Maurer-Cartan form - the lift of a vector field is adapted accordingly and allows us to apply the principle of least action for a reduced Lagrangian densities. The reduction is due to  symmetry induced by the Lie group action. 

Thus, in the followings (and it is the main interest of this article), the Poincaré-Cartan and multi-symplectic forms are obtained naturally for a principal $G$ bundle by mimicking, step by step, the construction made in the section~\ref{sec:LF} for standard Lagrangian formalism. It allows to formulate the \emph{Euler-Poincaré equations} of motion and leads to a Noether's current defined in the dual Lie algebra in order to study dynamical systems using field theories with symmetry.

%%%%%%%%%%%%%%%%%%%%%%%%%%%%%%%%%%%%%%%%%%%%%%%%%%%%

\subsection{Principal G-bundle, left invariant bases and canonical contact form}

Let consider now a principal $G$ bundle $\pi: E\rightarrow M$ with structure group G. An element of the bundle is written $\ptE=(x^{\mu},\y{A})$ where $\gg=(\y{1},\y{2},\ldots,\y{N})$ belongs to the group $G$, ($\mu = 1, \dots,n+1$ and $A=1, \dots, N$). Let $\baseyTG{A}$, be the left-invariant basis on $TG$ obtained by left translation $T_e L_{\gg}$ from the identity $e$ of the vectors $\basey{A}$ 
\begin{equation*}
\baseyTG{A}= T_e L_{\gg}(\basey{A}).
\end{equation*}
Considering this left-invariant basis, 
if a section ${\sect}$ is a representative of a point $\ptJE=(x^{\mu},\y{A},\xiL^A_\mu)$ of the 1-jet bundle $\JE$ over $\ptE=(x^{\mu},\y{A})$, then we have $\xiL^A_\mu=\left.\mcfL^A\right|_{\ptE}(\derp{\boldsymbol{\sect}}{x^{\mu}})$ (to be more precise $\pullback\sect(\xiL^A_\mu)=\left.\mcfL^A\right|_{\ptE}(\derp{\boldsymbol{\sect}}{x^{\mu}})$). Here, the form $\mcfL^A$ is the Maurer-Cartan 1-form: dual basis of the left-invariant basis $\baseyTG{A}$ defined by $\mcfL^A(\baseyTG{B})=\delta_{B}^A$. The computation of the contact form gives in this context $\fcvG_{\ptJE}=\fcG{A} \otimes \baseyTG{A}$ with
\begin{equation}\label{eq:contactG}
\fcG{A} = \mcfL^A-\xiL^A_{\mu} d x_{\mu}
\end{equation}
The basis associated with the velocity coordinates, $\vL{\mu}{A}$, will be denoted $\basevL{\mu}{A}$ and its dual basis $\dd \vL{\mu}{A}$.
%%%%%%%%%%%%%%%%%%%%%%%%%%%%%%%%%%%%%%%%%%%%%%%%%%

\subsection{Jet prolongation of vector fields: the lift}
Using Lie groups and Lie algebras, the lift of vectors fields requires now to consider the variation of the Maurer-Cartan form $\mcfLv=\mcfL^A\otimes\baseyTG{A}$ involved in the contact form~(\ref{eq:contactG}). It should be notice that this difficulty is not encounter without Lie group considerations: instead of $\mcfL^A$, the previous contact form (\ref{eq:contact}) involves only the invariant 1-form $\dy{A}$. Using the same geometric definition as presented in section \ref{sec:PLA} (and detailed in the appendix \ref{sect:jet-pGeom}), the one-jet prolongation mainly differs by a Lie bracket term from the standard formalism~(Prop. \ref{anx:liftVF:geom})
%%%%%%%%%%%%%%PROPOSITION%%%%%%%%%%%%%%%%%%%%%%%%%%%%%%%%%%%%
\begin{proposition}[One-jet prolongation (or lift) of vector fields with Lie groups]\label{prop:liftG}
Choosing a left representation, \ie using the basis $\bpm \basex{\mu}&\baseyTG{A}&\basevL{\mu}{A}\epm$ of $T\JE$, let $Z=\alpha^{\mu}\basex{\mu}+\beta^A\baseyTG{A}$ be a vector field of $\Chp{E}$. Its one-jet prolongation on $T\JE$, at point $\ptJE=(x^{\mu},\y{A},\xiL^A_\mu)$, is the vector field 
\begin{equation}\label{eq:Zv2}
j^1Z = \alpha^{\mu}\basex{\mu}+\beta^A\baseyTG{A}+
\left(
\derp{\zG{A}}{x^{\mu}}+\xiL^C_{\mu}T^B_C
\derp{\zG{A}}{y^B}
+\crochetL{\boldsymbol{\xiL}_{\mu}}{\boldsymbol{\beta}}^A\right)
\basevL{\mu}{A}
\quad\text{with}\quad \zG{A}=j^1Z\contr\fcG{A}=\beta^A-\xiL^A_{\nu} \alpha^{\nu},
\end{equation}
where $T^B_{C}= \dy{B}(\baseyTG{C})$ (see app. \ref{app:CB} for the change of basis between $\basey{B}$ and $\baseyTG{C}$).
\end{proposition}
%%%%%%%%%%%%%%PROPOSITION%%%%%%%%%%%%%%%%%%%%%%%%%%%%%%%%%%%%
All the details of the calculus are given in appendix \ref{app:LVFG}.
\subsection{Reduced Lagrangian, Euler Poincaré equations and Poincaré-Cartan form}
The reduced Lagrangian form is now written as
 \begin{equation}\label{eq:LR}
\lag=\ld (x^{\mu}, \y{A}, \xiL^A_{\mu} )\vol,\quad \ld \in \C^\infty (\JE),\; \vol \in \Lambda^{n+1} (M)
\end{equation}
in a natural local system $(x^{\mu}, \y{A}, \xiL^A_{\mu} )$ on $\JE$. 
%The function $\ld$ is called the reduced Lagrangian function associated with $\lag$ and $\vol$.
This reduced Lagrangian is used in the Hamilton principle~\ref{def:HP} with the action functional $\mathcal{A}=\int_{\mathcal{U}}\pullback{\sect}\lag$.

To obtain Euler's equation and the Poincaré-Cartan form in this new context, the same procedure to that of the preceding sections is used but this time the lift formula (proposition~\ref{prop:liftG}) is employed to compute $\delta\mathcal{A}$ (see appendix \ref{app:dAG}). Introducing the co-adjoint operator $\ade{}{}$ such that 
\begin{equation*}
\dual{\boldsymbol{\pi}}{\crochetL{\boldsymbol{\xiL}_{\mu}}{\boldsymbol{\beta}}}=
\dual{\boldsymbol{\pi}}{\ad{\boldsymbol{\xiL}_{\mu}}{\boldsymbol{\beta}}}=
\dual{\ade{\boldsymbol{\xiL}_{\mu}}{\boldsymbol{\pi}}}{\boldsymbol{\beta}},
\quad \forall\b\pi\in\mathfrak{g}^*
\end{equation*}
 the (Lagrangian) Poincaré-Cartan $(n+1)$-form must be written in coordinates
\begin{equation}\label{eq:G_PC}
\Theta_{\lag}=\derp{\ld}{\xiL^A_{\mu}}\mcfL^A\wedge d^{n}x_{\mu}-
\left(\derp{\ld}{\xiL^A_{\mu}}\xiL^A_{\mu}-
\ld\right)\vol,
\end{equation}
in order to obtain a similar expression as~(\ref{eq:dA1}) or~(\ref{eq:dA2}). That is, $\forall Z\in\Chp{E}$,
\begin{equation}\label{eq:dA3}
\delta\mathcal{A}=\int_{\partial\mathcal{U}}\pullback{\sect}
\left(
j^1Z\contr\PCl
\right)
-\int_{\mathcal{U}}
\pullback{\sect}\left(
j^1Z\contr
\left(
\fcG{A}\wedge\Gamma_{A}
\right)
\right),
\end{equation}
with the Euler-Lagrange form
\begin{equation}\label{eq:tauG}
\Gamma_{A}=
\dd \left(\derp{\ld}{\vL{\mu}{A}}\right)\wedge\dnx{n}{\mu}
-\left(\ade{\vLv_{\nu}}{\derp{\ld}{\vLv_{\nu}}}\right)_A\vol
-T^{B}_{A}\derp{\ld}{y^B}\vol.
\end{equation}
It obviously furnishes the equations of motion $\pullback{\sect}\left(\Gamma_{A}\right)=0$, named Euler-Poincaré equations, in other words, $\forall A=1,\ldots,N$
\begin{equation}\label{eq:G_EL}
\derp{}{x_\mu}\derp{\ld}{\xiL^A_{\mu}}\bigg|_{\prolonge{\sect}}
-\left(\ade{\boldsymbol{\xiL}_{\mu}}{\derp{\ld}{\boldsymbol{\xiL}_{\mu}}}\right)_A
\bigg|_{\prolonge{\sect}}-
T^{B}_{A}\derp{\ld}{y^B}\bigg|_{\prolonge{\sect}}=0.
\end{equation}
The Legendre transformation with new variables follows in a natural way as it appears clearly in~(\ref{eq:G_PC})
\begin{equation}\label{eq:legendreG}
\momL{\mu}{A}=\derp{\ld}{\xiL^A_{\mu}},\quad \hd=\derp{\ld}{\xiL^A_{\mu}}\xiL^A_{\mu}-\ld.
\end{equation}

The Lagrangian Poincaré-Cartan $(n+2)$-form (or pre-multisymplectic) $\PSl$ in $\JE$  is then introduced from the convenient expression $\PCl=\derp{\Ld}{\vL{\mu}{A}}\fcG{A}\wedge\dnx{n}{\mu}+\lag$,  that is (see end of appendix \ref{app:dAG})
\begin{equation}\label{eq:PSl}
\PSl=-\dd \PCl=
\left(
\fcG{A}\wedge\Gamma_{A}
+\derp{\ld}{\vL{\mu}{A}}\crochetL{\fcvG}{\fcvG}^A
\right)
\wedge\dnx{n}{\mu}.
\end{equation}
\paragraph*{Remark} A more intricate expression can also be obtained without using the contact and Euler-Lagrange forms,
\begin{equation*}
\PSl=\left(
\mcfL^A\wedge \dd\left( \derp{\ld}{\vL{\mu}{A}}\right)+ \derp{\ld}{\vL{\mu}{A}}\crochetL{\mcfLv}{\mcfLv}^A\right)
\wedge d^{n}x_{\mu}
+\left(
\xiL^A_{\mu} \dd\left( \derp{\ld}{\vL{\mu}{A}}\right)-T^{B}_{A}\derp{\ld}{y^B}\mcfL^A
\right)
\wedge \vol.\diamond
\end{equation*}
As it can be seen in~\refp{eq:PSl}, the (n+2)-form $\PSl$ vanishes along critical section $\prolonge\sect$ since it is proportional to the contact form $\fcvG$ and the Euler-Lagrange form $\b\Gamma$. So, the variation theorem for reduced problems is essentially the same as \ref{th:var} 
\begin{theorem}[Variation theorem (for reduced problems)]\label{th:varG}
The following assertions regarding a section $\sect$ of the bundle $\pi: E\rightarrow M$ are equivalent
\begin{enumerate}[(i)]
\item $\sect$ is a stationary point of the reduced action functional $\mathcal{A}=\int_{\mathcal{U}}\pullback{\sect}\lag$;
\item the Euler–Poincaré equations~(\ref{eq:G_EL}) hold in coordinates ;
\item for any vector field $W$ on $\JE$ 
\begin{equation}\label{eq:CG}
\pullback{\sect}(W\contr \Omega_{\lag}) =0;
\end{equation}
\end{enumerate}
where $\prolonge{\sect}$ is a holonomic section.
\end{theorem} 

\vspace{-1cm}

\subsection{Noether's theorem for reduced Lagrangian systems}
\subsubsection{Infinitesimal symmetries}
For the particular case where the group $G$ acts on itself, by a left action with $\b m \in G$
\begin{eqnarray*}
L_{\b m} : G &\rightarrow & G\\
\gg  &\mapsto &\b m\circ \gg
\end{eqnarray*}
and leaves the reduced Lagrangian $\lag$ \refp{eq:LR} invariant, an infinitesimal generator can be defined. Lets take a curve $\b m(s)=\exp( \b\eta s)$ through the identity at $s = 0$ with tangent vector $\b\eta$. It gives rise to the curve $\gg(s)=\b m(s)\circ \gg$  and, from a section $\sect$, to a family of sections $\sect_s=\b m(s)\circ \sect$. By definition, the infinitesimal generator vector field of the left action, at point $\gg$, is given by $S_{\eta}=\der{}{s}\big|_{s=0} \gg(s)$ which is the definition of the tangent map of the right action $R_g$ since
\begin{equation*}
\der{}{s}\bigg|_{s=0} \gg(s)=\der{}{s}\bigg|_{s=0} (\b m(s)\circ \gg)= \der{}{s}\bigg|_{s=0} R_g\left( \b m(s)\right)=T_g R(\b\eta).
\end{equation*}
In other words, the infinitesimal generator $S_{\eta}$ is the right invariant vector field $X^R_{\eta}$ generated by $\b\eta\in\mathfrak{g}$. It coincides, at any point $\gg$, with a left invariant vector field $X^L_{\psi}$ where $\b\psi$ is related to $\b\eta$ by the adjoint operator
\begin{equation}\label{eq:AdjOp}
\b\eta=\Ad{\b g}{\b\psi},\quad\text{or } \b\psi=\Ad{\b g^{-1}}{\b\eta}.
\end{equation}

\subsubsection{Noether current}
If an infinitesimal natural symmetry, $S_{\eta}$, leaves the reduced Lagrangian invariant (\ie $\Lie{j^1 S_{\eta}}\lag=0$), the Noether's theorem~\ref{th:Noether2} may be applied. That is to say that a Noether current $J_{\eta}=j^1 S_{\eta}\contr \Theta_{\lag}$ can be associated to each $\b\eta\in\mathfrak{g}$ using the Poincaré-Cartan form $\Theta_{\lag}$. The question is then to compute the prolongation $j^1 S_{\eta}$ from $S_{\eta}$. Since the proposition \ref{prop:liftG} is obtained using the left invariant basis $\baseyTG{A}$, the fundamental vector field $S_{\eta}$ of symmetry has to be expressed in this basis. As a right invariant vector field, it has constant coordinates on the right basis 
\begin{equation*}
S_{\eta}=\eta^A\baseyR{A}=X^R_{\eta},
\end{equation*}
(from now on, right (resp. left) invariant basis are denoted by $\baseyR{}$ (resp. $\baseyL{}$)).
So, in the left basis $\baseyL{A}$ we have
\begin{equation*}
S_{\eta}=X^L_{\psi}=\psi^A\baseyL{A}=(\Ad{g^{-1}}{\boldsymbol{\eta}})^A\baseyL{A},
\end{equation*}
according to the fact that $\b\psi$ and $\b\eta$ are related by the adjoint operator~\refp{eq:AdjOp}. The prolongation $j^1 S_{\eta}$ can then be written
\begin{equation*}
j^1 S_{\eta}=(\Ad{g^{-1}}{\boldsymbol{\eta}})^A\baseyL{A}+\gamma^A_\mu \basevL{\mu}{A},
\end{equation*}
for $\gamma^A_\mu$ given in proposition \ref{prop:liftG}. Its contraction with the Poincaré-Cartan  form $\Theta_{\lag}$ given in~\refp{eq:G_PC} is then
\begin{eqnarray*}
J_{\eta}&=&j^1 S_{\eta}\contr \Theta_{\lag}\\
&=&\derp{\ld}{\xiL^A_{\mu}}\mcfL^A(j^1 S_{\eta}) d^{n}x_{\mu}-
\left(\derp{\ld}{\xiL^A_{\mu}}\xiL^A_{\mu}-
\ld\right)\vol(j^1 S_{\eta})\\
&=&\derp{\ld}{\xiL^A_{\mu}}(\Ad{g^{-1}}{\boldsymbol{\eta}})^A d^{n}x_{\mu}
=\dual{\piL^\mu_A}{(\Ad{g^{-1}}{\boldsymbol{\eta}})^A} d^{n}x_{\mu}\\
&=&\dual{(\Ade{g^{-1}}{\boldsymbol{\pi}^\mu})^A}{\eta^A} d^{n}x_{\mu}
=\dual{\boldsymbol{\piR}^\mu}{\boldsymbol{\eta}} d^{n}x_{\mu}\quad \forall \boldsymbol{\eta}.
\end{eqnarray*}
It allows to define a Noether current n-form 
\begin{equation}\label{eq:NoetherCurrent}
J=\boldsymbol{\piR}^\mu d^{n}x_{\mu}
\end{equation}
by stating $J_{\eta}=\dual{J}{\boldsymbol{\eta}}$. Here we recognize the right momentum
 $\boldsymbol{\piR}^\mu=\Ade{g^{-1}}{\boldsymbol{\piL}^\mu}$ expressed from the left momentum \mbox{${\piL}^\mu_A=\derp{\ld}{\xiL^A_{\mu}}$}.  The n-form $J$ is constant (closed) on the critical sections of the variational problem posed by $\lag$ which gives the balance law
\begin{equation}\label{eq:NoetherBL}
0=d[\pullback{\sect}J]=\pullback{\sect}dJ=\derp{\boldsymbol{\piR}^\mu}{x^\mu}\, \vol\quad\Leftrightarrow \quad\sum_\mu\derp{\boldsymbol{\piR}^\mu}{x^\mu}=0.
\end{equation}

\subsection{Example: nonlinear model for Reissner Beam}
The Reissner beam is one of the simplest mechanical (acoustical) system that can be treated in the context of mechanics with symmetry. The Lie group SE(3) (Special Euclidean Group) is used to handle the symmetry of the system and is also a manifold on which the dynamic of the beam can be described.
\subsubsection{Reissner kinematics }\label{sec:Reissner}
A beam of length $L$, with cross-sectional area $A$ and mass per unit volume
$\rho$ is considered.  Following the Reissner kinematics, each section
of the beam is supposed to be a rigid body. The beam configuration can
be described by a position $\pos(\s,t)$ and a rotation $\Rot(\s,t)$ of
each section. The coordinate $\s$ corresponds to the position of the
section in a reference configuration $\Sigma_0$ (see
figure~\ref{fig:1}).
\begin{figure}[!ht]
\centering
\includegraphics[width=65mm,keepaspectratio=true]{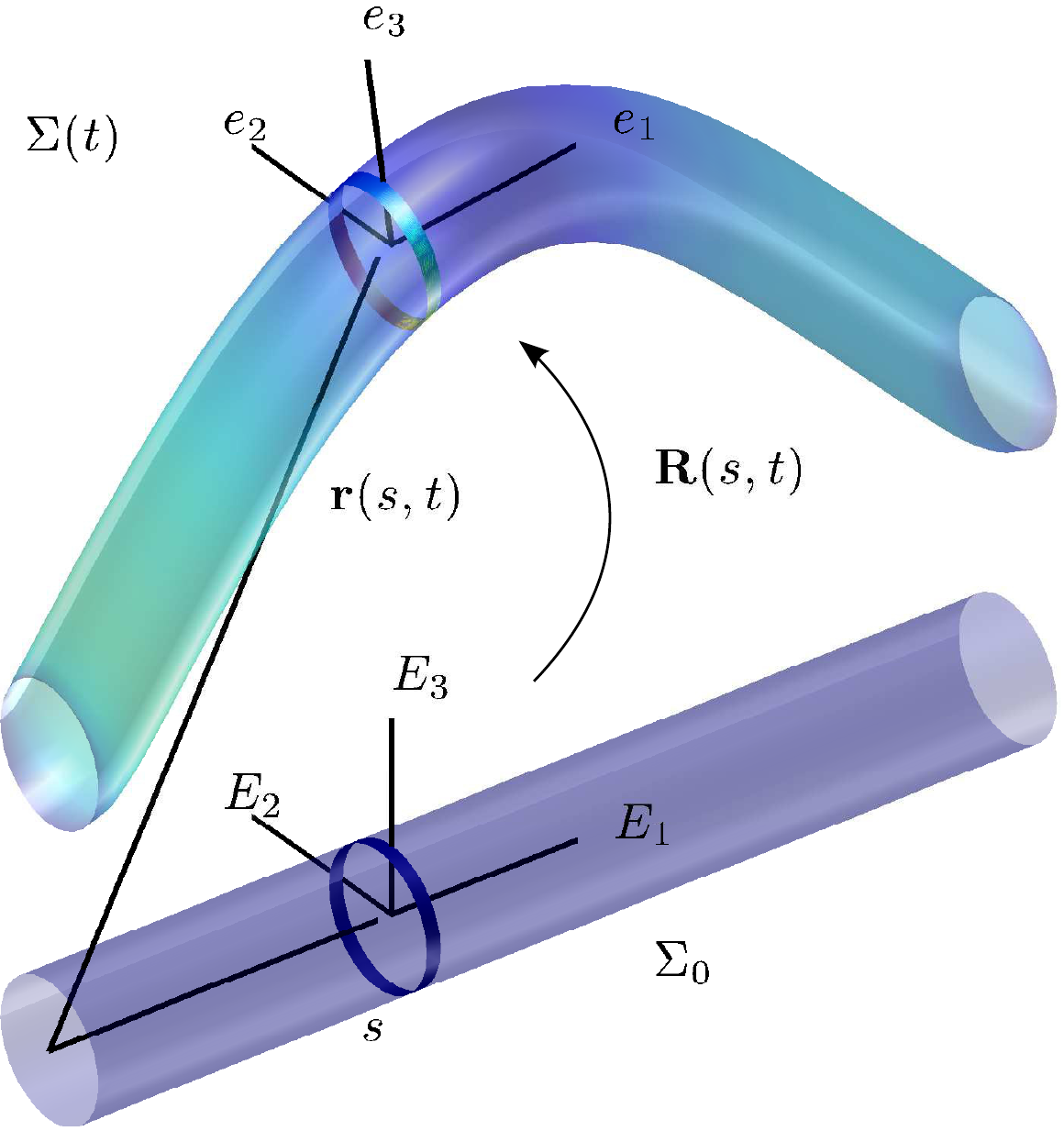}
\caption{Reference and current configuration of a beam. Each section, located at position $\s$ in the reference configuration $\Sigma_0$, is parametrized by a translation $\pos (\s,t)\in\R^3$ and a rotation $\Rot(\s,t)\in SO(3)$ in the current configation $\Sigma_t$.}
   \label{fig:1}
\end{figure}
%\subsubsection{Lie group configuration space}\label{sec:lie}

Any material point of the beam which is located at
$\xp(s,0)=\pos(s,0)+\Yo=\s\Eun+\Yo$ in the reference configuration
($t=0$) has a new position (at time $t$) $\xp(\s,t)=\pos(\s,t)+\Rot(\s,t)\Yo$. In other words, the current configuration of the beam $\Sigma_t$ is completely described by a map 
\begin{equation}\label{eq:H}
\begin{pmatrix}
	\xp(\s,t)\\
	1
\end{pmatrix}
=\underbrace{
\begin{pmatrix}
	\Rot(\s,t)&\pos(\s,t)\\
	0&1
\end{pmatrix}
}_{\gre(s,t)}
\begin{pmatrix}
	\Yo \\
	1
\end{pmatrix},\quad \Rot\in SO(3),\quad \pos\in \R^3,
\end{equation}
where the matrix $\gre(\s,t)$ is an element of the Lie group $SE(3)=SO(3)\times \R^3$, where $SO(3)$ is the group of rotation in $\R^3$. As a consequence, to any motion of the beam corresponds a function (a section of the principal $G$ bundle) $\gre(s,t)$ of the (scalar) independent variables $s$ and $t$. 

Using the reduction procedure, the tangent vectors, $\derp{\gre}{\s}$ and $\derp{\gre}{t}$,  to the group $SE(3)$ at the point $\gre$, are lifted to the Lie-algebra. The general definition, through the Maurer-Cartan form $\mcfLv$, is somewhat technical, but in the case of matrix groups this process is simply a multiplication by the inverse matrix $\gre^{-1}$. Using the left invariant representation, this operation gives rise to the definition of two vectors in the Lie algebra $\mathfrak{g}=\mathfrak{se}(3)\simeq \R^6$
\begin{eqnarray}\label{eq:Ep}
\hat{\Ep}_{L}(\s,t)&=&\gre^{-1}(\s,t)\derp{\gre}{\s}(\s,t)\quad\mapsto \Ep_{L}\\
\label{eq:X}
\hat{\X}_{L}(\s,t)&=&\gre^{-1}(\s,t)\derp{\gre}{t}(\s,t)\quad\mapsto \X_{L},
\end{eqnarray}
which describe the deformations and the velocities of the beam. Assuming a linear stress-strain relation, those definitions allow to define a reduced Langrangian by the difference of kinetic and potential energy $E_{c}-E_{p}$, with
\begin{equation*}
 E_{c}(\X_{L})=\int_{0}^{L}\frac{1}{2} \X_{L}^T\mat{J}\X_{L} d\s,\quad\text{and}\quad
  E_{p}(\Ep_{L})=\int_{0}^{L}\frac{1}{2} (\Ep_{L}-\Ep_{0})^T\mat{C}(\Ep_{L}-\Ep_{0}) d\s,
\end{equation*}
where $\mat{J}$ and $\mat{C}$ are matrix of inertia and Hooke tensor respectively. The deformation of the initial configuration corresponds to $\hat{\Ep}_{0}=\gre^{-1}(\s,0)\derp{\gre}{\s}(\s,0)$. The reduced Lagrangian form yields
\begin{equation*}
\lag=\ld(\X_{L},\Ep_{L})\omega=
\frac{1}{2} \left(
\X_{L}^T\mat{J}\X_{L}-(\Ep_{L}-\Ep_{0})^T\mat{C}(\Ep_{L}-\Ep_{0})
\right)
ds \wedge dt.
\end{equation*}

\subsubsection{Equations of motion}
Applying the Hamilton principle to the left invariant Lagrangian $\lag$ leads to the 
Euler-Poincar\'e equation
\begin{equation}\label{eq:EP}
\partial_{\s}{\Sc}-ad^{*}_{\Ep_{L}}\Sc+\partial_{t}{\Pc}-ad^{*}_{\X_{L}}\Pc=0,
\end{equation}
where the momenta are given by $\Sc=\derp{\ld}{\Ep_{L}}=-\mat{C}(\Ep_{L}-{\Ep}_{0})$ and $\Pc=\derp{\ld}{\X_{L}}=\mat{J}\X_{L}$.  The compatibility condition
\begin{equation}\label{eq:comp}
\partial_{\s}{\X_{L}}-\partial_{t}{\Ep_{L}}=ad_{\X_{L}}\Ep_{L},
\end{equation}
is obtained by differentiating~(\ref{eq:Ep}) and~(\ref{eq:X}) and gives a well-posed problem. It is a particular case of the Maurer-Cartan equation $d\mcfLv+\crochetL{\mcfLv}{\mcfLv}=0$ (see also \refp{eq:MCeq }). Furthermore, it should be noted that the operators $ad$ and $ad^{*}$ in eq.~(\ref{eq:EP})
\begin{eqnarray}
ad^{*}_{(\om,\vt)}(\m,\n)&=&(\m\times\om+\n\times\vt,\n\times\om),\quad
(\om,\vt)\in\mathfrak{g},\quad (\m,\n)\in\mathfrak{g}^*\\
ad_{(\om_{1},\vt_{1})}(\om_{2},\vt_{2})&=&(\om_{1}\times\om_{2},\om_{1}\times\vt_{2}-\om_{2}\times\vt_{1}),
\end{eqnarray}
depend only on the group $SE(3)$ and not on the choice of the particular "metric" $\lag$ that has been chosen to describe the physical problem. 

Equations~(\ref{eq:EP}) and~(\ref{eq:comp}) are written in material or left invariant form ($_L$ subscript). Spatial or right invariant version exists also ($_R$ subscript). In this case, right variables  are introduced by
\begin{eqnarray}\label{eq:Eps}
\hat{\Ep}_{R}(\s,t)&=&\derp{\gre}{\s}(\s,t)\gre^{-1}(\s,t)\quad\mapsto \Ep_{R}\\
\label{eq:Xs}
\hat{\X}_{R}(\s,t)&=&\derp{\gre}{t}(\s,t)\gre^{-1}(\s,t)\quad\mapsto \X_{R},
\end{eqnarray}
and~(\ref{eq:EP}) leads to the 
conservation law
\begin{equation}\label{eq:momenta}
\partial_{\s}\Ss+\partial_{t}\Ps=0
\end{equation}
where $\Ss=\Ade{\ig}\Sc$ and $\Ps=\Ade{\ig}\Pc$. The $\Ade{}$ map for $SE(3)$ is
\begin{equation}
\Ade{\ig}(\m,\n)=(\Rot\m+\pos \times \Rot\n,\Rot\n).
\end{equation}
Compatibility condition~(\ref{eq:comp}) becomes
\begin{equation}\label{eq:comps}
\partial_{\s}\X_{R}-\partial_{t}\Ep_{R}=ad_{\Ep_{R}}\X_{R}.
\end{equation}
Equations~\refp{eq:EP} and~\refp{eq:comp} (or alternatively~\refp{eq:momenta} and~\refp{eq:comps}) provide the exact non linear Reissner beam model.
\subsubsection{Noether current for the Reissner beam}
In that case, the dimension of the base space $M$ is $n+1=2$ and the volume form is given by $\vol=ds\wedge dt$. So, 
\begin{eqnarray*}
d^{1}x_{s}&=&\dpv{s}\contr (ds\wedge dt)=dt\\
d^{1}x_{t}&=&\dpv{t}\contr (ds\wedge dt)=-ds,
\end{eqnarray*}
which give, according to~\refp{eq:NoetherCurrent}, the Noether current (1-form)
\begin{equation*}
J= \Ss dt-\Ps ds.
\end{equation*}
The balance law, already given by~\refp{eq:momenta}, is also obtained using~\refp{eq:NoetherBL} since the pullback $\pullback{\sect}dJ=
\left(\partial_{\s}{\Ss} +\partial_{t}{\Ps}\right)\vol$ vanishes.
%\begin{equation*}
%0=\pullback{\sect}dJ=
%\pullback{\sect}\left(\derp{\Ss}{s} +\derp{\Ps}{t}\right)\vol.
%\end{equation*}

%%%%%%%%%%%%%%%%%%%%%%%%%%%%%%%%%%%%%%%%%%%%%%%%%%%%

%%%%%%%%%%%%%%%%%%%%%%%%%%%%%%%%%%%%%%%%%%%%%%%%%%%%
\section{Covariant or multi-symplectic Hamiltonian formalism}\label{sec:CHF}
%%%%%%%%%%%%%%%%%%%%%%%%%%%%%%%%%%%%%%%%%%%%%%%%%%%%

In this section, an intrinsic covariant (also named multi-symplectic) Hamiltonian formalism is developped using the covariant Legendre transformation for $\Lag$. This is a general discussion that will be adapted in the next section to principal G-bundles.

\subsection{Elements for the covariant Hamiltonian formalism}
The Legendre transformation is a fiber-preserving map between the Jet-bundle and its dual
$
\FLe: \JE\rightarrow \JE^*
$ which has the coordinate expressions (see also~\refp{eq:legendre})
\begin{equation}\label{eq:legendre2}
p^{\mu}_A=\derp{\Ld}{v^A_{\mu}},\quad \Hd=\derp{\Ld}{v^A_{\mu}}v^A_{\mu}-\Ld
\end{equation}
for the multimomenta $p^{\mu}_A$ and the covariant Hamiltonian $\Hd$. In that circunstancies, the Cartan form $\Theta_{\Lag}$~\refp{eq:PC} (resp. the multi-symplectic form $\Omega_{\Lag}=-d\Theta_{\Lag}$~\refp{PSL}) appears to be the pulling back of a corresponding form $\Theta_\Ha$ (resp. $\Omega_\Ha$) on $\JE^*$
\begin{equation*}
\Theta_{\Lag}=\FLe ^* \Theta_\Ha\quad \text{ (resp. }
\Omega_{\Lag}=\FLe ^* \Omega_\Ha).
\end{equation*}
In coordinates, it yields
\begin{eqnarray}\label{eq:PCH}
\Theta_\Ha&=&p^{\mu}_A dy^A\wedge d^{n}x_{\mu}-\Hd\vol\\
\label{eq:SFH}
\Omega_\Ha&=&dy^A\wedge dp^{\mu}_A \wedge d^{n}x_{\mu}+d\Hd\wedge \vol.
\end{eqnarray}
Following Marsden~\cite{Marsden99}, sections in the dual space $\JE^*$ are introduced by the definition
\begin{definition}[Conjugate (dual) section]
Let $\sect$ be a section of the fiber bundle $\pi:E\rightarrow M$ and  $\prolonge{\sect}$ its first jet. A section $\prld{\sect}$ of $\JE^*$ is called conjugate to 
$\prolonge{\sect}$ if 
\begin{equation*}
\prld{\sect}=\FLe \circ \prolonge{\sect} .
\end{equation*}
In this case, we say that $\prld{\sect}$ is holonomic.
\end{definition}
With this definition, the variation theorem~\ref{th:var} is modified to:
\begin{theorem}[Dual variation theorem]\label{th:vard}
If the Legendre transformation $\FLe: \JE\rightarrow \JE^*$ is a fiber diffeomorphism over $E$,
the following assertions regarding a section $\sect$ of the bundle $\pi: E\rightarrow M$ are equivalent
\begin{enumerate}
\item $\sect$ is a stationary point of $\int_{\mathcal{U}}\pullbackd{\sect}\Lag$;
\item the De Donder-Weyl equations \refp{eq:DDWforms} hold in coordinates; 
\item  $\prld{\sect}$ is a Hamiltonian section for $\Ha$, that is to say that for any vector field $W$ in $\Chp{\JE^*}$
\begin{equation}\label{eq:CH}
\pullbackd{\sect}(W\contr \Omega_\Ha) =0
\end{equation}
\end{enumerate}
\end{theorem}

\subsection{De~Donder-Weyl equations}

Let us compute~\refp{eq:CH} according to the multi-symplectic form $\Omega_\Ha$ given by formula~(\ref{eq:SFH})
\begin{equation*}
W\contr \Omega_\Ha
=W\contr 
\left(
\Ocan+d\Hd\wedge \vol 
\right),\quad \forall W \in \Chp{\JE^*}
\end{equation*}
upon introducing the canonical multi-symplectic $(n+2)$-form 
\begin{equation}\label{eq:CM1}
\Ocan=\Psymp{\mu}\wedge d^{n}x_{\mu}=dy^A\wedge dp^{\mu}_A\wedge d^{n}x_{\mu}.
\end{equation}
To this end, the test vector $W$ is written $W=W_{\nu}\basex{\nu}+\ldots$ which gives $\vol(W)=W_{\nu}d^{n}x_{\nu}=dx^{\nu}(W)d^{n}x_{\nu}$ and $d\Hd$ is expressed using its partial derivatives: $d\Hd= \derp{\Hd}{x_\mu}dx^\mu + \derp{\Hd}{y^A}dy^A + \derp{\Hd}{p^{\mu}_A}dp^{\mu}_A$. Futhermore, as we define $d^nx_\mu = \partial_{x_\mu}\contr \vol$ and $d^{n-1}x_{\mu\nu} = \partial_{x_\nu}\contr d^nx_\mu$ (see appendix \ref{app:OUF}),
the expression $W\contr \Omega_\Ha$, may then be factorized to
\begin{eqnarray*}
W\contr \Omega_\Ha
  &=&- \left(d\Hd \wedge d^{n}x_{\nu}- 
  \derp{\Hd}{x_\nu}\vol-\Psymp{\mu} \wedge d^{n-1}x_{\mu\nu}
   \right) dx^\nu(W)\\
&&+\left(dp^{\mu}_A  \wedge d^{n}x_{\mu} +\derp{\Hd}{y^A}\vol \right)\;dy^A(W)\\
&& - \left(dy^A\wedge d^{n}x_{\mu}-\derp{\Hd}{p^{\mu}_A}\vol \right) \; dp^{\mu}_A(W)
\end{eqnarray*}
for any test vector $W$ on $\Chp{\JE^*}$.
The pull-back by $\prld{\sect}$ yields the De~Donder-Weyl equations
\begin{equation}\label{eq:DDWforms}
\begin{cases}
\pullbackd{\sect}\left(d\Hd \wedge d^{n}x_{\nu}- 
  \derp{\Hd}{x_\nu}\vol-\Psymp{\mu} \wedge d^{n-1}x_{\mu\nu}\right)=0& (a)\\
\pullbackd{\sect}\left(dp^{\mu}_A  \wedge d^{n}x_{\mu} +\derp{\Hd}{y^A}\vol \right)=0&(b)\\
\pullbackd{\sect}\left(  dy^A\wedge d^{n}x_{\mu} - \derp{\Hd}{p^{\mu}_A}\vol \right)=0&(c)
\end{cases}
\end{equation}
The first equation surprisingly does not appear in the literature, whether in~\cite{Echeverria2007}, or in~\cite{marsden_shkoller} or even in~\cite{Helein2004c}, where it seems that the same computation has been done. Following E. Cartan, we are convinced not to neglect this equation as it describes the energy balance.

\subsection{The way to the bracket}
The dual variation theorem~\ref{th:vard} may be written in another way. Let us consider a $(n+1)$-vector field \mbox{$\mvec{X} = (X_1, \cdots, X_{n+1})$} tangent to the optimal section $\prld{\sect}$. That is, for each $\mu$, $X_{\mu}=T{\prld{\sect}}(\basex{\mu})$. The pullback~\refp{eq:CH} is by definition
\begin{equation*}
0=\pullbackd{\sect}(W\contr \Omega_\Ha)(\basex{1},\ldots,\basex{n+1}) =(W\contr \Omega_\Ha)(\mvec{X})=\Omega_\Ha(W,\mvec{X}),\quad \forall W\in\Chp{\JE^*}.
\end{equation*}
So, we have $\mvec{X}\contr \Omega_\Ha=0$, that is also (since $\Ocan(W,\mvec{X})=(-1)^{n+1} \Ocan(\mvec{X},W)$),
\begin{equation}\label{CH2}
\mvec{X}\contr \Ocan=(-1)^{n}\left(d\Hd \wedge \vol\right)(.\,,\mvec{X}).
\end{equation}

Introducing the 1-forms ${D\Hd}$ (see app.\ref{app:HF}), $\mathcal{T}^{\mu}_A$,  and $\theta_\Ha^A$, such that 
\begin{equation}\label{eq:1forms}
\begin{cases}
{D\Hd}=d\Hd-\derp{\Hd}{x_\alpha}dx^{\alpha}-\Psymp{\mu}(X_{\mu},X_{\alpha})dx^{\alpha}&(a)\\
\mathcal{T}^{\mu}_A = dp^{\mu}_A +\frac{1}{n+1} \derp{\Hd}{y^A}dx^{\mu}&(b)\\
\theta_\Ha^A=dy^A - \derp{\Hd}{p^{\mu}_A}dx^\mu&(c)
\end{cases}
\end{equation}
where $\theta_\Ha^A$ is the Hamilton version of the contact form~(\ref{eq:contact}), 
the de Donder-Weyl equations \refp{eq:DDWforms} can be written
\begin{equation}\label{eq:1formsDDW}
\begin{cases}
{D\Hd}(X_{\mu})=0\quad\forall \mu&(a)\\
\mathcal{T}^{\mu}_A (X_{\mu})=0 \quad\text{sum over } \mu, \forall A&(b)\\
\theta_\Ha^A(X_{\mu})=0\quad \forall \mu, \forall A&(c)
\end{cases}
\end{equation}

Using the lemma~\ref{lemma:M}, the right hand size of~\refp{CH2} may be evaluted using the 1-form $D\Hd$ to give
\begin{equation*}
\left(d\Hd \wedge \vol\right)(.,\mvec{X})=\left(D\Hd\wedge \vol\right)(.,\mvec{X})\\
=D\Hd+(-1)^{\alpha}\underbrace{D\Hd(X_{\alpha})}_{=0}\vol(.,\hat{\mvec{X}}_{\alpha})\\
=D\Hd.
\end{equation*}
And thus equation~\refp{CH2} yields
\begin{equation}\label{CH3}
\mvec{X}\contr\Ocan=
(-1)^{n} D\Hd,
\end{equation}
where a calculus (see app. \ref{app:HF}) shows that
\begin{equation*}\label{eq:DH}
D\Hd(W)=X_{\mu}\contr\Psymp{\mu}(W^{v})=\left(
\derp{\Hd}{y^A}\theta_\Ha^A+\derp{\Hd}{p^{\mu}_A}\mathcal{T}^{\mu}_A 
\right)(W^{v})=\dd \Hd(W^{v}),\quad \forall W=w_{\alpha}X_{\alpha}+W^{v}\in\Chp{\JE^*}.
\end{equation*}

The formulation~(\ref{CH3}) suggests to use the canonical multi-symplectic form $\Ocan$ given in~\refp{eq:CM1} as a replacement for the standard symplectic form as it essentially differs from classical mechanics by the 1-form 
\begin{equation*}
\Psymp{\mu}(X_{\mu},X_{\alpha})dx^{\alpha}
\end{equation*}
which has no component in $dy^A$ nor in $d p^\mu_A$. For Hamiltonians that do not depend explicitly on base variables $x_\alpha$ ($\derp{\Hd}{x_\alpha}\equiv0$), F. Hélein in~\cite{Helein2004c} rewrites~\refp{CH3} as
\begin{equation}\label{CH3b}
\mvec{X}\contr\Ocan=(-1)^{n}d\Hd\quad \text{mod}\quad dx^\alpha,
\end{equation}
where “mod $dx^\alpha$” means that the equality holds between the coefficients of $dy^A$ and $d p^\mu_A$ in both sides. F. Hélein mentions that "this relatively naive description deserves some critics since equation~(\ref{CH3b}) holds only “mod $dx^\alpha$”, which is not very aesthetic: this reflects a disymmetry between the space-time variables and the field component variables.
This critic can be cured by adding to the set of variable $(x,y,p)$ a further
variable $e\in \R$, canonically conjugate to the space-time volume form $\vol$"~\cite{Helein2004c}. Once again, following E. Cartan, we would prefer to consider the energy balance~\refp{eq:1formsDDW} (a), which specifies that tangent vectors to the critical section $\sect$ belong to the kernel of the 1-form $D\Hd$.

%%%%%%%%%%%%%%%%%%%%%%%

%
  %
    %
      %
         %

\section{Hamiltonian formalism for principal G-bundles}
\subsection{Legendre transformation}
The covariant Legendre transformation for $\lag$ is now constructed as it appears clearly in the Poincaré-Cartan form \refp{eq:G_PC}.
It is a fiber-preserving map between the Jet-bundle and its dual $\FLe: \JE\rightarrow \JE^*$ which has the coordinate expressions
\begin{equation}\label{eq:G_legendre}
\momL{\mu}{A}=\derp{\ld}{\xiL^A_{\mu}},\quad \hd=\derp{\ld}{\xiL^A_{\mu}}\xiL^A_{\mu}-\ld
\end{equation}
for the multimomenta $\momL{\mu}{A}$ and the covariant reduced Hamiltonian $\hd$. In this circumstance, the Cartan form $\Theta_{\lag}$ (resp. $\Omega_{\lag}$) appears to be the pulling back of a corresponding form $\Theta_\ha$ (resp. $\Omega_\ha$)  on $\JE^*$
\begin{equation*}
\Theta_{\lag}=\FLe ^* \Theta_\ha\quad \text{ (resp. }
\Omega_{\lag}=\FLe ^* \Omega_\ha),
\end{equation*}
that is
\begin{eqnarray}\label{eq:G_PCH}
\Theta_\ha&=&\momL{\mu}{A}\mcfL^A\wedge d^{n}x_{\mu}-\hd\vol\\
\label{eq:G_SFH}
\Omega_\ha&=&\left(
\mcfL^A\wedge d\momL{\mu}{A}+\momL{\mu}{A} \crochetL{\mcfLv}{\mcfLv}^A\right)
\wedge d^{n}x_{\mu}
+d\hd\wedge \vol,
\end{eqnarray}
using the Maurer-Cartan equation $d\mcfLv+\crochetL{\mcfLv}{\mcfLv}=0$ (also named zero curvature equation in the literature).

\subsection{De~Donder-Weyl equations}
Using the same computation as in the previous section, equation~(\ref{eq:CH}), \ie $\pullbackd{\sect}(W\contr \Omega_\ha) =0$, is expressed  according to the multi-symplectic form $\Omega_\ha$ used for reduced problems given by formula~(\ref{eq:G_SFH}). To this purpose, the canonical multi-symplectic $(n+2)$-form 
\begin{equation}\label{eq:CMG}
\ps=\psymp{\mu}\wedge d^{n}x_{\mu}=\left(\mcfL^A\wedge d\momL{\mu}{A}+\momL{\mu}{A} \crochetL{\mcfLv}{\mcfLv}^A\right)\wedge d^{n}x_{\mu},
\end{equation}
is introduced to obtain the reduced De~Donder-Weyl equations (let us say De~Donder-Weyl-Poincaré equations)
\begin{equation}\label{eq:RDDWLforms}
\begin{cases}
\prld{\sect}^*
\left[
 d\hd \wedge d^{n}x_{\nu}
-\derp{\hd}{x_\nu}\vol
 -\psymp{\mu}\wedge d^{n-1}x_{\mu\nu}
 \right]=0&(a)\\
\prld{\sect}^*
\left[
\left(
d\momL{\mu}{A}  - 
(\ade{\mcfLv}{\boldsymbol{\pi}^\mu})_A 
 +\frac{1}{n+1}T^B_A\derp{\hd}{y^B}\dx{\mu}
 \right)
 \wedge d^{n}x_{\mu}
\right]=0&(b)\\
\prld{\sect}^*
\left[
\left(
\mcfL^A  -\derp{\hd}{\momL{\nu}{A}}\dx{\nu}
\right)
\wedge d^{n}x_{\mu}
 \right]=0&(c).
\end{cases}
\end{equation}
Note that the last equation is the inverse Legendre transformation. The second equation is the Hamiltonian form of Euler-Poincaré equation~(\ref{eq:G_EL}) with $T^{A}_{B}=\dy{A}(\baseyTG{B})$. It may be written in a more convenient form as
\begin{equation}\label{DDWL}
\derp{\momL{\mu}{A}}{x_{\mu}}\bigg|_{\prld{\sect}}
-(\ade{\derp{\hd}{\boldsymbol{\pi}^\mu}}{\boldsymbol{\pi}^\mu})_A\bigg|_{\prld{\sect}}
+T^{B}_{A}\derp{\hd}{y^B}\bigg|_{\prld{\sect}}=0.
\end{equation}
A right invariant Halmitonian version can also be obtained as (change of sign in the co-adjoint term)
\begin{equation}\label{eq:DDWR}
\derp{\momR{\mu}{A}}{x_{\mu}}\bigg|_{\prolonge{\sect}}
+(\ade{\derp{{\hdR}}{\boldsymbol{\Pi}^\mu}}{\boldsymbol{\Pi}^\mu})_A\bigg|_{\prolonge{\sect}}
+\tilde{T}^{B}_{A}\derp{{\hdR}}{y^B}\bigg|_{\prolonge{\sect}}=0,
\end{equation}
where the Hamiltonian $\hdR$ is expressed with right multi-momentum $\b{\Pi}$ and right velocity $\b\chi\in\mathfrak{g}$ (see also appendix \ref{app:CB} for details about the change of basis given by operators ${T}^{A}_{B}$ and $\tilde{T}^{A}_{B}$).

The De~Donder-Weyl-Poincaré equations~\refp{eq:RDDWLforms} suggest to introduce the reduced 1-forms $D\hd$, $\tau^{\mu}_A$,  and $\theta_\ha^A$, such that
\begin{equation}\label{eq:dDWForms}
\begin{cases}
D\hd=d\hd
-\derp{\hd}{x_\alpha}dx^{\alpha}
 -\psymp{\mu}(X_{\nu},X_{\alpha})dx^{\alpha}&(a)\\
\tau^{\mu}_A = 
d\momL{\mu}{A}  - 
(\ade{\mcfLv}{\boldsymbol{\pi}^\mu})_A 
 +\frac{1}{n+1}T^B_A\derp{\hd}{y^B}\dx{\mu}&(b)\\
\fcG{A}_\ha=\mcfL^A  -\derp{\hd}{\momL{\nu}{A}}\dx{\nu}&(c)
\end{cases}
\end{equation}
where $\fcG{A}_\ha$ is the Hamilton version of the reduced contact form~(\ref{eq:contactG}). These forms vanish on vectors $X_\mu$ tangent to the crital section $\prolonge{\sect}$
\begin{equation}\label{eq:1formsDDWG}
\begin{cases}
{D\hd}(X_{\mu})=0\quad\forall \mu&(a)\\
\tau^{\mu}_A  (X_{\mu})=0\quad \text{sum over } \mu, \forall A&(b)\\
\fcG{A}_\ha(X_{\mu})=0\quad \forall \mu, \forall A&(c)
\end{cases}
\end{equation}
Again equation \refp{eq:CH} of the variation theorem may be formulated as
\begin{equation}\label{CH4}
\mvec{X}\contr\ps=
(-1)^{n} D\hd,
\end{equation}
using the canonical multi-symplectic form $\ps$. For any test vector $W=w_{\alpha}X_{\alpha}+W^{v}\in\Chp{\JE^*}$
\begin{equation*}
D \hd(W)=T^B_A\derp{\hd}{y^B}\mcfL^A(W^{v})+\derp{\hd}{\momL{\alpha}{A}} d\momL{\alpha}{A} (W^{v})=\dd \hd(W^{v}) .
\end{equation*}

%%%%%%%%%%%%%%%%%%%%%%%%%%%%%%%%%%%%%%%%%%%%%%

\subsection{Hamiltonian form of Noether conservation law}
In this section, it will be established that the right Euler-Poincaré equation~\refp{eq:DDWR} is the conservation law for left invariant Lagrangian problems (and conversely: left EP equation~\refp{DDWL} for right invariant Lagrangian). In the sequel, it is convenient to introduce the right invariant Maurer-Cartan  form $\mcfR$ dual to the right-invariant basis $\baseyR{A}$: $\mcfR^A(\baseyR{B})=\delta^A_B$. That is, if the point $\ptJE=(x^{\mu},\y{A},\xiR^A_\mu)$ of the 1-jet bundle $\JE$ is over the point $\ptE=(x^{\mu},\y{A})\in E$ then it exists a section $\boldsymbol{\sect}$ representative of that point such that $\xiR^A_\mu=\mcfR^A(\derp{\boldsymbol{\sect}}{x^{\mu}})$.
\subsubsection{Symmetric vector field}
If now, we consider a vector field of symmetry. To be more precise, if the Lagrangian density is invariant under the left action of a Lie group (left invariant Lagrangian), we consider a right invariant vector field \mbox{$S_{\eta}=\eta^A  \baseyR{A}$} with constant vector $\boldsymbol{\eta}$ in the Lie-algebra $\mathfrak{g}$. This vector field has, according to~(\ref{eq:Zv2}), an extension at the point $(x^\mu,y^A,\chi_\mu^A)\in{\JE}$
\begin{equation}\label{eq:J1Ze}
j^1S_{\eta}=\eta^A \baseyR{A}
-\crochetL{\boldsymbol{\chi}_\mu}{\boldsymbol{\eta}}^A
\basevR{\mu}{A}
\end{equation}
in the basis $\bpm \basex{\mu}&\baseyR{A}&\basevR{\mu}{A}\epm$ (note that the minus sign is due to the right invariant vector field and the formula~(\ref{eq:Zv2}) is for left invariance).
The reduced Lagrangian is expressed using the right invariant basis as
\begin{equation}\label{eq:lagR}
\lag'=\ldR (x^{\mu}, g^A, \chi^A_\mu)\vol,\quad \ldR \in \C^\infty (\JE),\; \vol \in \Lambda^{n+1} (M).
\end{equation}
By symmetry its Lie derivative  vanishes: $\Lie{j^1S_{\eta}}\lag'=0$, that is $j^1S_{\eta}\contr \dd \lag' +\dd (j^1S_{\eta}\contr \ldR\vol)=0$. This means $\dd \lag' (j^1S_{\eta})=0$, since $j^1S$ is a $\pi$-vertical vector field (no component on $\basex{\mu}$). 
\subsubsection{Conservation law}
So, using the symmetry $\dd \lag' (j^1S_{\eta})=0$ which is also $d\ldR (j^1S_{\eta})=0$ and with change of basis, we have
\begin{eqnarray*}
0&=&d\ldR (j^1S_{\eta})=\left(
\derp{\ldR}{x_\mu}dx^\mu+
\derp{\ldR}{y^A}\dy{A}+
\derp{\ldR}{\chi^A_\mu}d\chi^A_\mu
\right)(j^1S_{\eta})\\
&=&
\derp{\ldR}{y^A}\dy{A}(j^1S_{\eta})+
\derp{\ldR}{\chi^A_\mu}d\chi^A_\mu(j^1S_{\eta})\\
%&=&
%\derp{\ldR}{y^A}\tilde{T}^{A}_{B} \mcfR^B(j^1S_{\eta})+
%\derp{\ldR}{\chi^A_\mu}d\chi^A_\mu(j^1S_{\eta})\\
&=&
\derp{\ldR}{y^A}\tilde{T}^{A}_{B} \mcfR^B(j^1S_{\eta})-
\derp{\ldR}{\chi^A_\mu}\crochetL{\boldsymbol{\chi}_\mu}{\boldsymbol{\eta}}^A
,\quad\text{by lift~(\ref{eq:J1Ze})}\\
&=&
\derp{\ldR}{y^A}\tilde{T}^{A}_{B}\eta^B -
\piR_{\mu}^A\crochetL{\boldsymbol{\chi}_\mu}{\boldsymbol{\eta}}^A,\quad\text{by Legendre transf.}\\
&=&
\left(
-\derp{\hdR}{y^B}\tilde{T}^{B}_{A}-(\ade{\boldsymbol{\chi}_\mu}{\boldsymbol{\piR}^\mu})^A
\right)\eta^A,\quad
\text{since $\derp{\ldR}{y^B}=-\derp{\hdR}{y^B}$}\\
\end{eqnarray*}
for all $\boldsymbol{\eta}$, that is $(\ade{\boldsymbol{\chi}_\mu}{\boldsymbol{\piR}^\mu})^A=- \derp{\hdR}{y^B}\tilde{T}^{B}_{A}$. So, these two terms annihilate each other in the de Donder equation~(\ref{eq:DDWR}). That gives the conservation law $\derp{\piR_{\mu}^A}{x_{\mu}}\bigg|_{\prolonge{\sect}}=0$ which was already obtained before in~(\ref{eq:NoetherBL}). It appears that, with left-invariant Lagrangians, the first Noether theorem~\ref{th:Noether2} can be formulated by the right formulation of the Hamilton-Poincaré equation of motion~(\ref{eq:DDWR}).
\vspace{-0.5cm}
\subsubsection{Noether's current}
Two possibilities are available using the left or right representation of the dual Lie-algebra $\mathfrak{g}^*$.
\paragraph*{Left}
The Noether's current is computed by contracting the Hamiltonian version~(\ref{eq:G_PCH}) of the Poincaré-Cartan form with a left expression of the symmetric vector field
\begin{equation*}
j^1S_{\eta}=\ad{g^{-1}} {{\eta}}^A\baseyL{A}
+\gamma^A_{\mu}\basevL{\mu}{A}
\end{equation*}
for some $\b\gamma$. It gives
\begin{eqnarray*}
j^1S_{\eta}\contr \Theta_\ha&=&\left(\momL{\mu}{A}\mcfL^A\wedge d^{n}x_{\mu}-\hd\vol\right)(j^1S_{\eta})\\
&=&\momL{\mu}{A}(\ad{\b g^{-1}} {\boldsymbol{\eta}})^A d^{n}x_{\mu}
=\dual{\boldsymbol{\momL{\mu}{}}}{\ad{\b g^{-1}} {\boldsymbol{\eta}}} d^{n}x_{\mu}\\
J_L^*\,_{\eta}&=&\dual{\ade{\b g^{-1}} {\boldsymbol{\pi}^\mu}}{\boldsymbol{\eta}} d^{n}x_{\mu}=
\dual{J_L^*}{\boldsymbol{\eta}}.
\end{eqnarray*}
That is
\begin{equation}\label{eq:G_NCL}
J_L^*=\ade{\b g^{-1}} {\boldsymbol{\pi}^\mu} d^{n}x_{\mu}.
\end{equation}
\paragraph*{Right}
The right version is obtained by contracting the Hamiltonian Poincaré-Cartan form (right version $\Theta_{\haR}=\momR{\mu}{A}\mcfR^A\wedge d^{n}x_{\mu}-\hdR\vol$) with a right expression of the symmetric vector field~(\ref{eq:J1Ze}). It gives
\begin{eqnarray*}
j^1S_{\eta}\contr \Theta_{\haR}
&=&\left(\momR{\mu}{A}\mcfR^A\wedge d^{n}x_{\mu}-\hdR\vol\right)(j^1S_{\eta})
\\
&=&\piR_{\mu}^A{\eta}^A d^{n}x_{\mu}\\
J_R^*\,_{\eta}&=&\dual{\boldsymbol{\piR}^\mu}{\boldsymbol{\eta}} d^{n}x_{\mu}=
\dual{J_R^*}{\boldsymbol{\eta}}.
\end{eqnarray*}
Note that this result was already obtained in~(\ref{eq:NoetherCurrent}), \ie
\begin{equation}\label{eq:G_NCR}
J_R^*=\boldsymbol{\piR}^\mu d^{n}x_{\mu}.
\end{equation}

\begin{table}[!h]
\begin{tabular}{lll}
\hline \hline 
\multirow{4}{*}{\bf Lagrangian} 
&Euler-Lagrange&{ $\dfrac{\dd }{\dd t}\left(\dfrac{\partial  \Ld}{\partial \dot { q}}\right)  -  \dfrac{\partial  \Ld}{\partial q}=0$} \\
&Contact form&{ $\theta=\dd  q -  \dot{ q} \dd t$} \\
&Poincaré-Cartan&{ $\Theta_{\Ld}=\dfrac{\partial  \Ld }{\partial \dot{ q}}\dd q- \left(\dfrac{\partial  \Ld }{\partial \dot{ q}}\dot{ q}-\Ld  \right)\dd t$} \\
&Symplectic&{ $\Omega=-\dd \Theta_{\Ld}$}\\
\hline
{Legendre}&
$
\downarrow
\begin{cases}
p = \derp{ \Ld}{\dot q}\\
\Ha=\derp{\Ld}{\dot q}\dot q - \Ld 
\end{cases}
$
&
$
\uparrow
\begin{cases}
\dot q = \derp{ \Ha}{p}\\
 \Ld=\derp{ \Ha}{p}p - \Ha  
\end{cases}
$
\\
\hline
\multirow{4}{*}{\bf Hamiltonian} 
&Hamilt. eqs&{$ 
\begin{cases}
	\dd \Ha( X) = \derp{\Ha}{t}\\
	\dot{p} = -\derp{  \Ha}{ q}\\
	\dot { q} = \derp{\Ha}{p} 
\end{cases}$} \\
&Contact form&{ $\theta=\dd  q -  \derp{\Ha}{t} \dd t$} \\
&Poincaré-Cartan&{ $\Theta_{\Ha}=p\dd q- \Ha \dd t$} \\
&Symplectic&{ $\Omega=\dd q\wedge \dd p +\dd \Ha\wedge \dd t$}  \\
\hline
Variation Th.&&{ $\Lie{X} \PS=0$, $\forall X$ tangent to trajectories}  \\
\hline
\multirow{2}{*}{\bf Poisson} 
&Canonical form&{ $\Ocanon=\dd q\wedge \dd p$}\\
&Hamilt. vector fields&{ $X\contr\Ocanon=\dd \Ha$} \\
\hline\hline
\end{tabular}
\caption{\label{tab:S1}Principal results of the Cartan's lesson: symplectic case}
\end{table}

%%%%%%%%%%%%%%%%%%%%%%%%%%%%%%%%%%%%%%%%%%%%%%%%
%
\begin{table}[!h]
\begin{tabular}{lll}
\hline \hline 
\multirow{4}{*}{\bf Lagrangian} 
&Euler-Lagrange&{ 
$
\derp{}{x_\mu}\derp{\Ld}{v^A_{\mu}}\bigg|_{\prolonge{\sect}}-
\derp{\Ld}{y^A}\bigg|_{\prolonge{\sect}}=0,\quad A=1,\ldots,N.
$} \\
&Contact form&{ $\fcv_{\ptJE}= (\dy{A}-\v{\mu}{A} \dx{\mu}) \otimes \basey{A}$} \\
&Poincaré-Cartan&{ $\Theta_{\Lag}=\derp{\Ld}{v^A_{\mu}}\dy{A}\wedge d^{n}x_{\mu}-\left(\derp{\Ld}{v^A_{\mu}}v^A_{\mu}-\Ld\right)\vol$} \\
&Multi-symplectic&{ $\PSL=-\dd \Theta_{\Ld}$}\\
\hline
{Legendre}&
$
\downarrow
\begin{cases}
p^{\mu}_A=\derp{\Ld}{v^A_{\mu}}\\
\Hd=\derp{\Ld}{v^A_{\mu}}v^A_{\mu}-\Ld
\end{cases}
$
&
$
\uparrow
\begin{cases}
v^A_{\mu} = \derp{ \Hd}{p^{\mu}_A}\\
\Ld=\derp{ \Hd}{p^{\mu}_A}p^{\mu}_A - \Hd  
\end{cases}
$
\\
\hline
\multirow{4}{*}{\bf Hamiltonian} 
&de Donder Weyl&{$ 
\begin{cases}
\pullbackd{\sect}\left(d\Hd \wedge d^{n}x_{\nu}- 
  \derp{\Hd}{x_\nu}\vol-\Psymp{\mu} \wedge d^{n-1}x_{\mu\nu}\right)=0& (a)\\
\pullbackd{\sect}\left(dp^{\mu}_A  \wedge d^{n}x_{\mu} +\derp{\Hd}{y^A}\vol \right)=0&(b)\\
\pullbackd{\sect}\left(  dy^A\wedge d^{n}x_{\mu} - \derp{\Hd}{p^{\mu}_A}\vol \right)=0&(c)
\end{cases}
$} \\
&de DW forms&{
$
\begin{cases}
{D\Hd}=d\Hd-\derp{\Hd}{x_\alpha}dx^{\alpha}-\Psymp{\mu}(X_{\mu},X_{\alpha})dx^{\alpha}&(a)\\
\mathcal{T}^{\mu}_A = dp^{\mu}_A +\frac{1}{n+1} \derp{\Hd}{y^A}dx^{\mu}&(b)\\
\theta_\Ha^A=dy^A - \derp{\Hd}{p^{\mu}_A}dx^\mu&(c)
\end{cases}
$} \\
&Contact form&{ $\fcv_{\ptJE}=(dy^A - \derp{\Hd}{p^{\mu}_A}dx_\mu)\otimes \basey{A}$} \\
&Poincaré-Cartan&{ $\Theta_\Ha=p^{\mu}_A dy^A\wedge d^{n}x_{\mu}-\Hd\vol$} \\
&Multi-symplectic&{ $\Omega_\Ha=dy^A\wedge dp^{\mu}_A \wedge d^{n}x_{\mu}+d\Hd\wedge \vol$}  \\
\hline
Variation Th.&&{ $\pullback{\sect}(W\contr \Omega_{\Lag}) =0$, $\forall W\in \Chp{\JE}$}  \\
\hline
\multirow{2}{*}{\bf Poisson} 
&Canonical form&{ $\Ocan=\Psymp{\mu}\wedge d^{n}x_{\mu}=dy^A\wedge dp^{\mu}_A\wedge d^{n}x_{\mu}$}\\
&Ham. multi vector&{ $\mvec{X}\contr\Ocan=(-1)^{n} D\Hd$} \\
\hline\hline
\end{tabular}
\caption{\label{tab:S2}Principal results in the Multi-symplectic case with jet bundle $\JE$}
\end{table}

%%%%%
\begin{table}[!h]%[H!]%[!ht]
\begin{tabular}{lll}
\hline \hline 
\multirow{4}{*}{\bf Lagrangian} 
&Euler-Poincaré&{ $\derp{}{x_\mu}\derp{\ld}{\xiL^A_{\mu}}\bigg|_{\prolonge{\sect}}
-\left(\ade{\boldsymbol{\xiL}_{\mu}}{\derp{\ld}{\boldsymbol{\xiL}_{\mu}}}\right)^A
\bigg|_{\prolonge{\sect}}-
T^{B}_{A}\derp{\ld}{y^B}\bigg|_{\prolonge{\sect}}=0$} \\
&Contact form&{ $\left.\b\fcvG\right|_{\ptJE}= (\mcfL^A-\xiL^A_{\mu} d x_{\mu}) \otimes \baseyTG{A}$} \\
&Poincaré-Cartan&{ $\Theta_{\lag}=\derp{\ld}{\xiL^A_{\mu}}\mcfL^A\wedge d^{n}x_{\mu}-
\left(\derp{\ld}{\xiL^A_{\mu}}\xiL^A_{\mu}-
\ld\right)\vol$} \\
&Multi-symplectic&{ $ \PSl=-d \PCl$}\\
\hline
{Legendre}&
$
\downarrow
\begin{cases}
\momL{\mu}{A}=\derp{\ld}{\xiL^A_{\mu}}\\
\hd=\derp{\ld}{\xiL^A_{\mu}}\xiL^A_{\mu}-\ld
\end{cases}
$
&
$
\uparrow
\begin{cases}
\xiL^A_{\mu} = \derp{ \hd}{\momL{\mu}{A}}\\
\ld=\derp{ \hd}{\momL{\mu}{A}}\momL{\mu}{A} - \hd  
\end{cases}
$
\\
\hline
\multirow{4}{*}{\bf Hamiltonian} 
&de Donder Weyl&{$ 
\begin{cases}
\prld{\sect}^*
\left[
 d\hd \wedge d^{n}x_{\nu}
-\derp{\hd}{x_\nu}\vol
 -\psymp{\mu}\wedge d^{n-1}x_{\mu\nu}
 \right]=0&(a)\\
\prld{\sect}^*
\left[
\left(
d\momL{\mu}{A}  - 
(\ade{\mcfLv}{\boldsymbol{\pi}^\mu})_A 
 +\frac{1}{n+1}T^B_A\derp{\hd}{y^B}\dx{\mu}
 \right)
 \wedge d^{n}x_{\mu}
\right]=0&(b)\\
\prld{\sect}^*
\left[
\left(
\mcfL^A  -\derp{\hd}{\momL{\nu}{A}}\dx{\nu}
\right)
\wedge d^{n}x_{\mu}
 \right]=0&(c)
\end{cases}
$} \\
&de DW forms&{
$
\begin{cases}
D\hd=d\hd
-\derp{\hd}{x_\alpha}dx^{\alpha}
 -\psymp{\mu}(X_{\nu},X_{\alpha})dx^{\alpha}&(a)\\
\tau^{\mu}_A = 
d\momL{\mu}{A}  - 
(\ade{\mcfLv}{\boldsymbol{\pi}^\mu})_A 
 +\frac{1}{n+1}T^B_A\derp{\hd}{y^B}\dx{\mu}&(b)\\
\fcG{A}_\ha=\mcfL^A  -\derp{\hd}{\momL{\nu}{A}}\dx{\nu}&(c)
\end{cases}
$} \\
&Contact form&{ $\fcvG_{\ptJE}=(\mcfL^A  -\derp{\hd}{\momL{\nu}{A}}\dx{\nu})\otimes \basey{A}$} \\
&Poincaré-Cartan&{ $\Theta_\ha=\momL{\mu}{A}\mcfL^A\wedge d^{n}x_{\mu}-\hd\vol$} \\
&Multi-symplectic&{ $\Omega_\ha=\left(
\mcfL^A\wedge d\momL{\mu}{A}+\momL{\mu}{A} \crochetL{\mcfLv}{\mcfLv}^A\right)
\wedge d^{n}x_{\mu}
+d\hd\wedge \vol$}  \\
\hline
Variation Th.&&{ $\pullbackd{\sect}(W\contr \Omega_\ha) =0$, $\forall W\in \Chp{\JE}$}  \\
\hline
\multirow{2}{*}{\bf Poisson} 
&Canonical form&{ $\ps=\psymp{\mu}\wedge d^{n}x_{\mu}=\left(\mcfL^A\wedge d\momL{\mu}{A}+\momL{\mu}{A} \crochetL{\mcfLv}{\mcfLv}^A\right)\wedge d^{n}x_{\mu}$}\\
&Ham. multi vector&{ $\mvec{X}\contr\ps=(-1)^{n} D\hd$} \\
\hline\hline
\end{tabular}
\caption{\label{tab:S3}Principal results in the Multi-symplectic case with principal G bundle}
\end{table}

\clearpage
%%%%%%%%%%%%%%%%%%%%%%%%%%%%%%%%%%%%%%%%%%%%%%
\section{Conclusion}\label{sec:conclusion}

The main interest of this study is to obtain a way to formulate a field theory adapted to wave propagation including nonlinearities and symmetries. The covariant formulation is an alternative to avoid infinite dimensional configuration spaces giving to time and space the same signification. In modern geometry, this is characterized by bundles with multi-dimensional bases (time and space variables). The symmetric aspect is taken into account using Lie groups as configuration spaces.

Those two considerations may bring up the level of difficulties of the theory, thus Cartan's ideas are used as a guide to obtained all the needed elements from the contact form to the conserved quantities through Hamiltonian vector fields and equations of motion.

This process is done in two steps: firstly, the Cartan's lesson is extended from one independent variable to several ones - from symplectic to multi-symplectic; secondly, a Lie group action is considered leading to a reduction procedure. Three tables (\ref{tab:S1}, \ref{tab:S2} and \ref{tab:S3}) summarize the main results in order to make it possible a comparison between these apparently different formalisms, emphasizing the conceptual connections between them.
%%%%%%%%%%%%%%%%%%%%%%%%%%%%%%%%%%%%%%%%%%%%%%%%
%

\begin{acknowledgments}
We cannot express enough thanks to Jean-Pierre Marco for his continued support and encouragement and we are very grateful to Laurent Lazzarini for his comments and corrections on an earlier version of the manuscript. We also thank renowned persons Frédéric Hélein and Dominique Chevallier for sharing their pearls of wisdom with us during the course of this research.
\end{acknowledgments}

\appendix 

\addcontentsline{toc}{part}{Appendices}

\section{Geometric elements for the Cartan's lesson}\label{app:00}

Let consider a variation of a curve $q(t)$ given by a mapping $q(\Eps,t)$ as
illustrated in~fig.(\ref{fig:VarHam}). Let $Z$ be the vector field along the
direction of variation. Its pushforward $\tilde{Z}=(q^{-1})_*Z$ may be written as
any vector field $\tilde{Z}=z_1\derp{}{\Eps}+z_2 \derp{}{t}$ in the basis
$(\derp{}{\Eps}, \derp{}{t})$ for some components. Upon using the duality property of the base vector 
$\dd\Eps\left(\derp{}{\Eps}\right) = 1, \; \dd t\left(\derp{}{t}\right) = 1$ and $\dd
\Eps\left(\derp{}{t}\right) = \dd t\left( \derp{}{\Eps}\right) = 0$, these components, written
as
\begin{eqnarray*}
\dd \Eps (\tilde{Z})&=&\dd\Eps\left(z_1\derp{}{\Eps}+z_2 \derp{}{t}\right) = z_1\\
\dd t (\tilde{Z})&=&\dd t\left(z_1\derp{}{\Eps}+z_2 \derp{}{t}\right) = z_2,
\end{eqnarray*}
rely on an arbitrarily time parametrization.
\begin{figure}[!ht]
\begin{center}
\begin{tikzpicture}
%nodes
%A
\node[inner sep=0pt] (A) at (-5,0)
    {\includegraphics[width=.2\textwidth]{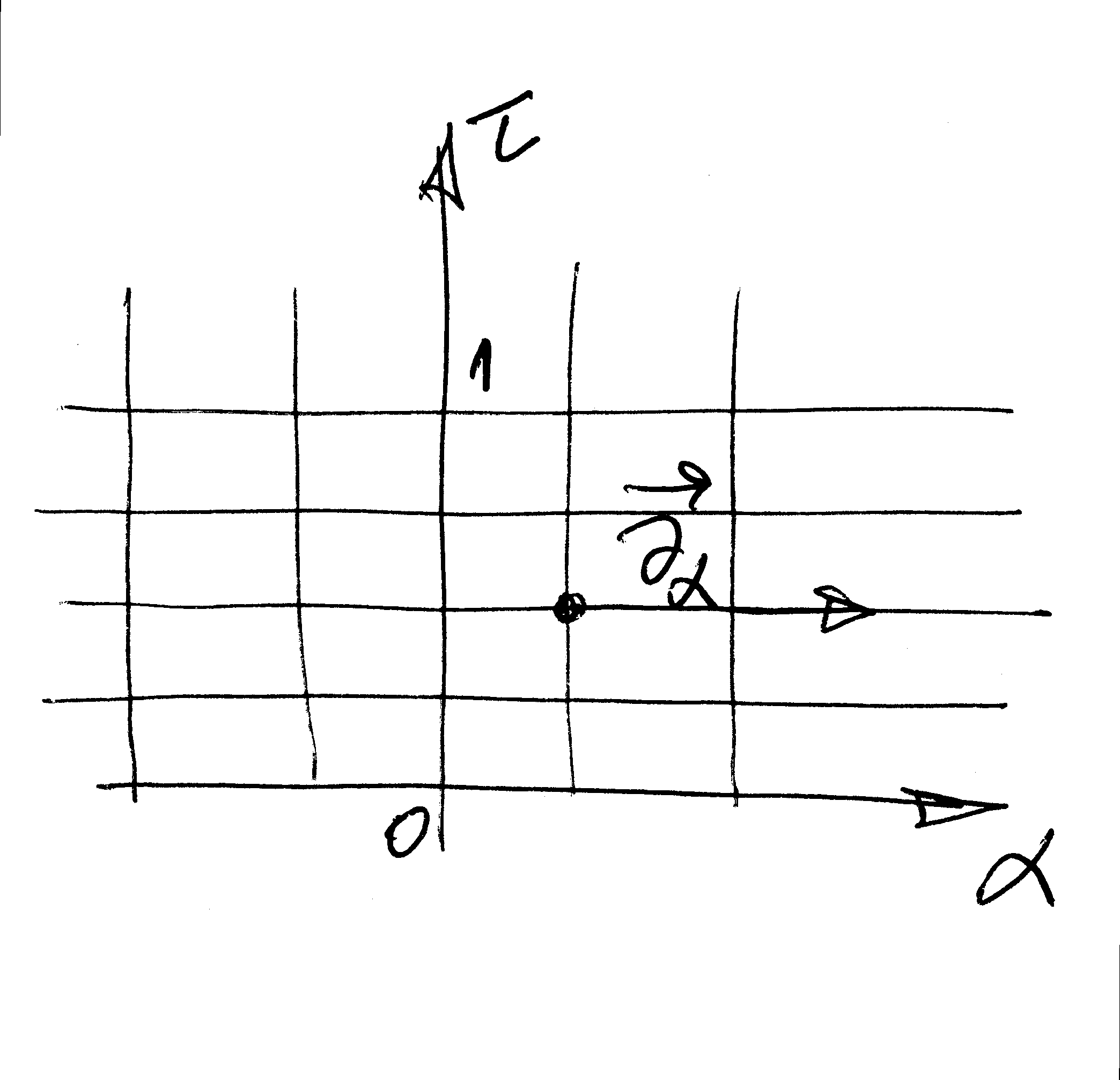}}; 
%B
\node[inner sep=0pt] (B) at (5,0)
    {\includegraphics[width=.3\textwidth]{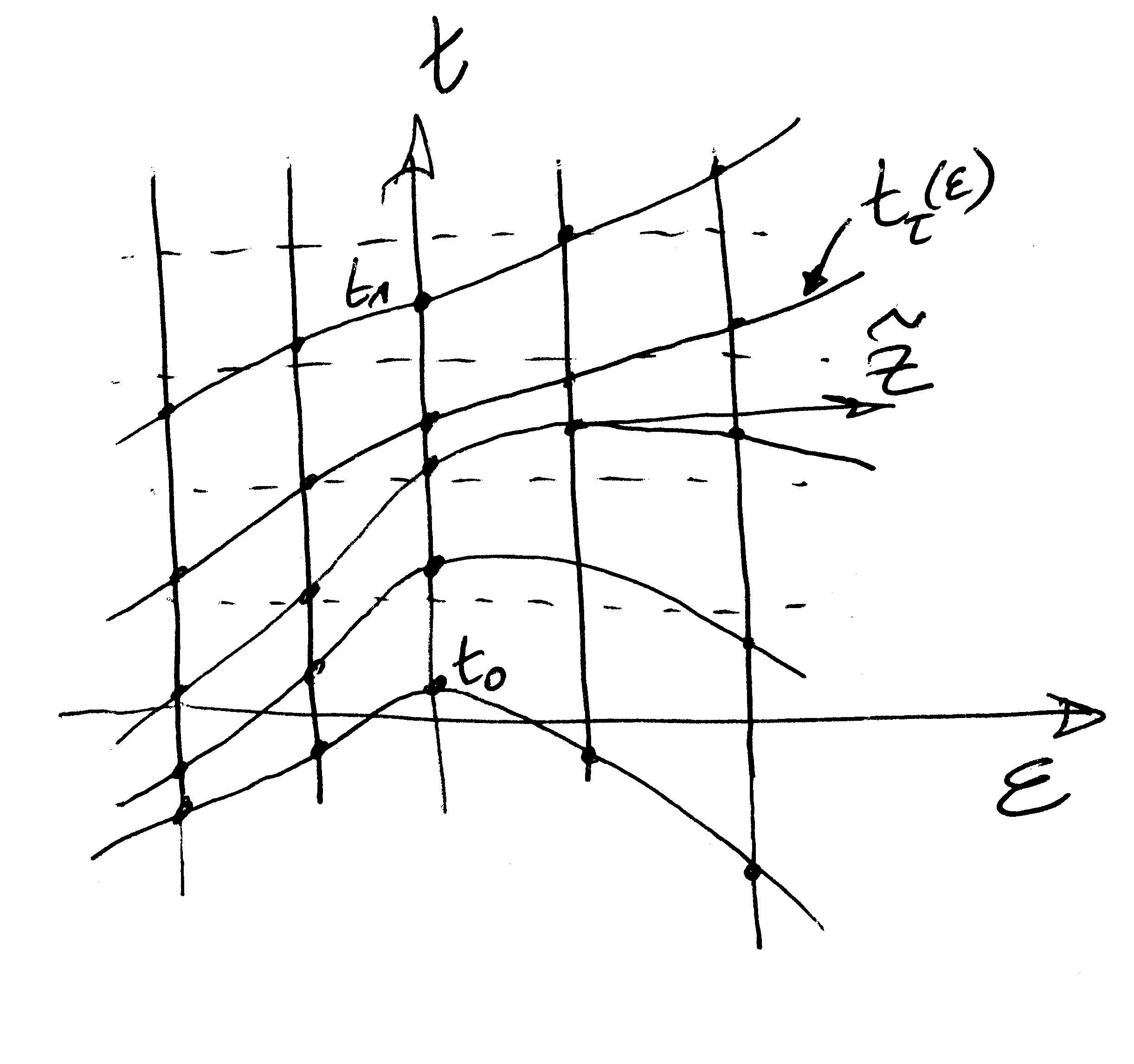}};
%C
\node[inner sep=0pt] (C) at (-2,-6)
    {\includegraphics[width=.6\textwidth]{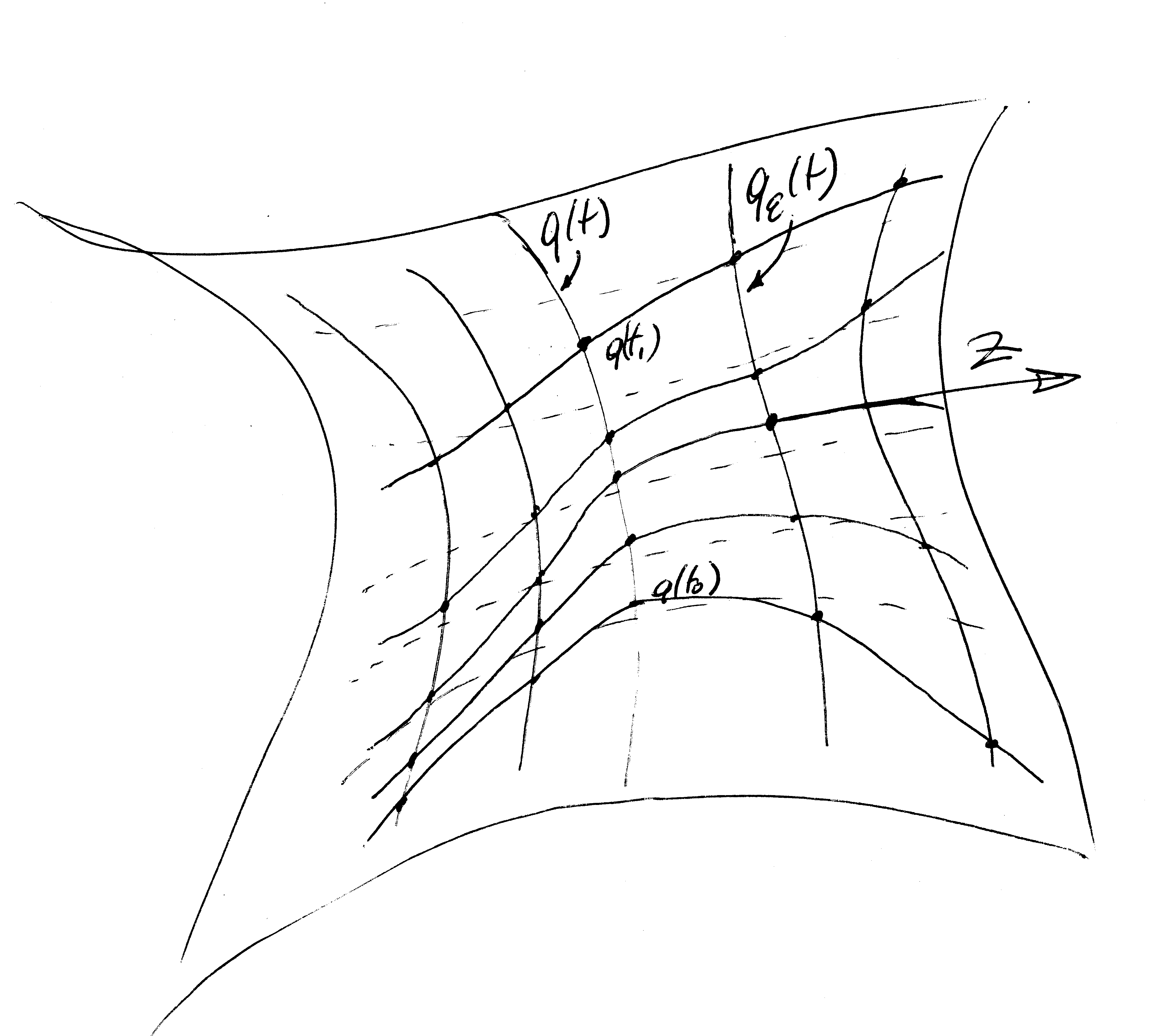}};
%D
\node[inner sep=0pt] (E) at (4,-3.5){$q(\epsilon,t)$};
\node[inner sep=0pt] (E) at (-4.5,-9){$\mathcal{Q}$};
%arrows
\draw[->,thick] (A) to[bend left] (B)
node[midway]{$\mathcal{T}(\alpha,\tau)$} ;
\draw[->,thick] (B.south) to[bend left] (C.east);

\end{tikzpicture}%
\caption{On top figures: Each virtual trajectory is parametrized by a vertical
 line in the "paramater spaces" $(\tau,\alpha)$ or $(t, \Eps)$. On bottom: Varied trajectories on the configuration space $\mathcal{Q}$.
 If an arbitrary time parametrization $t=t(\alpha,\tau)=t_\tau(\epsilon)$ is allowed, the vector field $Z$ associated to variations is no longer aligned along the line of coordinates $\Eps$ and we have $Z=\derp{q}{\Eps} + \dot{ q}\der{t_\tau}{\Eps}$.
 }
\label{fig:VarHam}
\end{center}
\end{figure}
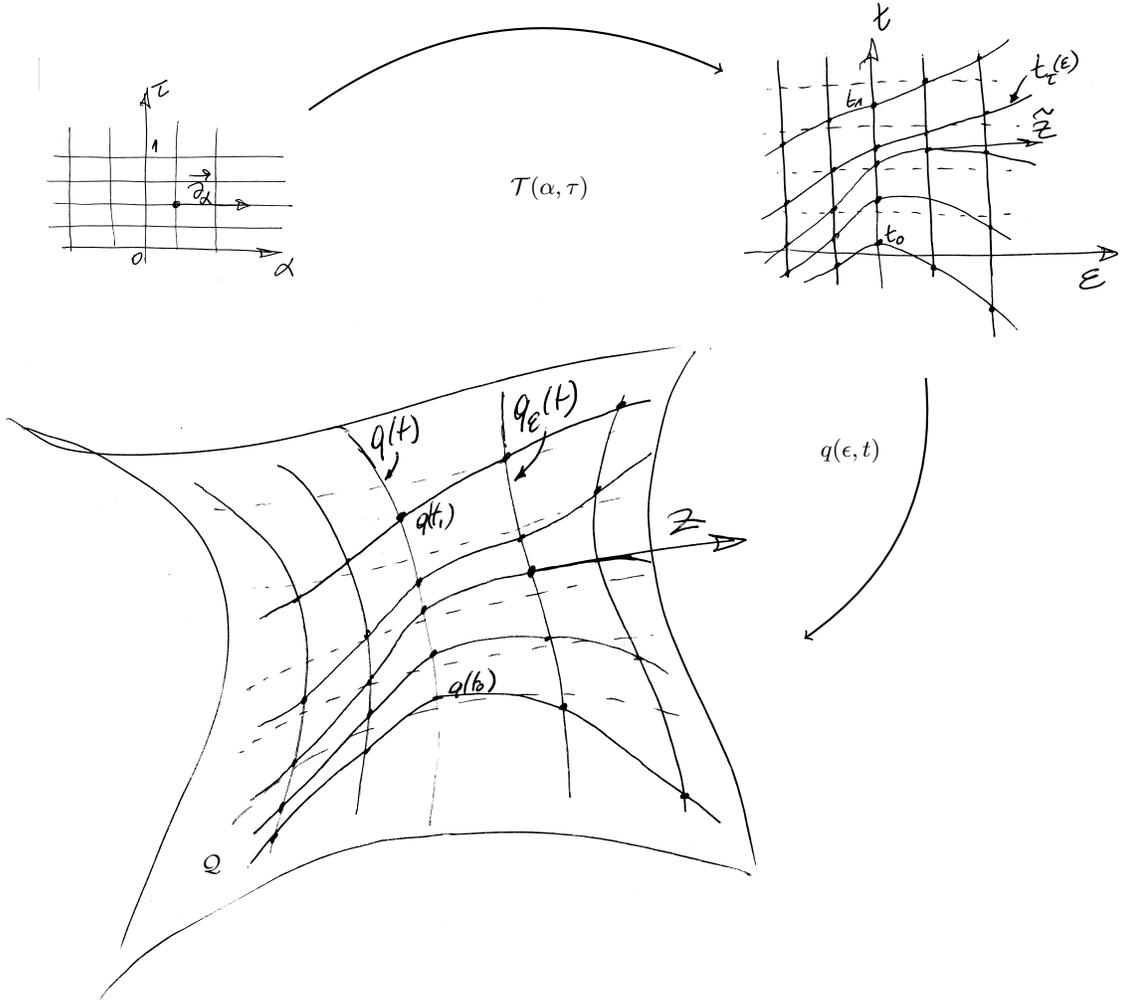

To see that, let us consider, without loss of generality, that a normalized vector
field verifying $\dd \Eps(\tilde{Z})=1$ may be obtained by choosing a mapping
\begin{eqnarray*}
        \mathcal{T}:U_1&\mapsto& U_2\\
        (\alpha,\tau)&\mapsto&(\epsilon,t)
\end{eqnarray*}
such that
 \begin{equation*}
                \mathcal{T}(\alpha,\tau)=
        \begin{cases}
                \epsilon=\alpha\\
                t=t(\alpha,\tau)=t_\tau(\epsilon)
        \end{cases}
\end{equation*}
for any time parametrization $t=t(\alpha,\tau)=t_\tau(\epsilon)$ (see
figure~fig.(\ref{fig:VarHam})). Since $\tilde{Z}$ is generated by
$\vec{\partial}_\alpha$ throught the differential $\dd \mathcal{T}$ by
 \begin{equation*}
      \tilde{Z}=          \dd \mathcal{T}(\vec{\partial}_\alpha)=
        \begin{cases}
                \dd \epsilon(\vec{\partial}_\alpha)=\dd \alpha(\vec{\partial}_\alpha)=1\\
                \dd t(\vec{\partial}_\alpha)= \left(\derp{t}{\alpha} \dd{\alpha} + \derp{
t}{\tau}\dd{\tau}\right)(\vec{\partial}_\alpha)=\der{t_\tau(\epsilon)}{\epsilon},
        \end{cases}
\end{equation*}
it appears clearly that the second component $\dd t(\tilde{Z})=\der{t_\tau}{\Eps}$ is related to the time parametrization $t_\tau(\epsilon)$. So finally, we have $\tilde{Z}= \derp{}{\Eps} +
\der{t_\tau}{\Eps}\derp{}{t}$. Evaluate the differential $\dd q =  \derp{q}{\Eps}
\dd{\Eps} + \derp{ q}{t}\dd{t}$ along this direction leads then to a general
expression of the vector field in the configuration space
 \begin{eqnarray*}
Z=\dd q (\tilde{Z}) &=&  \derp{q}{\Eps} \dd{\Eps}\left( \derp{}{\Eps} +
\der{t_\tau}{\Eps}\derp{}{t}\right)
 + \derp{ q}{t}\dd{t}\left( \derp{}{\Eps} + \der{t_\tau}{\Eps}\derp{}{t}\right)\\
 &=&  \derp{q}{\Eps} + \derp{ q}{t}\der{t_\tau}{\Eps}=    \derp{q}{\Eps} + \dot{ q}
\dd t( \tilde{Z})\\
& \Rightarrow&\quad\derp{q}{\Eps}=\dd q (\tilde{Z})- \dot{ q}\dd t( \tilde{Z})
\end{eqnarray*}
In the Cartan's lesson, this is written, by abuse of notation
\begin{equation}\label{eq:ContactCartan_app}
\left. \dfrac{\partial  q_\Eps}{\partial \Eps}\right|_{\Eps=0} = \dd  q( Z) -  \dot{
q} \dd t( Z),
\end{equation}

%----%
%s\section{Elements of differential geometry for first-order Lagrangian field theories}

\section{Jet prolongation of vector fields }\label{app:jet-p}
There are two ways to obtain the one-jet prolongation of vector fields given by \refp{anx:liftVF:geom}
\begin{itemize}
\item using the invariance of the contact module,
\item saying that a contact transformation preserves the holonomicity.
\end{itemize}

\subsection{Preservation of the contact module}\label{sect:jet-pContactInv}

\begin{definition}[Contact transformation]
A local diffeomorphism $\mathcal{T} : \JE \rightarrow \JE$ defines a contact transformation if it preserves the contact ideal, meaning that if $\sigma$ is any contact form on $\JE$, then $\mathcal{T}^*\sigma$ is also a contact form.
\end{definition}
\begin{proposition}[]
 The flow generated by a vector field $j^1Z$ on the jet space $\JE$ forms a one-parameter group of contact transformations if and only if the Lie derivative $\Lie{J^1Z}(\sigma)$ of any contact form $\sigma$ preserves the contact ideal or module.
\end{proposition}
So, starting from a general vector field $Z=\alpha^{\mu}\basex{\mu}+\beta^A\basey{A}$ where $\alpha^{\mu}$ and $\beta^A$depend on $(\x{\mu}$, $\y{A})$, and writing the jet prolongation $j^1Z$ on the jet space $\JE$ as
\begin{equation*}
j^1Z=\alpha^{\mu}\basex{\mu}+\beta^A\basey{A}+\gamma^A_{\mu}\basev{A}{\mu},
\end{equation*}
the only problem is to calculate the coefficients $\gamma^A_{\mu}$. 

Echeverria \& al~\cite{Echeverria00} traduce the preservation of the contact module using only the components $\fc{A}= \dy{A}-\v{\nu}{A} \dx{\nu}$ of the contact form \refp{anx:eqFC} by the fact that they must have 
\begin{equation}\label{preservContactModEq}
\Lie{J^1Z}(\fc{A})=\zeta^A_B \fc{B}, \quad \zeta^A_B\in \C^\infty(\JE).
\end{equation}
It is then straightforward, using the Cartan formula in the Lie derivative $\Lie{J^1Z}$ and identifying the different terms, to obtain the 1-jet prolongation of $Z$ given by proposition~\ref{anx:liftVF:geom}. However, the formulation in components~\refp{preservContactModEq} doesn't appear to be clear for us. Thus, in the next section we propose a general and geometric definition for the lift of vector fields.

\subsection{Jet prolongation- geometrical definition}\label{sect:jet-pGeom}
This section rely on another way to compute the one-jet prolongation of vector fields, defining
\begin{definition}[The contact transformation]
A local diffeomorphism $\mathcal{T} : \JE \rightarrow \JE$ defines a contact transformation if it maps holonomic sections to holonomic ones.
\end{definition}
So, starting from an arbitrary vector field $Z\in \Chp{E}$, the problem is to extend it to $j^1Z\in \Chp{\JE}$. To do that, we consider the transformation of a holonomic section, $j^1\sect$, with tangent vector $\bar X$, that derives from a section $\sect$ of $\pi$ with tangent vector $X$. The idea is then to obtain a equation for $j^1Z$ specifying that the transformed section is also holonomic by checking the holonomic criteria $\fcv\big|_{\ptJE}(\bar X) =0$ (proposition \ref{prop:holsec2}).

Nonetheless, this criteria uses elements of the one-jet bundle tangent space, $\bar X\in \Chp{\JE}$, while the given data $Z$ of our problem  belongs to $TE$. Even if this criteria may also be projected in $TE$, it involves a inconvenient differential $\dd j^1\sect$ since, using the tangent vector $X\in T_{\ptE}E$ to the section $\sect$, we can adapt the holonomic criteria to
\begin{equation}\label{anx:condHol}
\fcv\big|_{\ptJE}(\bar X) = \fcv  \Big|_{\ptJE}\left(\dd j^1\sect \circ \dd\pi (X) \right) = 0.
\end{equation}
However, this can be simplified by noting that the contact form, $\fcv_{\ptJE}= (\dy{A}-\v{\mu}{A} \dx{\mu}) \otimes \basey{A}$, does not involve the  form $\dv{\mu}{A}$. Its evaluation along $\bar X$ or along the natural extension at point $\ptJE$ of $X=\dd\pi^1(\bar X)$ given by $\tilde{X} = (X,0)$ gives the same result
\begin{equation*}
\fcv\big|_{\ptJE}(\bar X)=\fcv\big|_{\ptJE}(\tilde{X} )=\fcv\big|_{\ptJE}((X,0)).
\end{equation*}

\paragraph*{Remark}
To be more precise, this property of the contact form $\fcv$ may be expressed more intrinsically  by
\begin{equation*}
\fcv\big|_{\ptJE}=\pi_{\b v}^*\fcv\big|_{\ptJE},\quad \pi_{\b v}=\varphi_{\b v}\circ\pi^1
\end{equation*}
where $\varphi_{\b v}$ is a constant section of $\pi^1$ given by
\begin{eqnarray*}
\varphi_{\b v} : E &\rightarrow & \JE\\
	\ptE&\mapsto& \varphi_{\b v} (\ptE) =(\ptE,\b v)= \ptJE.
\end{eqnarray*}
That is to say that $\fcv$ is invariant by projection on specific hyper-planes of $\JE$.
\begin{proof}
On one hand, by definition, $\fcv\big|_{\ptJE}(\bar X ) = \dd^V_\ptE \sect (\dd\pi^1(\bar X)) = \dd^V_\ptE\sect (X)$. On the other hand
\begin{eqnarray*}
\pi_{\b v}^*\fcv\big|_{\ptJE}(\bar X )&=&
(\varphi_{\b v}\circ\pi^1)^*\fcv\big|_{\ptJE}(\bar X )=(\pi^1)^*(\varphi_{\b v})^*\fcv\big|_{\ptJE}(\bar X )\\
&=&(\varphi_{\b v})^*\fcv\big|_{\ptJE}(\dd \pi^1(\bar X ))=\fcv\big|_{\ptJE}(\dd \varphi_{\b v}(\dd \pi^1(\bar X )))\\
&=&\fcv\big|_{\ptJE}(\dd \varphi_{\b v} (X)),\quad \text{since }\dd \pi^1(\bar X )=X\\
&=& \dd^V_\ptE \sect (\dd\pi^1\circ \dd \varphi_{\b v}(X)) = \dd^V_\ptE \sect(X).\diamond
\end{eqnarray*}
\end{proof}
This property allows us to project the holonomic criteria (prop. \ref{prop:holsec2}) on $TE$ by defining 
\begin{definition}[The holonomic map]\label{def:holMap}
Let $\ptE = (\x{\mu}, \y{A})$ be a point of $E$ in a local system of coordinates, $X$ be vector of $T_\ptE E$ and $\fcv$ be the canonical contact form on $\JE$. The holonomic map is
\begin{eqnarray*}
\Hol : TE &\rightarrow& \JE\\
	(\ptE, X) &\mapsto& \ptJE = (\x{\mu}, \y{A}, \v{\mu}{A})
\end{eqnarray*}
with ${\b v}$ such that 
\begin{equation}\label{anx:condHol1}
\fcv  \Big|_{\ptJE}\left( (X, 0)\right) = 0\quad \text{(or $(\varphi_{\b v})^*\fcv\big|_{\ptJE}(X)=0$).}
\end{equation}
\end{definition}
And conversely,
\begin{definition}[The inverse holonomic map]\label{def:invholMap}
Let $\ptJE = (\x{\mu}, \y{A}, \v{\mu}{A})$ be a point of $\JE$ in a local system of coordinates. The inverse holonomic map is
\begin{eqnarray*}
\Hol^{-1} : \JE &\rightarrow& TE\\
	\ptJE &\mapsto& (\ptE, X)  
\end{eqnarray*}
with $\ptE = (\x{\mu}, \y{A})$ and $X$ a vector of $T_\ptE E$ such that 
\begin{equation}\label{anx:condHol2}
\fcv  \Big|_{\ptJE}\left( (X, 0)\right) = 0\quad \text{(or $(\varphi_{\b v})^*\fcv\big|_{\ptJE}(X)=0)$}.
\end{equation}
\end{definition}
All this allows us to obtain the proposition 
\begin{proposition}[]
If $X$ is a tangent vector to a section $\sect$ at point $\ptE\in E$ then the holonomic lifted section $j^1\sect$ goes through the point  $\ptJE = \Hol(\ptE,X)$ of $\JE$.
\end{proposition}

According to this, the one-jet prolongation, $j^1Z$, of $Z$ can now be defined using the differential of the holonomic map without referring to sections as
\begin{definition}[One-jet prolongation of vector fields]\label{def:JetProl}
Let $\ptJE \in \JE$ and $(\ptE,X)\in T_{\ptE}E$ such that $(\ptE,X)=\Hol^{-1}(\ptJE \in \JE)$. The one-jet prolongation, $j^1Z$ at point $\ptJE$ of a vector field $Z\in \Chp{E}$ is the differential of the holonomic map $\dd\Hol : TT_{(\ptE, X)}E  \rightarrow  T_{\ptJE}\JE$ evaluated in the $Z$ direction
\begin{equation*}
j^1Z=\dd\Hol(Z).
\end{equation*} 
\end{definition}
In practice the following algorithm is used
\begin{enumerate}
\item start from a given point $\ptJE\in \JE$ and a vector field $Z$ over $E$,
\item compute $(\ptE,X)=\Hol^{-1}(\ptJE)$. This induces a field of (hyper)planes 
characterized by $n+1$ tangent and normalized vector fields $X_{\mu}$, representative of the point $\ptJE$
\item consider an integral curve $\mathcal{C} : \R \supset I \rightarrow TE$ of the one-parameter transformation group  $\tau_{ \Eps}^Z$ along $Z$ given by $\Eps\mapsto (\ptE_\Eps,{X}_\mu^{\Eps})=(\tau^Z_\Eps(\ptE),T_{\ptE}{(\tau_{ \Eps}^Z)} (X_\mu))$, where the notation $T{f}(X)$ means the tangential map of $f$ evaluated in the direction $X$, also denoted $f_*X$ pushforward of $X$ by $f$.
\item use the Lie braket $\crochetL{X_\mu}{Z}$ to express $T_{\ptE}{(\tau_{ \Eps}^Z)} (X_\mu)$
\item determine the curve $\Eps \mapsto \ptJE_\Eps=\Hol(\ptE_\Eps,{X}_\mu^{\Eps})$ 
\item compute the differential at point $\ptJE$ by the infinitesimal procedure
 \begin{equation*}
j^1Z=\lim_{\Eps\rightarrow 0}\frac{\ptJE_\Eps-\ptJE}{\Eps}.
\end{equation*}
\end{enumerate}
Let us see this algorithm in details.
%%%%%%%%%%%%%%%%%%%%%%%%%%%%
\paragraph*{Step (i)} Let $\ptJE = (\x{\mu}, \y{A}, \v{\mu}{A})$ be a given point of $\JE$ in a local system of coordinates and let $Z\in \Chp{E}$ be a vector field over $E$ given by $Z=\alpha^{\mu}\basex{\mu}+\beta^A\basey{A}$.
%%%%%%%%%%%%%%%%%%%%%%%%%%%%
\paragraph*{Step (ii)} 
The inverse holonomic map \ref{def:invholMap} is used to compute the tangent vector fields representative of the point $\ptJE$ by  $(\ptE,X)=\Hol^{-1}(\ptJE)$, that is
\begin{equation}\label{holonomicXmu}
\b\fcv_{\ptJE}(({X},0))=0\Leftrightarrow \dy{A}\left({X}\right)=\v{\nu}{A}\dx{\nu}\left({X}\right).
\end{equation}
It gives  a field of (hyper)planes characterized by $n+1$ tangent vector fields. Without loss of generality, these vectors can be 
normalized to
\begin{equation}\label{eq:NX}
{X}_{\mu}=\basex{\mu}+\v{\mu}{B}\basey{B},\quad \dx{\nu}\left({X}_{\mu}\right)=\delta_\mu^\nu
,\quad ({\mu}=1,\ldots,n,0).
\end{equation}
Obviously ${X}_{\mu}$ verifies \refp{holonomicXmu}.
%%%%%%%%%%%%%%%%%%%%%%%%%%%%
\paragraph*{Step (iii)}
Let us consider an integral curve $\mathcal{C}$ of the one-parameter transformation group  $\tau_{ \Eps}^Z$ along $Z$  given by
\begin{eqnarray*}
\mathcal{C}: \R \supset I &\rightarrow& E\\
\Eps&\mapsto& (\ptE_\Eps,{X}_\mu^{\Eps})=(\tau^Z_\Eps(\ptE),T_{\ptE}{(\tau_{ \Eps}^Z)} (X_\mu)).
\end{eqnarray*}
%%%%%%%%%%%%%%%%%%%%%%%%%%%%
\paragraph*{Step (iv)}
This last procedure involves the bracket of vector fields. More precisely, the Lie derivative of the two vector fields $X_\mu$ and $Z$, at point $\ptE$, is given by (see~\cite{jost2005riemannian}, for example),
\begin{eqnarray*}
&&\Lie{Z}X_\mu=\crochetL{Z}{X_\mu}\Big|_\ptE=\der{}{ \Eps}\bigg|_{ \Eps=0}\left((\tau_\Eps^Z)^* X_\mu\big|_{\ptE_\Eps}\right)\\
&&=\der{}{ \Eps}\bigg|_{ \Eps=0}\left(T_{\ptE_\Eps}{(\tau_{- \Eps}^Z)} (X_\mu)\right)
=\lim_{ \Eps\rightarrow 0}\frac{T_{\ptE_\Eps}{(\tau_{- \Eps}^Z)} (X_\mu)-X_\mu\big|_{\ptE}}{ \Eps}
\end{eqnarray*}
where $\ptE_\Eps=\tau_{\Eps}^Z(\ptE)$. Evaluating this definition at point $\ptE_\Eps$, one obtains a finite expansion of the tangent map of the transformation $\tau_{ \Eps}^Z$
\begin{equation}\label{expXmu}
X_\mu^{ \Eps}\big|_{\ptE_\Eps}=T_{\ptE}{(\tau_{ \Eps}^Z)} (X_\mu)=X_\mu\big|_{\ptE_\Eps}+ \Eps\crochetL{X_\mu}{Z}\big|_{\ptE_\Eps}+\mathcal{O}( \Eps^2).
\end{equation}
%%%%%%%%%%%%%%%%%%%%%%%%%%%%
\paragraph*{Step (v)}
Using the definition \ref{def:holMap}, the curve $\Eps \mapsto \ptJE_\Eps=\Hol(\ptE_\Eps,{X}_\mu^{\Eps})$ is obtained by solving
\begin{equation}\label{holonomicXmuEps}
\b\fcv_{\ptJE_\Eps}(({X}_\mu^{\Eps},0))=0\Leftrightarrow \dy{A}\left({X}_\mu^{\Eps}\right)=(v_{\Eps})^A_{\mu}\dx{\nu}\left({X}_\mu^{\Eps}\right).
\end{equation}
%%%%%%%%%%%%%%%%%%%%%%%%%%%%
\paragraph*{Step (vi)}
Using the finite expansion \refp{expXmu} and tacking into account \refp{holonomicXmu} and \refp{eq:NX}, the holonomic criteria \refp{holonomicXmuEps} becomes 
\begin{equation}
\Big((v_\Eps)_\mu^A - \v{\mu}{A} \Big)= \Eps \, \fc{A}\big|_{\ptJE_\Eps}(\crochetL{X_\mu}{Z}) + \mathcal{O}(\Eps^2),
\end{equation}
where the contact form \refp{eq:contact} is used again. This expansion furnishes the third component of the jet-prolongation, $j^1Z$, at point $\ptJE=(\x{\mu},\y{A},\v{\mu}{A})$ 
\begin{equation}\label{eq:3C}
\gamma^A_{\mu}=\dv{\mu}{A}(j^1Z)=
\lim_{\Eps\rightarrow 0}\frac{(v_\Eps)_\mu^A - \v{\mu}{A}}{\Eps}=
\fc{A}\big|_{\ptJE}(\crochetL{ X_\mu}{Z}) = \fc{A}\big|_{\ptJE}(\Lie{ X_\mu}{Z}).
\end{equation}
This is another form of the proposition \ref{anx:liftVF:geom}. This can be seen by inserting into \refp{eq:3C} the following computation of the Lie bracket (where it is important to notice that the ${X}_{\mu}$ are defined for a fixed $\v{\mu}{B}$). For an arbitrary function $f$
\begin{eqnarray*}
\crochetL{ X_\mu}{Z}f &=&  X_\mu Z f - Z  X_\mu f 
=  X_\mu (\alpha_\nu \basex{\nu} + \beta^A \basey{A}) f - Z(\basex{\mu} + \v{\mu}{B} \basey{B}) f \\
&=& 
  (\basex{\mu} + \v{\mu}{B} \basey{B}) (\alpha_\nu \derp{f}{\x{\nu}} + \beta^A\derp{f}{\y{A}}) -  (\alpha_\nu \basex{\nu} + \beta^A \basey{A})  (\derp{f}{\x{\mu}}+ \v{\mu}{B}\derp{f}{\y{B}})\\
  &=& \left( \derp{\alpha_\nu}{\x{\mu}} + \v{\mu}{B}\derp{\alpha_\nu}{\y{B}}\right)\derp{f}{\x{\nu}} + \left(  \derp{\beta^A}{\x{\mu}}+ \v{\mu}{B} \derp{\beta^A}{\y{B}}\right) \derp{f}{\y{A}}.
\end{eqnarray*}

%%%%%%%%%%%%%%%%%%%%%%%%%%%%%%%%%%%%%%%%%%%%%

\section{The Hamilton principle: calculus of variations}\label{anx:deltaA}
This appendix is dedicated to obtain the Poincaré-Cartan form $\PCL$ from the Hamilton principle (definition \ref{def:HP}). In other words, to establish equation \refp{eq:dA2}. This step is not a straightforward exercise since it consists
\begin{enumerate}
\item to formulate the Hamilton principle using a Lie derivative,
\item to use the Cartan formula and the lift of vector field,
\item to "integrate by part",
\item and to introduce into the boundary integral the contact form to reveal the Poincaré-Cartan form.
\end{enumerate}
Following~\cite{Echeverria00} in applying the Hamilton principle \ref{def:HP}, we look for “stationary” action with respect to a diffeomorphism $\tau_\Eps^Z$ : a local one-parameter transformation group associated to an arbitrary vector field  $Z\in \Chp{E}$. This transformation induces a diffeomorphism $\tau_M:M\rightarrow M$ on the basis $M$ in such a way that the following diagram commutes
\begin{diagram}
\JE&\rTo^{\prolonge{\tau_\Eps^Z}}&\JE\\
\dTo^{\pi^1}& &\dTo^{\pi^1}\\
E&\rTo^{\tau_\Eps^Z}&E\\
\dTo^{\pi}& &\dTo^{\pi}\\
M&\rTo^{\tau_M}&M\\
\end{diagram}
Then, the variations of $\sect$ are given by a holonomic section $\sect_\Eps=\tau_\Eps^Z \circ \sect \circ \tau_M^{-1}$. Since these diffeomorphisms are characterized by the following property
$
\prolonge{\sect_\Eps}=\prolonge{(\tau_\Eps^Z\circ\sect\circ \tau_M^{-1})} = \prolonge{\tau_\Eps^Z} \circ\prolonge{\sect}\circ \tau_M^{-1}
$, the variation of the action functional 
\begin{equation*}
\delta \mathcal{A} =\left.\der{}{\Eps}\right|_{\Eps = 0}\mathcal{A}(\sect_\Eps)
= \lim_{\Eps\rightarrow 0}
\frac{\mathcal{A}(\sect_\Eps)- \mathcal{A}(\sect)}{\Eps},
\end{equation*}
may be written, for $\mathcal{U}\subset M$, using a Lie derivative as
\begin{eqnarray*}
\delta \mathcal{A} &=&
\lim_{\Eps\rightarrow 0}\frac{1}{\Eps} 
\left(\int_{\tau_M(\mathcal{U})} \pullback{\sect_\Eps}\Lag - \int_{\mathcal{U}} \pullback{\sect}\Lag\right)
=\lim_{\Eps\rightarrow 0}\frac{1}{\Eps} 
\left(\int_{\tau_M(\mathcal{U})} \pullback{\tau_\Eps^Z \circ\prolonge{\sect}\circ \tau_M^{-1}}\Lag - \int_{\mathcal{U}} \pullback{\sect}\Lag\right)\\
&=&\lim_{\Eps\rightarrow 0}\frac{1}{\Eps} 
\left(\int_{\tau_M(\mathcal{U})} (\tau_M^{-1})^*\pullback{\sect}\pullback{\tau_\Eps^Z} \Lag - \int_{\mathcal{U}} \pullback{\sect}\Lag\right)
=\lim_{\Eps\rightarrow 0}\frac{1}{\Eps} 
\left(\int_{\mathcal{U}} \pullback{\sect}\pullback{\tau_\Eps^Z} \Lag - \int_{\mathcal{U}} \pullback{\sect}\Lag\right)\\
&=&\lim_{\Eps\rightarrow 0}\frac{1}{\Eps} 
\int_{\mathcal{U}} \pullback{\sect}\left[\pullback{\tau_\Eps^Z} \Lag -\Lag\right]
= \int_{\mathcal{U}}\pullback{\sect}\Lie{j^1Z}\Lag.
\end{eqnarray*}
Using the Cartan formula, $\Lie{j^1Z}\Lag=\dd(j^1Z\contr \Lag)+j^1Z\contr
\dd\Lag$, and applying Stokes' theorem, this yields to
\begin{align}\label{eq:split1}
\delta \mathcal{A} = \int_{\partial\mathcal{U}}\pullback{\sect}\left(j^1Z\contr
\Lag\right)+\int_{\mathcal{U}}\pullback{\sect}j^1Z\contr d\Lag= I_1 + I_2.
\end{align}
Let us compute $j^1Z\contr d\Lag$ to determine $I_2$ first. It is convenient to introduce the contact form $\fcv$ in the proposition \ref{anx:liftVF:geom} which gives the lift of $Z$ at point $\ptJE=(x^{\mu},\y{A},\v{\mu}{A})$. With the notation $\z{A}=\beta^A- \v{\nu}{A}\alpha^{\nu}=\fc{A}(j^1Z)$ and the partial derivatives of the Lagrangian, we have
\begin{eqnarray*}
j^1Z\contr \dd\Lag&=&j^1Z\contr \left(\dd\Ld\wedge \vol\right)
=j^1Z\contr \left(
\derp{\Ld}{\y{A}}\fc{A}\wedge\vol+\derp{\Ld}{\v{\mu}{A}}\dv{\mu}{A}\wedge\vol
\right)\quad \text{(since $\dy{A}\wedge\vol=\fc{A}\wedge\vol$)}\\
&=&\derp{\Ld}{\y{A}}
\left[
\fc{A}(j^1Z)\vol-\fc{A}\wedge\vol(j^1Z)
\right]+
\derp{\Ld}{\v{\mu}{A}}
\left[
\dv{\mu}{A}(j^1Z)\vol-\dv{\mu}{A}\wedge\vol(j^1Z)
\right]\\
%%%%%%%%%%%%%%%%%%%%%%%%%
&=&\derp{\Ld}{\y{A}}
\left[
\z{A}\vol-\alpha^\lambda\fc{A}\wedge\dnx{n}{\lambda}
\right]+
\derp{\Ld}{\v{\mu}{A}}
\left[
\left(\derp{\z{A}}{\x{\mu}}+\v{\mu}{B}\derp{\z{A}}{\y{B}}\right)\vol
-\alpha^\lambda\dv{\mu}{A}\wedge\dnx{n}{\lambda}
\right],
\end{eqnarray*}
using the notation $d^{n}x_{\mu}=\basex{\mu}\contr \vol$ with the identity $dx^{\nu}\wedge d^{n}x_{\mu}=\vol\delta_{\nu\mu}$ (see appendix \ref{app:OUF}). The term $\derp{\z{A}}{\x{\mu}}\vol$ may be written in another way 
\begin{eqnarray*}
\derp{\z{A}}{\x{\mu}}\vol=\derp{\z{A}}{\x{\nu}}\vol\delta_{\mu\nu}
&=&\derp{\z{A}}{\x{\nu}}\dx{\nu}\wedge\dnx{n}{\mu}
=\left(\dd \z{A}-\derp{\z{A}}{\y{B}}\dy{B}-\derp{\z{A}}{\v{\lambda}{B}}\dv{\lambda}{B}\right)\wedge\dnx{n}{\mu}\\
%%%%%
&=&\dd \z{A}\wedge\dnx{n}{\mu}
-\derp{\z{A}}{\y{B}}\dy{B}\wedge\dnx{n}{\mu}+
\alpha^\lambda\dv{\lambda}{A}\wedge\dnx{n}{\mu},
\end{eqnarray*}
since $\derp{\z{A}}{\v{\lambda}{B}}= -\alpha^\lambda\delta_{B}^{A}$. Futhermore, taking into account that ("integration by part") 
\begin{eqnarray*}
\derp{\Ld}{\v{\mu}{A}}\dd \z{A}\wedge\dnx{n}{\mu}=\dd \left(\z{A}\derp{\Ld}{\v{\mu}{A}}\dnx{n}{\mu}\right)
-\z{A} \dd \left(\derp{\Ld}{\v{\mu}{A}}\right)\wedge\dnx{n}{\mu},
\end{eqnarray*}
and writing $\vol=\frac{1}{n+1}\dx{\mu}\wedge\dnx{n}{\mu}$, a total differential and the contact form $\fcv$ are displayed
\begin{eqnarray*}
j^1Z\contr \dd\Lag&=&
\dd \left(\z{A}\derp{\Ld}{\v{\mu}{A}}\dnx{n}{\mu}\right)
-\left[ \z{A} \left(
\dd \left(\derp{\Ld}{\v{\mu}{A}}\right)-
\frac{1}{n+1}\derp{\Ld}{\y{A}}\dx{\mu}\right)\wedge\dnx{n}{\mu}\right]\\
&-&\left[
\derp{\Ld}{\v{\mu}{A}}\derp{\z{A}}{\y{B}}
+\derp{\Ld}{\y{B}}\alpha^\mu
\right]\fc{B}\wedge\dnx{n}{\mu}
+\left[
\derp{\Ld}{\v{\mu}{A}}\alpha^\lambda
\left(
\dv{\lambda}{A}\wedge\dnx{n}{\mu}-
\dv{\mu}{A}\wedge\dnx{n}{\lambda}
\right)
\right]
\end{eqnarray*}
Using the holomic criteria $\pullback{\sect}\fc{B}\equiv 0$, and the fact that
\begin{equation*}
\pullback{\sect}
\left(
\dv{\lambda}{A}\wedge\dnx{n}{\mu}-
\dv{\mu}{A}\wedge\dnx{n}{\lambda}
\right)
=
\derpm{\sect^A}{\x{\mu}}{\x{\lambda}}-\derpm{\sect^A}{\x{\lambda}}{\x{\mu}}
\vol=0,
\end{equation*}
the two last terms are canceled when they are evaluated along $\prolonge\sect$. Then by Stokes'  theorem again a new expression of the
variation of the action~\refp{eq:split1} is obtained
\begin{eqnarray*}
\delta\mathcal{A}&=&\int_{\partial\mathcal{U}}\pullback{\sect}\left(\z{A}\derp{\Ld}{\v{\mu}{A}}\dnx{n}{\mu}+j^1Z\contr\Lag
\right)
-\int_{\mathcal{U}}\pullback{\sect}\left[ \z{A} \left(
\dd \left(\derp{\Ld}{\v{\mu}{A}}\right)-
\frac{1}{n+1}\derp{\Ld}{\y{A}}\dx{\mu}\right)\wedge\dnx{n}{\mu}\right]
,\quad \forall \zv.
\end{eqnarray*}
As usual, the first integral, along the boundaries $\partial \mathcal{U}$, can also be written 
\begin{eqnarray*}
\int_{\partial\mathcal{U}}\pullback{\sect}\left[
j^1Z\contr\left(\derp{\Ld}{\v{\mu}{A}}\fc{A}\wedge\dnx{n}{\mu}+\Lag\right)\right],
\end{eqnarray*}
using the
contact form with $\z{A}=\fc{A}(j^1Z)$ and using the holomic criteria again to insure that $\pullback{\sect}\left[
\left(j^1Z\contr\fc{A}\right)\wedge\dnx{n}{\mu}\right]=\pullback{\sect}\left[
j^1Z\contr\left(\fc{A}\wedge\dnx{n}{\mu}\right)\right]$.
Gathering all the terms together, the shape of equation \refp{eq:dA2} ($\delta\mathcal{A}=\int_{\partial\mathcal{U}}\ldots-\int_{\mathcal{U}}\ldots$) is obtained if the (Lagrangian) Poincaré-Cartan $(n+1)$-form in $\JE$ is 
introduced by
\begin{align*}
\PCL=\derp{\Ld}{\v{\mu}{A}}\fc{A}\wedge\dnx{n}{\mu}+\Lag.
\end{align*}
\paragraph*{Remark} The Legendre transforms displayed in the well-known Poincaré-Cartan (n+1)-form $\PCL$ in~\refp{eq:PC}, appear naturally by replacing the contact form $\fcv$ by its expression in coordinates~\refp{eq:contact}.
However, it is useful to keep the contact form and to introduce the Euler-Lagrange form~\refp{eq:ELF} in the Lagrangian multisymplectic $(n+2)$-form on $\JE$, $\PSL =-\dd \PCL $ given in~\refp{PSL}. These calculations are carried out starting from the last convenient expression of $\PCL$
\begin{eqnarray*}
-\PSL&=&\dd \PCL=
\dd\left(\derp{\Ld}{\v{\mu}{A}}\right)\wedge\fc{A}\wedge\dnx{n}{\mu}
+\derp{\Ld}{\v{\mu}{A}}\dd \fc{A}\wedge\dnx{n}{\mu}
+\dd\Ld\wedge\vol\\
&=&\left(\mathcal{T}_{A}^{\mu}+
\frac{1}{n+1}\derp{\Ld}{\y{A}}\dx{\mu}\right)
\wedge\fc{A}\wedge\dnx{n}{\mu}
-\derp{\Ld}{\v{\mu}{A}}\dd \v{\nu}{A}\wedge \dx{\nu}
\wedge\dnx{n}{\mu}
+\derp{\Ld}{\y{A}}\dy{A}\wedge\vol+\derp{\Ld}{\v{\mu}{A}}\dd \v{\mu}{A}\wedge\vol\\
&=&\mathcal{T}_{A}^{\mu}
\wedge\fc{A}\wedge\dnx{n}{\mu}
-\derp{\Ld}{\y{A}}
\fc{A}\wedge\vol
-\cancel{\derp{\Ld}{\v{\mu}{A}}\dd \v{\mu}{A}\wedge \vol}
+\derp{\Ld}{\y{A}}\dy{A}\wedge\vol
+\cancel{\derp{\Ld}{\v{\mu}{A}}\dd \v{\mu}{A}\wedge\vol}\\
&=&\mathcal{T}_{A}^{\mu}\wedge
\fc{A}\wedge\dnx{n}{\mu}
-\derp{\Ld}{\y{A}}
\dy{A}\wedge\vol
+\derp{\Ld}{\y{A}}\dy{A}\wedge\vol\\
&=&\mathcal{T}_{A}^{\mu}\wedge
\fc{A}\wedge\dnx{n}{\mu}.\diamond
\end{eqnarray*}

\section{Lift of vector fileds with Lie groups}\label{app:LVFG}

This appendix is dedicated to the one-jet prolongation of vector fields on a principal $G$ bundle. The process is the one used in app.~\ref{sect:jet-pGeom} adapted to the presence of a Lie group $G$.

\paragraph*{Step (i)}Let $\ptJE=(x^{\mu},\g{A},\xiL^A_\mu)$ be a given point of $\JE$ in a local system of coordinates and let $Z\in \Chp{E}$ be a vector field over $E$ expressed in the left-invariant basis $\baseyTG{A}$ given by $Z = \alpha^{\mu}\basex{\mu}+\beta^A\baseyTG{A}$.

\paragraph*{Step (ii)}
The inverse holonomic map \refp{def:invholMap} is used to compute the tangent vector fields representative of the point $\ptJE$ by  $(\ptE,X)=\Hol^{-1}(\ptJE)$. This is achieved by choosing the contact form~\refp{eq:contactG} adapted to the principal $G$ bundle
\begin{equation}\label{holonomicXmuG}
\fcvG_{\ptJE}((X,0))=0\Leftrightarrow \mcfL^A\left(X\right)=\vL{\nu}{A}\dx{\nu}\left(X\right).
\end{equation}
It gives  a field of (hyper)planes characterized by $n+1$ tangent vector fields. Without loss of generality, these vectors can be 
normalized to
\begin{equation}\label{eq:NXG}
{X}_{\mu}=\basex{\mu}+\vL{\mu}{B}\baseyL{B},
\quad \text{with }\dx{\nu}\left({X}_{\mu}\right)=\delta_\mu^\nu
,\quad ({\mu}=1,\ldots,n+1).
\end{equation}
Obviously ${X}_{\mu}$ verifies \refp{holonomicXmuG}.
%%%%%%%%%%%%%%%%%%%%%%%%%%%%
\paragraph*{Step (iii)}
Let us consider an integral curve $\mathcal{C}$ of the one-parameter transformation group $\tau_{ \Eps}^Z$ along $Z$  given by
\begin{eqnarray*}
\mathcal{C}: \R \supset I &\rightarrow& E\\
\Eps&\mapsto& (\ptE_\Eps,{X}_\mu^{\Eps})=(\tau^Z_\Eps(\ptE),T_{\ptE}{(\tau_{ \Eps}^Z)} (X_\mu)).
\end{eqnarray*}
%%%%%%%%%%%%%%%%%%%%%%%%%%%%
\paragraph*{Step (iv)}
This last procedure involves the bracket of vector fields. More precisely, the Lie derivative of the two vector fields $X_\mu$ and $Z$, at point $\ptE$, is given by (see~\cite{jost2005riemannian}  for example),
\begin{eqnarray*}
&&\Lie{Z}X_\mu=\crochetL{Z}{X_\mu}\Big|_\ptE=\der{}{ \Eps}\bigg|_{ \Eps=0}\left((\tau_\Eps^Z)^* X_\mu\big|_{\ptE_\Eps}\right)\\
&&=\der{}{ \Eps}\bigg|_{ \Eps=0}\left(T_{\ptE_\Eps}{(\tau_{- \Eps}^Z)} (X_\mu)\right)
=\lim_{ \Eps\rightarrow 0}\frac{T_{\ptE_\Eps}{(\tau_{- \Eps}^Z)} (X_\mu)-X_\mu\big|_{\ptE}}{ \Eps}
\end{eqnarray*}
where $\ptE_\Eps=\tau_{\Eps}^Z(\ptE)$. Evaluating this definition at point $\ptE_\Eps$, one obtains a finite expansion of the tangent map of the transformation $\tau_{ \Eps}^Z$
\begin{equation}\label{expXmuG}
X_\mu^{ \Eps}\big|_{\ptE_\Eps}=T_{\ptE}{(\tau_{ \Eps}^Z)} (X_\mu)=X_\mu\big|_{\ptE_\Eps}+ \Eps\crochetL{X_\mu}{Z}\big|_{\ptE_\Eps}+\mathcal{O}( \Eps^2).
\end{equation}
%%%%%%%%%%%%%%%%%%%%%%%%%%%%
\paragraph*{Step (v)}
Using definition~\ref{def:holMap}, the curve $\Eps \mapsto \ptJE_\Eps=\Hol(\ptE_\Eps,{X}_\mu^{\Eps})$ is obtained by solving
\begin{equation}\label{holonomicXmuEpsG}
\fcvG_{\ptJE_{\Eps}}((X_\mu^{\Eps},0))=0\Leftrightarrow (\mcfL_{\Eps})^A\left(X_\mu^\Eps\right)=(\xiL_{\Eps})^A_{\nu}\dx{\nu}\left(X_\mu^\Eps\right),
\end{equation}
where the contact form \refp{eq:contactG} is used again.
%%%%%%%%%%%%%%%%%%%%%%%%%%%%
\paragraph*{Step (vi)}
This expression involves the variation of the Maurer-Cartan form $(\mcfL_{ \Eps})^A$ in the $Z$ direction that can be expressed using the Lie derivative formula 
\begin{equation*}
\left(\Lie{Z}\mcfL^A\right)({X}_\mu)
=\lim_{\Eps\rightarrow 0}\frac{ \left.(\mcfL_{ \Eps})^A\right|_{\ptE_{\Eps}}(X_\mu^\Eps)-\left.\mcfL^A\right|_{\ptE}({X}_\mu)}{\Eps},
\end{equation*}
to give
\begin{equation}\label{eq:DLmcG}
\left.(\mcfL_\Eps)^A\right|_{\ptE_{\Eps}}(X_\mu^\Eps)=\left.\mcfL^A\right|_{\ptE}({X}_\mu)+\Eps\left(\Lie{Z}\mcfL^A\right)({X}_\mu)+\mathcal{O}(\Eps^2).
\end{equation}

%%%%%%%%%%%%%%%%%%%%%%%%%%%%
\paragraph*{Step (vi)}
Using the finite expansion \refp{expXmuG} and \refp{eq:DLmcG} and tacking into account \refp{holonomicXmuG} and \refp{eq:NXG}, the holonomic criteria \refp{holonomicXmuEpsG} becomes 
\begin{equation}
\left((\xiL_{\Eps})^A_{\nu}- \xiL^A_{\mu}\right) =
\Eps
\Big(
\left(\Lie{Z}\mcfL^A\right)({X}_{\mu})
-\xiL^A_{\nu}dx^{\nu}(\crochetL{{X}_{\mu}}{Z})
\Big)
 +\mathcal{O}(\Eps^2).
\end{equation}
This expansion furnishes the third component of the jet-prolongation, $j^1Z$, at point $(\x{\mu},\g{A},\vL{\mu}{A})$ 
\begin{equation}\label{eq:3CG}
\gamma^A_{\mu}=d\xiL^A_\mu(j^1Z)=\left(\Lie{Z}\mcfL^A\right)({X}_{\mu})
-\xiL^A_{\nu}dx^{\nu}(\crochetL{{X}_{\mu}}{Z})
\end{equation}
- an expression that can be simplified using three more steps.
%%%%%%%%%%%%%%%%%%%%%%%%%%%%
\paragraph*{Step (vii)}
Using the vectorial 1-form $\mcfLv$, we have on one hand, $\mcfLv_{\ptE}({X}_{\mu})=\boldsymbol{\xiL}_{\mu}$.  On the other hand taking into account the Maurer-Cartan equation for $\mcfLv$ (zero curvature equation)
\begin{equation}\label{eq:MCeq }
d\mcfLv+\crochetL{\mcfLv}{\mcfLv}=0,
\end{equation}
we have
\begin{eqnarray*}
&&\left(\Lie{Z}\mcfLv\right)({X}_{\mu})=
\left(
d(Z\contr \mcfLv)+(Z\contr d\mcfLv)
\right)({X}_{\mu})\\
&=&
\left(
d\boldsymbol{\beta}-(Z\contr \crochetL{\mcfLv}{\mcfLv})
\right)({X}_{\mu})\quad \text{since $d\mcfLv=-\crochetL{\mcfLv}{\mcfLv}$}\\
&=&
\left(
d\boldsymbol{\beta}-(\crochetL{\mcfLv(Z)}{\mcfLv})
\right)({X}_{\mu})\\
&=&
d\boldsymbol{\beta}({X}_{\mu})+\crochetL{\mcfLv({X}_{\mu})}{\mcfLv(Z)}
\\
&=&
d\boldsymbol{\beta}({X}_{\mu})+\crochetL{\boldsymbol{\xiL}_{\mu}}{\boldsymbol{\beta}},\quad
\left(=d\boldsymbol{\beta}({X}_{\mu})+\mcfLv\crochetL{{X}_{\mu}}{Z}\right)
\end{eqnarray*}
and then \refp{eq:3CG} becomes
\begin{eqnarray*}
d\xiL^A_\mu(j^1Z)&=&d\beta^A ({X}_{\mu})+ \underbrace{\mcfL^A(\crochetL{{X}_{\mu}}{Z})-\xiL^A_{\nu}dx^{\nu}(\crochetL{{X}_{\mu}}{Z})}_{\fcG{A}_{\ptJE}(\crochetL{{X}_{\mu}}{Z})}\\
&=& d\beta^A ({X}_{\mu}) + \crochetL{\boldsymbol{\xiL}_{\mu}}{\boldsymbol{\beta}}^A -\xiL^A_{\nu}dx^{\nu}(\crochetL{{X}_{\mu}}{Z}).
\end{eqnarray*}
%%%%%%%%%%%%%%%%%%%%%%%%%%%%
\paragraph*{Step (viii)}
As the differential of the map $(x^{\mu},y^A)\mapsto \baseyTG{A}$ has no component along $\basex{\mu}$, \ie   $d(\baseyTG{A})=0\basex{\mu}+\eta_A^B\baseyTG{B}$ where $\eta_A^B$ is a 1-form, the computation of the commutator
 \begin{eqnarray*}
&&\crochetL{{X}_{\mu}}{Z}=dZ({X}_{\mu})-d{X}_{\mu}(Z)\\
&=&
\left(d\alpha^{\nu} \basex{\nu}+d\beta^A \baseyTG{A} +\beta^A d(\baseyTG{A})\right)({X}_{\mu})-
\left(\xiL^A_{\mu} d(\baseyTG{A})\right)(Z)\\
&=&d\alpha^{\nu} ({X}_{\mu})\basex{\nu}+\left[
(d\beta^B+\beta^A\eta_A^B)({X_\mu})-\xiL^A_{\mu} \eta_A^B(Z)
\right]\baseyTG{B},
\end{eqnarray*}
shows that the term $dx^{\nu}(\crochetL{{X}_{\mu}}{Z})$ equals $d\alpha^{\nu}({X}_{\mu})$. It is important to notice, in this computation, that the ${X}_{\mu}$ are defined for a fixed $\xiL^B_{\mu}.$ So we have now
\begin{equation}\label{eq:prolong4G}
d\xiL^A_\mu(j^1Z)=d\beta^A ({X}_{\mu})+\crochetL{\boldsymbol{\xiL}_{\mu}}{\boldsymbol{\beta}}^A-\xiL^A_{\nu}d\alpha^{\nu}({X}_{\mu}).
\end{equation}
%%%%%%%%%%%%%%%%%%%%%%%%%%%%
\paragraph*{Step (ix)}
Expressing any vector ${X}_{\mu}$ in both basis $\baseyTG{A}$ and $\basey{A}$ gives ${X}_{\mu}=\basex{\mu}+\xiL^I_{\mu}\baseyTG{I}=\basex{\mu}+v^Q_{\mu}\basey{Q}$, for some $\xiL^I_{\mu}$ and $v^Q_{\mu}$. That is, as $\dy{A}$ is the dual basis to the basis $\basey{A}$,
\begin{equation*}
\begin{cases}
	\mcfL^A({X}_{\mu})&=\xiL^A_{\mu}=\mcfL^A(\basex{\mu}+v^Q_{\mu}\basey{Q})=v^Q_{\mu}\mcfL^A(\basey{Q})\\
	\dy{B}({X}_{\mu})&=v^B_{\mu}=\dy{B}(\basex{\mu}+\xiL^I_{\mu}\baseyTG{I})=\xiL^I_{\mu} \dy{B}(\baseyTG{I})
\end{cases}
\end{equation*} and then since $\beta^A$ is a function of $\b x$ and $\b y$, we have
\begin{eqnarray*}
&&d\beta^A({X}_{\mu})=\left(\derp{\beta^A}{x^{\nu}}dx^{\nu}+\derp{\beta^A}{y^B}\dy{B}\right)
\left(\basex{\mu}+\xiL^I_{\mu}\baseyTG{I}\right)\\
&&=\derp{\beta^A}{x^{\mu}}+\derp{\beta^A}{y^B}\xiL^I_{\mu}\dy{B}(\baseyTG{I})=
\left(\derp{\beta^A}{x^{\mu}}+v^B_{\mu}\derp{\beta^A}{y^B}
\right).
\end{eqnarray*}
Same computation for $d\alpha^{\nu}(X_\mu)$ gives 
\begin{equation*}
d\alpha^{\nu}(X_\mu) = \derp{\alpha^{\nu}}{x^\mu} + v_\mu^B\derp{\alpha^{\nu}}{y^B}.
\end{equation*}
Finally, including these expressions into \refp{eq:prolong4G} and introducing $\zG{A}=\beta^A-\xiL^A_{\nu} \alpha^{\nu}$, the lift of vector fields is given by proposition~\ref{prop:liftG}.
%%%%%%%%%%%%%%%%%%%%%%%%%%%%%%%%%%%

\section{The Hamilton principle: calculus of variations with reduction by Lie groups}\label{app:dAG}

This appendix is dedicated to obtain the variation of the action~\refp{eq:dA3} using a reduced Lagrangian starting from the Hamilton principle~\ref{def:HP}. This step is mainly the one used for appendix~\ref{anx:deltaA} except that now the lift of $Z$ is given by proposition~\ref{prop:liftG} where a Lie bracket appears. As in these preceding computations, the variation is splitted into two integrals
\begin{align}\label{eq:G_split1}
\delta \mathcal{A} = \int_{\partial\mathcal{U}}\pullback{\sect}\left(j^1Z\contr
\lag\right)+\int_{\mathcal{U}}\pullback{\sect}j^1Z\contr d\lag= I_1 + I_2.
\end{align}
With the lift of $Z$ at point $\ptJE=(x^{\mu},\y{A},\vL{\mu}{A})$ given by proposition~\ref{prop:liftG}, let us compute $j^1Z\contr d\lag$ to determine $I_2$ first. It is convenient to introduce the contact form $\fcvG$~(\ref{eq:contactG}) in this proposition. With the notation $\zG{A}=\beta^A- \vL{\nu}{A}\alpha^{\nu}= \fcG{A}(j^1Z)$, so that $\vol(j^1Z)=\alpha^\nu\wedge\dnx{n}{\nu}$ and the partial derivatives of the Lagrangian, we have
\begin{eqnarray*}
j^1Z\contr \dd\lag&=&j^1Z\contr \left(\dd\ld\wedge \vol\right)
=j^1Z\contr \left(
\derp{\ld}{\y{A}}T^{A}_{B}\fcG{B}\wedge\vol+\derp{\ld}{\vL{\mu}{A}}\dd \vL{\mu}{A}\wedge\vol
\right)\\
&=&\derp{\ld}{\y{A}}
T^{A}_{B}\left[
\fcG{B}(j^1Z)\vol-\fcG{B}\wedge\vol(j^1Z)
\right]+
\derp{\ld}{\vL{\mu}{A}}
\left[
\dd \vL{\mu}{A}(j^1Z)\vol-\dd \vL{\mu}{A}\wedge\vol(j^1Z)
\right]\\
%%%%%%%%%%%%%%%%%%%%%%%%%
&=&\derp{\ld}{\y{A}}T^{A}_{B}
\left[
\zG{B}\vol-\alpha^\nu \fcG{B}\wedge\dnx{n}{\nu}
\right]+
\derp{\ld}{\vL{\mu}{A}}
\left[
\left(
\derp{\zG{A}}{x^{\mu}}+\xiL^C_{\mu}T^B_C
\derp{\zG{A}}{y^B}
+\crochetL{\boldsymbol{\xiL}_{\mu}}{\boldsymbol{\beta}}^A\right)\vol
-\alpha^\nu\dd \vL{\mu}{A}\wedge\dnx{n}{\nu}
\right].
\end{eqnarray*}
The term $\derp{\zG{A}}{\x{\mu}}\vol$ may be written in another way 
\begin{eqnarray*}
\derp{\zG{A}}{\x{\mu}}\vol=\derp{\zG{A}}{\x{\nu}}\vol\delta_{\mu\nu}
=\derp{\zG{A}}{\x{\nu}}\dx{\nu}\wedge\dnx{n}{\mu}
&=&\left(\dd \zG{A}-\derp{\zG{A}}{\y{B}}{T}^B_C\mcfL^{C}-\derp{\zG{A}}{\vL{\nu}{B}}\dd \vL{\nu}{B}\right)\wedge\dnx{n}{\mu}\\
%%%%%
&=&\dd \zG{A}\wedge\dnx{n}{\mu}-\derp{\zG{A}}{\y{B}}{T}^B_C\mcfL^{C}\wedge\dnx{n}{\mu}+
\alpha^\nu\dd \vL{\nu}{A}\wedge\dnx{n}{\mu}.
\end{eqnarray*}
Futhermore, taking into account (integration by part) that
\begin{eqnarray*}
\derp{\ld}{\vL{\mu}{A}}\dd \zG{A}\wedge\dnx{n}{\mu}=\dd \left(\zG{A}\derp{\ld}{\vL{\mu}{A}}\dnx{n}{\mu}\right)
-\zG{A} \dd \left(\derp{\ld}{\vL{\mu}{A}}\right)\wedge\dnx{n}{\mu},
\end{eqnarray*}
a total differential and the contact form $\fcvG$ are displayed
\begin{eqnarray*}
j^1Z\contr \dd\lag&=&
\dd \left(\zG{A}\derp{\ld}{\vL{\mu}{A}}\dnx{n}{\mu}\right)
-\left[ \zG{A} \left(
\dd \left(\derp{\ld}{\vL{\mu}{A}}\right)\wedge\dnx{n}{\mu}
-\left(\ade{\vLv_{\nu}}{\derp{\ld}{\vLv_{\nu}}}\right)_A\vol
-T^{B}_{A}\derp{\ld}{y^B}\vol
\right)\right]\\
&-&T^{B}_{C}\left[
\derp{\ld}{\vL{\mu}{A}}\derp{\zG{A}}{\y{B}}
+\derp{\ld}{\y{B}}\alpha^\mu
\right]\fcG{C}\wedge\dnx{n}{\mu}
+\left[
\derp{\ld}{\vL{\mu}{A}}\alpha^\nu
\left(
\dd \vL{\nu}{A}\wedge\dnx{n}{\mu}-
\dd \vL{\mu}{A}\wedge\dnx{n}{\nu}
+\crochetL{\vLv_{\mu}}{\vLv_{\nu}}^A\vol
\right)
\right],
\end{eqnarray*}
with the co-adjoint operator $\ade{}{}$ defined by 
$
\dual{\boldsymbol{\pi}}{\crochetL{\boldsymbol{\xiL}_{\mu}}{\boldsymbol{\beta}}}=
\dual{\boldsymbol{\pi}}{\ad{\boldsymbol{\xiL}_{\mu}}{\boldsymbol{\beta}}}=
\dual{\ade{\boldsymbol{\xiL}_{\mu}}{\boldsymbol{\pi}}}{\boldsymbol{\beta}}
$. The two last terms in parentheses are canceled when they are evaluated along $\prolonge\sect$. This is achieved by using the holomic criteria $\pullback{\sect}\fcG{C}\equiv 0$ and the Maurer-Cartan equation $d\mcfLv+\crochetL{\mcfLv}{\mcfLv}=0$ (see appendix~\ref{app:zero}). Then by Stokes'  theorem and using the contact form with $\zG{A}=\fcG{A}(j^1Z)$, a new expression of the
variation of the action~\refp{eq:G_split1} is obtained as
\begin{equation*}
\delta\mathcal{A}=\int_{\partial\mathcal{U}}\pullback{\sect}
\left(
j^1Z\contr\PCl
\right)
-\int_{\mathcal{U}}
\pullback{\sect}\left(
j^1Z\contr
\left(
\fcG{A}\wedge\Gamma_{A}
\right)\right),
\end{equation*}
with
\begin{equation*}
\PCl=\derp{\ld}{\vL{\mu}{A}}\fcG{A}\wedge\dnx{n}{\mu}+\lag,
\quad\text{and}\quad
\Gamma_{A}=
\dd \left(\derp{\ld}{\vL{\mu}{A}}\right)\wedge\dnx{n}{\mu}
-\left(\ade{\vLv_{\nu}}{\derp{\ld}{\vLv_{\nu}}}\right)_A\vol
-T^{B}_{A}\derp{\ld}{y^B}\vol.
\end{equation*}
\paragraph*{Remark} The Legendre transforms displayed in the (reduced) Poincaré-Cartan (n+1)-form $\PCl$ given by~\refp{eq:G_PC}, appear naturally by replacing the contact form $\fcvG$ by its coordinates expression~\refp{eq:contactG}. However,  the Lagrangian pre-multisymplectic $(n+2)$-form on $\JE$,
\begin{equation*}
\PSl=
\fcG{A}\wedge\Gamma_{A}
+\derp{\ld}{\vL{\mu}{A}}\crochetL{\fcvG}{\fcvG}^A
\wedge\dnx{n}{\mu}
\end{equation*}
given in~\refp{eq:PSl} is obtained by carrying out 
the following calculations starting from the last convenient expression of $\PCl$
%%%%%%%%%%%%%%%%%%%%%%%%
\begin{eqnarray*}
\PSl&=&-\dd \PCl=-\dd \left(
\derp{\ld}{\vL{\mu}{A}}\fcG{A}\wedge\dnx{n}{\mu}+\lag
\right)=
\fcG{A}\wedge\dd \left(\derp{\ld}{\vL{\mu}{A}}\right)\wedge\dnx{n}{\mu}
-\derp{\ld}{\vL{\mu}{A}}\dd \fcG{A}\wedge\dnx{n}{\mu}
-\dd\ld\wedge\vol\\
%%%%%%%%%%%%
&=&\fcG{A}\wedge
\left(
\Gamma_{A}+
\left(\ade{\vLv_{\nu}}{\derp{\ld}{\vLv_{\nu}}}\right)_A\vol
+T^{B}_{A}\derp{\ld}{y^B}\vol
\right)
\\
%%%%%%%%
&-&\derp{\ld}{\vL{\mu}{A}}\dd \mcfL^{A}\wedge\dnx{n}{\mu}
+\cancel{\derp{\ld}{\vL{\mu}{A}}\dd \vL{\nu}{A}\wedge \dx{\nu}\wedge\dnx{n}{\mu}}
-\derp{\ld}{\y{A}}\dy{A}\wedge\vol
-\cancel{\derp{\ld}{\vL{\mu}{A}}\dd \vL{\mu}{A}\wedge\vol}\\
%%%%%%%%
%%%%%%%%%%%%
&=&\fcG{A}\wedge\Gamma_{A}
+\left(\ade{\vLv_{\nu}}{\derp{\ld}{\vLv_{\nu}}}\right)_A\mcfL^{A}\wedge\vol
+\cancel{T^{B}_{A}\derp{\ld}{y^B}\mcfL^{A}\wedge\vol}
+\derp{\ld}{\vL{\mu}{A}}\crochetL{\mcfLv}{\mcfLv}^A\wedge\dnx{n}{\mu}
-\cancel{T^{B}_{A}\derp{\ld}{\y{B}}\mcfL^{A}\wedge\vol}
\\
%%%%%%%%%%%%
&=&\fcG{A}\wedge\Gamma_{A}
+\derp{\ld}{\vL{\mu}{A}}
\left(\underbrace{
\crochetL{\vLv_{\mu}}{\mcfLv}^A\wedge\vol
+
\crochetL{\mcfLv}{\mcfLv}^A\wedge\dnx{n}{\mu}}_{\Upsilon_{\mu}^{A}}
\right).
\end{eqnarray*}
As in the general case, since the contact form $\fcG{A}$ and the Euler-Lagrange form $\Gamma_{A}$ vanish along critical section $\prolonge\sect$, the (n+2)-form $\fcG{A}\wedge\Gamma_{A}$ has the same property. So do the additional form $\Upsilon_{\mu}^{A}$ \ie $\pullback{\sect}\left(W\contr\Upsilon_{\mu}^{A}\right)=0$, since we have (see appendix \ref{app:OUF} for the bracket of $\mathfrak{g}$-valued 1-forms)
\begin{eqnarray*}
\Upsilon_{\mu}^{A}
&=&
\crochetL{\vLv_{\mu}}{\mcfLv}^A\wedge\vol
+\crochetL{\mcfLv}{\mcfLv}^A\wedge\dnx{n}{\mu}\\
&=&
\crochetL{\vLv_{\mu}}{\fcvG}^A\wedge\vol
+
\crochetL{\vLv_{\mu}}{\vLv_{\nu}\dx{\nu}}^A\wedge\vol
+
\frac{1}{2}\crochetF{(\fcvG+\vLv_{\nu}\dx{\nu})}{(\fcvG+\vLv_{\alpha}\dx{\alpha})}^A\wedge\dnx{n}{\mu}\\
&=&
\crochetL{\vLv_{\mu}}{\fcvG}^A\wedge\vol
+
\crochetL{\vLv_{\mu}}{\vLv_{\nu}}^A\cancel{\dx{\nu}\wedge\vol}
+
\frac{1}{2}\crochetF{\fcvG}{\fcvG}^A\wedge\dnx{n}{\mu}
+
\crochetF{\fcvG}{\vLv_{\alpha}\dx{\alpha}}^A\wedge\dnx{n}{\mu}
+\frac{1}{2}\crochetF{\vLv_{\nu}\dx{\nu}}{\vLv_{\alpha}\dx{\alpha}}^A\wedge\dnx{n}{\mu}\\
&=&
\crochetL{\vLv_{\mu}}{\fcvG}^A\wedge\vol
+
\frac{1}{2}\crochetF{\fcvG}{\fcvG}^A\wedge\dnx{n}{\mu}
+
\crochetL{\fcvG}{\vLv_{\alpha}}^A\wedge\dx{\alpha}\wedge\dnx{n}{\mu}
+\frac{1}{2}\crochetL{\vLv_{\nu}\dx{\nu}}{\vLv_{\alpha}}^A\wedge\dx{\alpha}\wedge\dnx{n}{\mu}\\
&=&
\cancel{\crochetL{\vLv_{\mu}}{\fcvG}^A\wedge\vol}
+
\crochetL{\fcvG}{\fcvG}^A\wedge\dnx{n}{\mu}
+
\cancel{\crochetL{\fcvG}{\vLv_{\mu}}^A\wedge\vol}
+\frac{1}{2}\crochetL{\vLv_{\nu}\dx{\nu}}{\vLv_{\mu}}^A\wedge\vol\\
&=&
\crochetL{\fcvG}{\fcvG}^A\wedge\dnx{n}{\mu}
+\frac{1}{2}\crochetL{\vLv_{\nu}}{\vLv_{\mu}}^A\cancel{\dx{\nu}\wedge\vol}\\
&=&
\crochetL{\fcvG}{\fcvG}^A\wedge\dnx{n}{\mu}.
\end{eqnarray*}

\section{Change of basis}\label{app:CB}
Let $\dy{A}$ and $\mcfL^{A}$ be the dual basis of $\basey{B}$ and $\baseyTG{B}$ respectively. Any vector belonging to the tangent space of the fiber, $V\in\ChpV{\pi}{E}$, may be written in the two basis with specific coordinates
\begin{equation*}
V=\alpha^B\basey{B}\text{ or } V=\beta^B\baseyTG{B}.
\end{equation*}
So,
\begin{eqnarray*}
\dy{A}(V)&=&\alpha^A=\dy{A}(\beta^B\baseyTG{B})=\dy{A}(\baseyTG{B})\beta^B=\dy{A}(\baseyTG{B})\mcfL^B(V),\quad \forall V\in\ChpV{\pi}{E}
\end{eqnarray*}
Thus, the change of dual basis may be expressed by a matrix $T$ (conversely $L$) such that
\begin{eqnarray*}
\dy{A}&=\dy{A}(\baseyTG{B})\mcfL^B=T^{A}_{B} \mcfL^B,\quad &\text{that is } T^{A}_{B}=\dy{A}(\baseyTG{B}), \text{ conversely}\\
\mcfL^{A}&=\mcfL^{A}(\basey{B})\dy{B}=L^{A}_{B} \dy{B},\quad &\text{that is } L^{A}_{B}=\mcfL^{A}(\basey{B}).
\end{eqnarray*}
For a right invariant basis $\baseyR{B}$ and right Maurer-Cartan form $\mcfR$, we would have
\begin{eqnarray*}
\dy{A}&=\dy{A}(\baseyR{B})\mcfR^B=\tilde{T}^{A}_{B} \mcfR^B,\quad &\text{that is } \tilde{T}^{A}_{B}=\dy{A}(\baseyR{B}), \text{ conversely}\\
\mcfR^{A}&=\mcfR^{A}(\basey{B})\dy{B}=\tilde{L}^{A}_{B} \dy{B},\quad &\text{that is } \tilde{L}^{A}_{B}=\mcfR^{A}(\basey{B}).
\end{eqnarray*}

\section{Elements for the Hamiltonian formulation}\label{app:HF}
\subsection{The canonical form}
With the 1-forms $\theta_\Ha^A$ and $\mathcal{T}^{\mu}_A$ defined in~\refp{eq:1forms} (b,c), the canonical form $\Ocan=\Psymp{\mu}\wedge d^{n}x_{\mu}=dy^A\wedge dp^{\mu}_A \wedge d^{n}x_{\mu}$ yields 
\begin{eqnarray*}
\Ocan&=&
\left(\theta_\Ha^A+\derp{\Hd}{p^{\nu}_A}dx^\nu\right)\wedge
\left(\mathcal{T}^{\mu}_A-\frac{1}{n+1}\derp{\Hd}{y^A}dx^\mu\right)\wedge d^nx_\mu\\
&=&\theta_\Ha^A \wedge \mathcal{T}^{\mu}_A\wedge d^nx_\mu -
\theta_\Ha^A \wedge\derp{\Hd}{y^A}\vol+
\derp{\Hd}{p^{\nu}_A}dx^\nu \wedge  \mathcal{T}^{\mu}_A\wedge d^nx_\mu-
\derp{\Hd}{y^A} \derp{\Hd}{p^{\nu}_A} \cancel{dx^\nu\wedge \vol}.
\end{eqnarray*}
Let us evaluate this $(n+2)$-form on a test vector $W$ and the $(n+1)$-multi-vector $\mvec{X}$ using the lemma
\begin{lemma}[]\label{lemma:M}
Let $\alpha$ and $\beta$ be a $1$-form and $n$-form respectively. Let also $\mvec{X}=X_1\wedge \cdots \wedge X_{n+1}$ be a $n+1$-multivector, we have
\begin{equation}\label{eq:Ir1}
\alpha\wedge\beta (\mvec{X})=(-1)^{i+1}\alpha(X_i)\beta(\hat{\mvec{X}}_{i}),
\end{equation}
where $\hat{\mvec{X}}_{i}=X_1\wedge \cdots \wedge\hat{X}_i\wedge \cdots\wedge X_{n+1}$ is obtain from $\mvec{X}$ eliminating the i-th vector $X_i$.
\end{lemma}
and the corollary
\begin{corollary}[]
With the same hypothesis as the lemma~\ref{lemma:M}, with $\beta$  a $(n+1)$-form and 
 $W$ a vector, we have
\begin{equation}\label{eq:Ir2}
\alpha\wedge\beta (W,\mvec{X})=\alpha(W)\beta(\mvec{X})+(-1)^{j}\alpha(X_j)\beta(W \wedge \hat{\mvec{X}}_{j}).
\end{equation}
\end{corollary}
It gives
\begin{eqnarray*}
\Ocan(W,\mvec{X})&=&
\theta_\Ha ^A (W) \wedge \cancel{\left(\mathcal{T}^{\mu}_A\wedge d^nx_\mu\right) (\mvec{X})}
+ (-1)^{\alpha}\cancel{\theta_\Ha^A(X_{\alpha})} \wedge \left(\mathcal{T}^{\mu}_A\wedge d^nx_\mu\right)(W,\hat{\mvec{X}}_{\alpha})\\
&-&
\theta_\Ha^A(W)\wedge\derp{\Hd}{y^A}\vol(\mvec{X})
- (-1)^{\alpha}\cancel{\theta_\Ha^A(X_{\alpha})} \wedge \derp{\Hd}{y^A}\vol (W,\hat{\mvec{X}}_{\alpha})\\
&+&
\derp{\Hd}{p^{\nu}_A}dx^\nu(W) \wedge \cancel{\left(\mathcal{T}^{\mu}_A\wedge d^nx_\mu\right)(\mvec{X})}
+ (-1)^{\alpha}\derp{\Hd}{p^{\nu}_A}dx^\nu (X_{\alpha}) \wedge \left(\mathcal{T}^{\mu}_A\wedge d^nx_\mu\right) (W,\hat{\mvec{X}}_{\alpha})\\
&=&
-\derp{\Hd}{y^A}\theta_\Ha^A(W)+ (-1)^{\alpha}\derp{\Hd}{p^{\alpha}_A}   \mathcal{T}^{\mu}_A\wedge d^nx_\mu (W,\hat{\mvec{X}}_{\alpha}).
\end{eqnarray*}
It appears clearly that if $W$ is part of  the multi-vector $\mvec{X}$ (that is $W=X_\nu$), this evaluation vanishes: $\mvec{X}\contr \Ocan(X_\nu)=0$, $\forall X_\nu$. In fact,  $\mathcal{T}^{\mu}_A\wedge d^nx_\mu (X_\nu,\hat{\mvec{X}}_{\alpha})=0$ if $\nu=\alpha$. Now, if $\nu\neq\alpha$
\begin{eqnarray*}
\mathcal{T}^{\mu}_A\wedge d^nx_\mu (X_\nu,\hat{\mvec{X}}_{\alpha})&=&
\mathcal{T}^{\mu}_A(X_\nu)d^nx_\mu (\hat{\mvec{X}}_{\alpha})+
(-1)^{\beta}\mathcal{T}^{\mu}_A(X_{\beta}) d^{n}x_{\mu}(X_\nu,\hat{\mvec{X}}_{\alpha \beta})\\
&=&
(-1)^{\alpha+1}\mathcal{T}^{\mu}_A(X_\nu)\delta_{\mu\alpha}+
(-1)^{\beta}\mathcal{T}^{\mu}_A(X_{\beta}) (-1)^{\alpha+1}\delta_{\mu\alpha}(-1)^{\beta+1}\delta_{\beta\nu}\\
&=&
(-1)^{\alpha+1}\mathcal{T}^{\alpha}_A(X_\nu)-(-1)^{\alpha+1}\mathcal{T}^{\alpha}_A(X_\nu)=0.
\end{eqnarray*}
So we have for $W=\beta_\nu X_\nu+W^{v}\in\Chp{\JE}$
\begin{eqnarray*}
\Ocan(W,\mvec{X})
&=&
-\derp{\Hd}{y^A}\theta_\Ha^A(W^{v})+ (-1)^{\alpha}\derp{\Hd}{p^{\alpha}_A}   \mathcal{T}^{\mu}_A\wedge d^nx_\mu (W^{v},\hat{\mvec{X}}_{\alpha})\\
&=&
-\derp{\Hd}{y^A}\theta_\Ha^A(W^{v})+ (-1)^{\alpha}\derp{\Hd}{p^{\alpha}_A}\left[
\mathcal{T}^{\mu}_A(W^{v})d^nx_\mu (\hat{\mvec{X}}_{\alpha})+
(-1)^{\beta}\mathcal{T}^{\mu}_A(X_{\beta})\cancel{ d^{n}x_{\mu}(W^{v},\hat{\mvec{X}}_{\alpha \beta})}\
\right]\\
&=&
-\derp{\Hd}{y^A}\theta_\Ha^A(W^{v})
-\derp{\Hd}{p^{\alpha}_A}\mathcal{T}^{\alpha}_A(W^{v})\\
&=&
-\derp{\Hd}{y^A}\theta_\Ha^A(W^{v})
-\derp{\Hd}{p^{\alpha}_A} dp^{\alpha}_A (W^{v})-\derp{\Hd}{y^A}\derp{\Hd}{p^{\alpha}_A}dx^{\alpha}(W^{v})\\
&=&
-\derp{\Hd}{y^A}\left[
\theta_\Ha^A+\derp{\Hd}{p^{\alpha}_A}dx^{\alpha}
\right]
(W^{v})
-\derp{\Hd}{p^{\alpha}_A} dp^{\alpha}_A (W^{v})\\
&=&
-\derp{\Hd}{y^A}dy^A(W^{v})-\derp{\Hd}{p^{\alpha}_A} dp^{\alpha}_A(W^{v}).
\end{eqnarray*}
\subsection{Elements for the Poisson equation}
Let us compute also
\begin{eqnarray*}
\left(d\Hd \wedge \vol\right)(W,\mvec{X})&=&d\Hd(W)\vol(\mvec{X})+
(-1)^{\alpha}d\Hd(X_{\alpha})\vol(W,\hat{\mvec{X}}_{\alpha})\\
&=&d\Hd(W)+
(-1)^{\alpha}d\Hd(X_{\alpha})\vol(W,\hat{\mvec{X}}_{\alpha}).
\end{eqnarray*}
Evaluated on $W=\beta_\nu X_\nu+W^{v}$ it gives 
\begin{eqnarray*}
\left(d\Hd \wedge \vol\right)(\beta_\nu X_\nu+W^{v},\mvec{X})
&=&
d\Hd(\beta_\nu X_\nu)+d\Hd(W^{v})+
(-1)^{\alpha}d\Hd(X_{\alpha})\vol(\beta_\nu X_\nu,\hat{\mvec{X}}_{\alpha})\\
&+&
(-1)^{\alpha}d\Hd(X_{\alpha})\cancel{\vol(W^{v},\hat{\mvec{X}}_{\alpha})}\\
&=&
d\Hd(\beta_\nu X_\nu)+d\Hd(W^{v})+
(-1)^{\alpha}\beta_\nu d\Hd(X_{\alpha})\delta_{\nu\alpha}(-1)^{\nu+1}\\
&=&
d\Hd(\beta_\nu X_\nu)+d\Hd(W^{v})+
(-1)^{\alpha}\beta_\alpha d\Hd(X_{\alpha})(-1)^{\alpha+1}\\
&=&d\Hd(W^{v})
\end{eqnarray*}

Computing equation~\refp{eq:DDWforms}(a) 
\begin{equation*}
\prld{\sect}^*\left(d\Hd \wedge d^{n}x_{\nu} 
  -\derp{\Hd}{x_\nu}\vol-\Psymp{\mu} \wedge d^{n-1}x_{\mu\nu}\right)=0.
\end{equation*}
Three terms have to be computed
\begin{eqnarray*}
\prld{\sect}^*(d\Hd \wedge d^{n}x_{\nu})(\basex{1},\ldots,\basex{n+1}) &=&\left(d\Hd \wedge d^{n}x_{\nu}\right)(\mvec{X}),\\
\prld{\sect}^*(\derp{\Hd}{x_\nu}\vol)(\basex{1},\ldots,\basex{n+1}) &=&\left(\derp{\Hd}{x_\nu}\vol\right)(\mvec{X})=\derp{\Hd}{x_\nu},\\
\prld{\sect}^*(\Psymp{\mu} \wedge d^{n-1}x_{\mu\nu})(\basex{1},\ldots,\basex{n+1}) &=&\left(\Psymp{\mu} \wedge d^{n-1}x_{\mu\nu}\right)(\mvec{X}).
\end{eqnarray*}
The last one may be computed knowing that
\begin{eqnarray*}
\Ocan(X_{\nu})&=&\Psymp{\mu} (X_{\nu})\wedge d^{n}x_{\mu}+\Psymp{\mu} \wedge d^{n}x_{\mu}(X_{\nu})\\
&=&
\Psymp{\mu} (X_{\nu})\wedge d^{n}x_{\mu}+\Psymp{\mu} \wedge d^{n-1}x_{\mu\nu}.
\end{eqnarray*}
So we have
\begin{equation*}
0=\Ocan(X_{\nu},\mvec{X})=
\left(\Psymp{\mu} (X_{\nu})\wedge d^{n}x_{\mu}\right)(\mvec{X})
+\left(\Psymp{\mu} \wedge d^{n-1}x_{\mu\nu}\right)(\mvec{X}),
\end{equation*}
\begin{eqnarray*}
\left(\Psymp{\mu} \wedge d^{n-1}x_{\mu\nu}\right)(\mvec{X})&=&
-\left(\Psymp{\mu} (X_{\nu})\wedge d^{n}x_{\mu}\right)(\mvec{X})\\
&=&-
(-1)^{\beta+1}\Psymp{\mu} (X_{\nu},X_{\beta})d^{n}x_{\mu}(\hat{\mvec{X}}_{\beta})\\
&=&-
(-1)^{\beta+1}\Psymp{\mu} (X_{\nu},X_{\beta})(-1)^{\beta+1}\delta_{ \mu \beta}\\
&=&
-\Psymp{\mu} (X_{\nu},X_{\mu})\\
&=&
\Psymp{\mu} (X_{\mu},X_{\alpha})dx^\alpha (X_{\nu}).
\end{eqnarray*}
The two other terms yield
\begin{eqnarray*}
\left(d\Hd \wedge d^{n}x_{\nu}-\derp{\Hd}{x_\nu}\vol\right)(\mvec{X})&=&
(-1)^{\beta+1}d\Hd (X_{\beta})d^{n}x_{\nu}(\hat{\mvec{X}}_{\beta})-\derp{\Hd}{x_\nu}\vol(\mvec{X})\\
&=&
(-1)^{\beta+1}d\Hd (X_{\beta})(-1)^{\nu+1}\delta_{ \nu \beta}-\derp{\Hd}{x_\nu}\\
&=&
d\Hd (X_{\nu})-\derp{\Hd}{x_\nu}.
\end{eqnarray*}
So, according to~\refp{eq:DDWforms}(a), a one-form $D\Hd$, that vanishes on vectors tangent to the optimal section $\prld{\sect}$, may be introduced 
\begin{equation}\label{anx:DH}
\begin{cases}
D\Hd=d\Hd-\derp{\Hd}{x_\mu}dx^{\mu}-\Psymp{\mu} (X_{\mu},X_{\alpha})dx^\alpha\\
D\Hd(X_{\nu})=0.
\end{cases}
\end{equation}
However, using partial derivatives for the differential $d\Hd=\derp{\Hd}{x_\alpha}dx^{\alpha}+\derp{\Hd}{y^A}dy^A+\derp{\Hd}{p^{\alpha}_A} dp^{\alpha}_A$ and taking into account the traditional de Donder Weyl equations~\refp{eq:1formsDDW} (b,c), it can be shown that the 1-form $d\Hd-\derp{\Hd}{x_\alpha}dx^{\alpha}-X_{\mu}\contr\Psymp{\mu}$, which differ slightly from $D\Hd$, vanishes identically 
\begin{equation}\label{eq:null}
d\Hd-\derp{\Hd}{x_\alpha}dx^{\alpha}-X_{\mu}\contr\Psymp{\mu}\equiv 0.
\end{equation}
So, for any vector $W\in\Chp{\JE}$ decomposed as $W=\beta_{\nu}X_{\nu}+W^{v}$, we have
\begin{eqnarray*}
D\Hd(W)&=&d\Hd(W)-\derp{\Hd}{x_\mu}dx^{\mu}(W)-\Psymp{\mu} (X_{\mu},X_{\alpha})dx^\alpha(W)\\
&=&
\underbrace{d\Hd(W)-\derp{\Hd}{x_\mu}dx^{\mu}(W)-\Psymp{\mu} (X_{\mu},W)}_{=0}+\Psymp{\mu} (X_{\mu},W^{v})\\
&=&
X_{\mu}\contr\Psymp{\mu}(W^{v})=\left(
\derp{\Hd}{y^A}\theta_\Ha^A+\derp{\Hd}{p^{\mu}_A}\mathcal{T}^{\mu}_A 
\right)(W^{v}).
\end{eqnarray*}

\section{Maurer-Cartan equation}\label{app:zero}
The differential of the Maurer-Cartan 1-form $\mcfLv$, dual basis of the left-invariant basis $\baseyTG{}$ defined by $\mcfL^A(\baseyTG{B})=\delta_{B}^A$, may be rely to the Lie bracket using the Maurer-Cartan equation 
\begin{equation}\label{eq:MCE}
d\mcfLv+\crochetL{\mcfLv}{\mcfLv}=0.
\end{equation}
Due to the vanishing right hand side, it is also named \emph{zero curvature equation} in the literature. Note that  a minus sign is required for a right invariance choice $d\mcfRv-\crochetL{\mcfRv}{\mcfRv}=0$.
Nevertheless, considering a left-invariant basis, let us recall that
if a section $\boldsymbol{\sect}$ is a representative of a point $\ptJE=(x^{\mu},\y{A},\xiL^A_\mu)$ of the 1-jet bundle $\JE$  over $\ptE=(x^{\mu},\y{A})$, then  
\begin{equation*}
\pullback\sect(\xiL^A_\mu)=\left.\mcfL^A\right|_{\ptE}(\derp{\boldsymbol{\sect}}{x^{\mu}}).
\end{equation*}
Using the pull back definition, it can be written
\begin{equation*}
\sect^*\mcfLv=\pullback\sect(\vLv_\mu dx^{\mu}),
\end{equation*}
since $\derp{\boldsymbol{\sect}}{x^{\mu}}$ corresponds to the push forward $T_{\boldsymbol{\sect}}\basex{\mu}$, of the basis vector $\basex{\mu}$. So, the differential $\dd \mcfLv$ gives in one hand
\begin{eqnarray*}
\sect^*\dd \mcfLv(\basex{i},\basex{j})&=&\pullback\sect(\dd \vLv_\mu dx^{\mu})(\basex{i},\basex{j})\\
&=&\dd \left(\pullback\sect \vLv_\mu\right)\wedge dx^{\mu}(\basex{i},\basex{j})\\
&=&\derp{\left(\pullback\sect \vLv_\mu\right)}{x_\nu}dx^{\nu}\wedge dx^{\mu}(\basex{i},\basex{j})\\
&=&\derp{\left(\pullback\sect \vLv_j\right)}{x_i}-\derp{\left(\pullback\sect \vLv_i\right)}{x_j}.
\end{eqnarray*}
On the other hand, using the Maurer-Cartan equation, 
\begin{eqnarray*}
\sect^*\dd \mcfLv(\basex{i},\basex{j})&=&-\sect^*\crochetL{\mcfLv}{\mcfLv}(\basex{i},\basex{j})\\
&=&-\crochetL{\sect^*\mcfLv}{\sect^*\mcfLv}(\basex{i},\basex{j})\\
&=&-\crochetL{\left(\pullback\sect \vLv_i\right)}{\left(\pullback\sect \vLv_\mu\right)}.
\end{eqnarray*}
So, knowing that $\pullback\sect (\dd \vLv_\nu \wedge d^{n}x_{\mu})=\derp{\left(\pullback\sect \vLv_\nu\right)}{x_\mu}\vol$, this gives the useful identity
\begin{equation}
\pullback\sect\left(
\dd \vL{\nu}{A}\wedge\dnx{n}{\mu}-
\dd \vL{\mu}{A}\wedge\dnx{n}{\nu}
+\crochetL{\vLv_{\mu}}{\vLv_{\nu}}^A\vol
\right)=0.
\end{equation}
\section{Other useful formulae}\label{app:OUF}
\subsection{Volume form and contractions}
\begin{align}\label{eq:dnx}
d^{n}x_{\mu}=\basex{\mu}\contr \vol&=(-1)^{\mu+1}dx^1\wedge\ldots\wedge
dx^{\mu-1}\wedge dx^{\mu+1}\wedge \ldots \wedge dx^{n+1}\\
&=(-1)^{\mu+1}dx^1\wedge\ldots\wedge
dx^{\mu-1}\wedge \widehat{dx^{\mu}}\wedge dx^{\mu+1}\wedge \ldots \wedge dx^{n+1}
\end{align}
\begin{equation}\label{id:01}
\dx{\nu}\wedge\dnx{n}{\mu}=\vol\delta_{\mu\nu}
\end{equation}
\begin{equation}\label{id:02}
\dx{\mu}\wedge\dnx{n}{\mu}=(n+1)\vol
\end{equation}
\begin{equation*}
d^{n-1}x_{\mu\nu} = \partial_{x_\nu}\contr d^nx_\mu
	= \begin{cases}
	(-1)^{\mu+\nu}dx^1 \wedge \cdots \wedge \widehat{dx^{\nu}}\wedge\cdots \wedge  \widehat{dx^{\mu}}\wedge \cdots dx^{n+1}, & \nu < \mu\\
	0, & \nu=\mu,\\
	(-1)^{\mu+\nu+1}dx^1 \wedge \cdots \wedge\widehat{dx^{\mu}}\wedge\cdots \wedge  \widehat{dx^{\nu}}\wedge \cdots dx^{n+1}, &\nu > \mu,
	\end{cases}
\end{equation*}
\begin{equation}\label{id:03}
\dx{\alpha}\wedge\dnx{n-1}{\mu\nu}=\dnx{n}{\mu}\delta_{\alpha\nu}-\dnx{n}{\nu}\delta_{\alpha\mu}
\end{equation}
\begin{equation}\label{id:04}
\dx{\mu}\wedge\dnx{n-1}{\mu\nu}=-n\dnx{n}{\nu}\quad
\dx{\nu}\wedge\dnx{n-1}{\mu\nu}=+n\dnx{n}{\mu}
\end{equation}

\subsection{$\mathfrak{g}$-valued 1-forms}
Letting $(\baseyTG{1},\baseyTG{2},\ldots,\baseyTG{N})$ be a basis for $\mathfrak{g}$, a $\mathfrak{g}$-valued 1-form may be written $\b\alpha=\alpha^i\otimes\baseyTG{i}$ where the $\alpha^i$ are a 1-form. The "bracket" on $\mathfrak{g}$-valued 1-forms is a combination of exterior product and bracket operation
\begin{eqnarray*}
\crochetL{\b\alpha}{\b\beta}(X,Y)&=&\crochetL{\b\alpha(X)}{\b\beta(Y)}\\
\crochetF{\b\alpha}{\b\beta}(X,Y)&=&\crochetL{\b\alpha(X)}{\b\beta(Y)}-\crochetL{\b\alpha(Y)}{\b\beta(X)}.
\end{eqnarray*}
Writing $\b\alpha=\alpha^i\otimes\baseyTG{i}$ and $\b\beta=\beta^j\otimes\baseyTG{j}$, we have
\begin{equation*}
\crochetF{\b\alpha}{\b\beta}=\alpha^i\wedge\beta^j\otimes\crochetL{\baseyTG{i}}{\baseyTG{j}}.
\end{equation*}
In particular, we have the symmetry
\begin{equation*}
\crochetF{\b\alpha}{\b\beta}=\crochetL{\b\alpha}{\b\beta}+\crochetL{\b\beta}{\b\alpha},
\end{equation*}
which gives
\begin{equation*}
\crochetF{\b\alpha}{\b\alpha}=2\crochetL{\b\alpha}{\b\alpha}.
\end{equation*}
The followings identities are useful. Let $\b\alpha$ and $\b\beta$ be $\mathfrak{g}$-valued 1-forms, $\vLv$ a vector in $\mathfrak{g}$ and $\b{\pi}$ a covector in $\mathfrak{g}^*$
\begin{eqnarray*}
&&\crochetF{\b\alpha}{\vLv_{\nu}\dx{\nu}}=\crochetL{\b\alpha}{\vLv_{\nu}}\wedge\dx{\nu}\\
&&\dual{\momL{}{A}}{\crochetF{\b\alpha}{\b\beta}^A}=\left(\ade{\b\alpha}{\b{\pi}}\right)_A\wedge\beta^A\\
&&\dual{\momL{}{A}}{\crochetF{\b\alpha}{\vLv}^A}=\left(\ade{\b\alpha}{\b{\pi}}\right)_A\vL{}{A}.
\end{eqnarray*}

%\part{References}
\addcontentsline{toc}{part}{References}

\bibliography{UM}
\bibliographystyle{apsrev4-1custom}

\end{document}

%% file: Packages.tex
\usepackage[pdftex]{graphicx}  
\DeclareGraphicsExtensions{.eps, .ps, .eps.gz, .ps.gz,.pdf}
\usepackage{asymptote}
\usepackage{media9}
\usepackage{amsmath}
\usepackage{amssymb}
\usepackage{mathtools}
\usepackage[makeroom]{cancel}
\usepackage{amsthm}
\usepackage[T1]{fontenc}
\usepackage[utf8]{inputenc}              
\usepackage{amsfonts}  
\usepackage{bbm} 
\usepackage{mathrsfs}%pour avoir le C rond du Cinfty
\usepackage{subfigure}  
\usepackage{eurosym}
\usepackage{todonotes}
\usepackage{upgreek}
\usepackage{enumerate}
\usepackage{scalefnt}
\usepackage{rotating}
\usepackage{multirow}
\usepackage[small,nohug,heads=vee]{diagrams}
\usepackage{diagrams}
\usepackage{color}                          % colors
\usepackage{framed}                         % frames
\usepackage{mdframed}                       % better frames

%% file: Commandes.tex
\renewcommand{\b}[1]{\boldsymbol{ #1}}
\newcommand{\refp}[1]{(\ref{#1})}			%citation  d equations avec parentheses autour
\newcommand{\ie}{{\textit{i.e.~}}}

\newcommand\R{\ensuremath{\mathbbm R}}     % ensemble R    ou ça : \mathrm{I\!R}
\newcommand\C{\ensuremath{\mathscr{C}}}     % C rond de C\infinity

\newcommand{\vect}[1]{\mathbf{#1}}  

\newcommand{\mat}[1]{\mathbb{#1}}  
\newcommand{\dd}{\ensuremath{\mathrm d}} % d droit

\newcommand\Eps{\ensuremath{\varepsilon}}         %beau epsilon
\newcommand{\bpm}{\ensuremath{\begin{pmatrix}}}        %balise de début de matrice
\newcommand{\epm}{\ensuremath{\end{pmatrix}}}            %balise de fin de matrice
\newcommand{ \mvec}[1]{\overset{\circ}{#1}} %multi-vecteur (rond au dessus de l'argument)
\newcommand{	\Hol}{\mathsf{Hol}}
\newcommand{\cev}[1]{\reflectbox{\ensuremath{\vec{\reflectbox{\ensuremath{#1}}}}}}%flèche d'un vecteur vers la gauche
\newcommand{\Lag}{\mathcal{L}}
\newcommand{\Ld}{\mathsterling}%density 
\newcommand{\lag}{\ell}
\newcommand{\ld}{\mathbbm{l}}%density

\newcommand{\ldR}{\mathbbm{l'}}
\newcommand{\Ha}{\mathcal{H}}
\newcommand{\ha}{h}
\newcommand{\haR}{h'}
\newcommand{\Hd}{\mathbbm{H}}
\newcommand{\hd}{\mathbbm{h}}
\newcommand{\hdR}{\mathbbm{h'}}

\newcommand{\mcfL}{\lambda}
\newcommand{\mcfLv}{\boldsymbol{\lambda}}
\newcommand{\mcfR}{\rho}
\newcommand{\mcfRv}{\boldsymbol{\rho}}

\newcommand{\fcv}{\ensuremath{\b\theta}}
\newcommand{\fcvG}{\ensuremath{\b\vartheta}}
\newcommand{\fc}[1]{\ensuremath{\theta^{#1}}}
\newcommand{\fcG}[1]{\ensuremath{\vartheta^{#1}}}
\newcommand{\vol}{\ensuremath{\omega}}
\newcommand{\dnx}[2]{\ensuremath{\dd^{#1}x_{#2}}}
\newcommand{\PC}{\ensuremath{\Theta}}
\newcommand{\PCL}{\ensuremath{\Theta_\Lag}}

\newcommand{\PCl}{\ensuremath{\Theta_\lag}}

\newcommand{\PS}{\ensuremath{\Omega	}}

\newcommand{\PSL}{\ensuremath{\Omega_\Lag}}

\newcommand{\PSl}{\ensuremath{\Omega_\lag}}
\newcommand{\PSh}{\ensuremath{\Omega_\ha}}

\newcommand{\Psymp}[1]{\ensuremath{\Omega}^{#1}}

\newcommand{\Ocan}{\ensuremath{\Omega^*}}
\newcommand{\Ocanon}{\ensuremath{\Omega^*}}

\newcommand{\ps}{\PSh^*}
\newcommand{\psymp}[1]{\PSh^{#1}}
\newcommand{\sect}{\ensuremath{\phi}}

\newcommand{\Chp}[1]{\ensuremath{\chi(#1)}}
\newcommand{\ChpV}[2]{\ensuremath{\chi^{V(#1)}(#2)}}

\newcommand{\prolonge}[1]{j^1{#1}}
\newcommand{\pullback}[1]{(j^{1}{#1})^{*}} 
\newcommand{\prld}[1]{\tilde{j}^{1}{#1}}
\newcommand{\pullbackd}[1]{(\prld{#1})^{*}} 

\newcommand{\jpi}{\jmath \pi}

\newcommand{\FLe}{\mathbb{F}\mathfrak{L}}

\newcommand{\x}[1]{{x}^{#1}}
\newcommand{\y}[1]{{y}^{#1}}
\renewcommand{\v}[2]{v_{#1}^{#2}} %attention efface la commande latex de base
\newcommand{\z}[1]{\zeta^{#1}} 
\newcommand{\zv}{\boldsymbol{\zeta}} 
\newcommand{\zG}[1]{\Xi^{#1}} 
 
\newcommand{\g}[1]{{g}^{#1}}

\newcommand{\xiL}{\xi}
\newcommand{\xiR}{\chi}

\newcommand{\vL}[2]{\xiL_{#1}^{#2}}

\newcommand{\vLv}{\boldsymbol{\xi}}

\newcommand{\piL}{\pi}
\newcommand{\piR}{\Pi}
\newcommand{\momL}[2]{\piL^{#1}_{#2}}
\newcommand{\momR}[2]{\piR^{#1}_{#2}}

\newcommand{\basex}[1]{\vec{\partial}_{#1}}
\newcommand{\basey}[1]{\vec{\partial}_{#1}}
\newcommand{\basev}[2]{\vec{\partial}_{#1}^{#2}}

\newcommand{\baseyTG}[1]{\vec{e}_{#1}}

\newcommand{\baseyL}[1]{\vec{e}_{#1}}
\newcommand{\baseyR}[1]{\cev{e}_{#1}}
\newcommand{\basevL}[2]{\vec{\Game}^{#1}_{#2}}
\newcommand{\basevR}[2]{\cev{\Game}^{#1}_{#2}}

\newcommand{\dx}[1]{\dd \x{#1}}
\newcommand{\dy}[1]{\dd \y{#1}}
\newcommand{\dv}[2]{\dd \v{#1}{#2}}
\newcommand{\ptE}{\ensuremath{\textsc{p}}}
\newcommand{\ptJE}{\ensuremath{\bar \ptE}}
\renewcommand{\gg}{\ensuremath{\b g}} %supprime la commande latex de base

\newcommand{\JE}{\ensuremath{J^1E}}
\newcommand{\JEpt}[1]{\ensuremath{J^1_{#1}E}}

\newcommand{\der}[2]{\frac{d{#1}}{d{#2}}}
\newcommand{\derp}[2]{\frac{\partial#1}{\partial{#2}}}
\newcommand{\derpm}[3]{\frac{\partial^2#1}{\partial{#2}\partial{#3}}}
\newcommand{\dpv}[1]{\partial_{#1}}
\newcommand{\contr}{\,\lrcorner\,}

\newcommand{\crochetL}[2]{\left[#1,#2\right]}
\newcommand{\crochetF}[2]{\left[#1\wedge#2\right]}
\newcommand{\dual}[2]{\left(#1,#2\right)}
\newcommand{\Lie}[1]{\ensuremath{\mathsf{L}_{#1}}}

\newcommand{	\Ad}[2]{\text{Ad}_{#1}#2}
\newcommand{	\Ade}[2]{\text{Ad}^*_{#1}#2}
\newcommand{	\ad}[2]{\text{ad}_{#1}#2}
\newcommand{	\ade}[2]{\text{ad}^*_{#1}#2}

\newcommand{\pos}{\mathbf{r}}
\newcommand{\Rot}{\mat{R}}
\newcommand{\s}{s}

\newcommand{\gre}{\mathbf{H}}
\newcommand{\ig}{\mathbf{H}^{-1}}
\newcommand{\xp}{\mathbf{x}}
\newcommand{\Yo}{\mathbf{w}_{0}}
\newcommand{\Eun}{\vect{E}_1}

\newcommand{\Ps}{\boldsymbol{\pi}_{R}}
\newcommand{\Pc}{\boldsymbol{\pi}_{L}}

\newcommand{\Ss}{\boldsymbol{\sigma}_{R}}
\newcommand{\Sc}{\boldsymbol{\sigma}_{L}}

\newcommand{\Ep}{\boldsymbol{\epsilon}}

\newcommand{\X}{\boldsymbol{\chi}}

\newcommand{\om}{\boldsymbol{\omega}}  
\newcommand{\vt}{\vect{v}} 
\newcommand{\m}{\vect{m}} 
\newcommand{\n}{\vect{p}}

\newarrow {Corresponds} <--->

\newtheoremstyle{myThmStyle}% name of the theorem
{6pt}% habovespacei
{6pt}% hbelowspacei
{\itshape}% font
{}% hindenti
{}% hheadfonti
{}% hheadpuncti
{ }% hheadspacei
{\thmname{\normalfont\textsc{#1}} \thmnumber{\normalfont #2} \thmnote{#3}\\}

\definecolor{thmBgColor}{RGB}{250,250,250}
\definecolor{thmLnColor}{RGB}{200,200,200}

\mdfdefinestyle{MDFStyGrayScreen}{%
    linecolor=thmLnColor,
        backgroundcolor=thmBgColor,
        linewidth=1pt,
        topline=true,
    bottomline=true,
    rightline=false,
        leftline=false,
    outerlinewidth=2pt,
    roundcorner=0pt,
    innertopmargin=4pt, %\baselineskip
    innerbottommargin=4pt, %\baselineskip,
    innerrightmargin=3pt,
    innerleftmargin=3pt,
        skipabove=\topskip,
        skipbelow=\topskip,
        nobreak=true
        }

\theoremstyle{myThmStyle}
\newtheorem{myTheorem}{Theorem}[section]

\newenvironment{theorem}[1]
{
\begin{mdframed}[style=MDFStyGrayScreen]
\begin{myTheorem}#1
}
{
\end{myTheorem}
\end{mdframed}
}

\theoremstyle{myThmStyle}
\newtheorem{myDefinition}[myTheorem]{Definition}%[section]
\newenvironment{definition}[1]
{
\begin{mdframed}[style=MDFStyGrayScreen]
\begin{myDefinition}#1
}
{
\end{myDefinition}
\end{mdframed}
}

\theoremstyle{myThmStyle}
\newtheorem{myProposition}[myTheorem]{Proposition}%[section]
\newenvironment{proposition}[1]
{
\begin{mdframed}[style=MDFStyGrayScreen]
\begin{myProposition}#1
}
{
\end{myProposition}
\end{mdframed}
}

\theoremstyle{myThmStyle}
\newtheorem{myLemma}[myTheorem]{Lemma}%[section]
\newenvironment{lemma}[1]
{
\begin{mdframed}[style=MDFStyGrayScreen]
\begin{myLemma}#1
}
{
\end{myLemma}
\end{mdframed}
}

\theoremstyle{myThmStyle}
\newtheorem{myCorollary}[myTheorem]{Corollary}%[section]
\newenvironment{corollary}[1]
{
\begin{mdframed}[style=MDFStyGrayScreen]
\begin{myCorollary}#1
}
{
\end{myCorollary}
\end{mdframed}
}

%% file: Multisymplectic_bensoam_bauge_article.bbl
%Control: production of eprint (1) enabled
\begin{thebibliography}{19}%
\makeatletter
\providecommand \@ifxundefined [1]{%
 \@ifx{#1\undefined}
}%
\providecommand \@ifnum [1]{%
 \ifnum #1\expandafter \@firstoftwo
 \else \expandafter \@secondoftwo
 \fi
}%
\providecommand \@ifx [1]{%
 \ifx #1\expandafter \@firstoftwo
 \else \expandafter \@secondoftwo
 \fi
}%
\providecommand \natexlab [1]{#1}%
\providecommand \enquote  [1]{``#1''}%
\providecommand \bibnamefont  [1]{#1}%
\providecommand \bibfnamefont [1]{#1}%
\providecommand \citenamefont [1]{#1}%
\providecommand \href@noop [0]{\@secondoftwo}%
\providecommand \href [0]{\begingroup \@sanitize@url \@href}%
\providecommand \@href[1]{\@@startlink{#1}\@@href}%
\providecommand \@@href[1]{\endgroup#1\@@endlink}%
\providecommand \@sanitize@url [0]{\catcode `\\12\catcode `\$12\catcode
  `\&12\catcode `\#12\catcode `\^12\catcode `\_12\catcode `\%12\relax}%
\providecommand \@@startlink[1]{}%
\providecommand \@@endlink[0]{}%
\providecommand \url  [0]{\begingroup\@sanitize@url \@url }%
\providecommand \@url [1]{\endgroup\@href {#1}{\urlprefix }}%
\providecommand \urlprefix  [0]{URL }%
\providecommand \Eprint [0]{\href }%
\providecommand \doibase [0]{http://dx.doi.org/}%
\providecommand \selectlanguage [0]{\@gobble}%
\providecommand \bibinfo  [0]{\@secondoftwo}%
\providecommand \bibfield  [0]{\@secondoftwo}%
\providecommand \translation [1]{[#1]}%
\providecommand \BibitemOpen [0]{}%
\providecommand \bibitemStop [0]{}%
\providecommand \bibitemNoStop [0]{.\EOS\space}%
\providecommand \EOS [0]{\spacefactor3000\relax}%
\providecommand \BibitemShut  [1]{\csname bibitem#1\endcsname}%
\let\auto@bib@innerbib\@empty
%</preamble>
\bibitem [{\citenamefont {Marle}(2003)}]{Marle2003}%
  \BibitemOpen
  \bibfield  {author} {\bibinfo {author} {\bibfnamefont {C.~M.}\ \bibnamefont
  {Marle}},\ }\bibfield  {title} {\enquote {\bibinfo {title} {On mechanical
  systems with a lie group as configuration space},}\ }in\ \href@noop {} {\emph
  {\bibinfo {booktitle} {Mathematical Physics Studies, Jean Leray ’99
  Conference Proceedings}}},\ Vol.~\bibinfo {volume} {24}\ (\bibinfo {year}
  {2003})\ pp.\ \bibinfo {pages} {183--203}\BibitemShut {NoStop}%
\bibitem [{\citenamefont {Euler}(1765)}]{Euler1765}%
  \BibitemOpen
  \bibfield  {author} {\bibinfo {author} {\bibfnamefont {L.}~\bibnamefont
  {Euler}},\ }\href@noop {} {\emph {\bibinfo {title} {Theoria motus corporum
  solidorum seu rigidorum}}}\ (\bibinfo  {publisher} {AE Roser.},\ \bibinfo
  {year} {1765})\BibitemShut {NoStop}%
\bibitem [{\citenamefont {Arnold}(1966)}]{Arnold:1966}%
  \BibitemOpen
  \bibfield  {author} {\bibinfo {author} {\bibfnamefont {V.~I.}\ \bibnamefont
  {Arnold}},\ }\bibfield  {title} {\enquote {\bibinfo {title} {Sur la
  g\'{e}om\'{e}trie diff\'{e}rentielle des groupes de {L}ie de dimension
  infinie et ses applications \`{a} l'hydrodynamique des fluides parfaits},}\
  }\href@noop {} {\bibfield  {journal} {\bibinfo  {journal} {Annales l'Ins.
  Fourier}\ ,\ \bibinfo {pages} {319--361}} (\bibinfo {year}
  {1966})}\BibitemShut {NoStop}%
\bibitem [{\citenamefont {Marsden}\ and\ \citenamefont
  {Ratiu}(1999)}]{marsden:1999}%
  \BibitemOpen
  \bibfield  {author} {\bibinfo {author} {\bibfnamefont {J.}~\bibnamefont
  {Marsden}}\ and\ \bibinfo {author} {\bibfnamefont {T.}~\bibnamefont
  {Ratiu}},\ }\href@noop {} {\emph {\bibinfo {title} {Introduction to Mechanics
  and Symmetry}}}\ (\bibinfo  {publisher} {2nd edn. Springer-Verlag},\ \bibinfo
  {year} {1999})\BibitemShut {NoStop}%
\bibitem [{\citenamefont {De~Donder}(1930)}]{de1930theorie}%
  \BibitemOpen
  \bibfield  {author} {\bibinfo {author} {\bibfnamefont {T.}~\bibnamefont
  {De~Donder}},\ }\href@noop {} {\emph {\bibinfo {title} {Th{\'e}orie
  invariantive du calcul des variations}}}\ (\bibinfo  {publisher} {Paris},\
  \bibinfo {year} {1930})\BibitemShut {NoStop}%
\bibitem [{\citenamefont {Weyl}(1935)}]{weyl1935geodesic}%
  \BibitemOpen
  \bibfield  {author} {\bibinfo {author} {\bibfnamefont {H.}~\bibnamefont
  {Weyl}},\ }\bibfield  {title} {\enquote {\bibinfo {title} {Geodesic fields in
  the calculus of variation for multiple integrals},}\ }\href@noop {}
  {\bibfield  {journal} {\bibinfo  {journal} {Annals of Mathematics}\ ,\
  \bibinfo {pages} {607--629}} (\bibinfo {year} {1935})}\BibitemShut {NoStop}%
\bibitem [{\citenamefont {Carath{\'e}odory}(1999)}]{caratheodory1999calculus}%
  \BibitemOpen
  \bibfield  {author} {\bibinfo {author} {\bibfnamefont {C.}~\bibnamefont
  {Carath{\'e}odory}},\ }\href@noop {} {\emph {\bibinfo {title} {Calculus of
  variations and partial differential equations of the first order}}}\
  (\bibinfo  {publisher} {Chelsea Publishing Company},\ \bibinfo {year}
  {1999})\BibitemShut {NoStop}%
\bibitem [{\citenamefont {Souriau}(1970)}]{Souriau:1970}%
  \BibitemOpen
  \bibfield  {author} {\bibinfo {author} {\bibfnamefont {J.-M.}\ \bibnamefont
  {Souriau}},\ }\href@noop {} {\emph {\bibinfo {title} {Structure des systèmes
  dynamiques}}}\ (\bibinfo  {publisher} {Dunod},\ \bibinfo {year}
  {1970})\BibitemShut {NoStop}%
\bibitem [{\citenamefont {Kanatchikov}(1998)}]{kanatchikov1998canonical}%
  \BibitemOpen
  \bibfield  {author} {\bibinfo {author} {\bibfnamefont {I.~V.}\ \bibnamefont
  {Kanatchikov}},\ }\bibfield  {title} {\enquote {\bibinfo {title} {Canonical
  structure of classical field theory in the polymomentum phase space},}\
  }\href@noop {} {\bibfield  {journal} {\bibinfo  {journal} {Reports on
  Mathematical Physics}\ }\textbf {\bibinfo {volume} {41}},\ \bibinfo {pages}
  {49--90} (\bibinfo {year} {1998})}\BibitemShut {NoStop}%
\bibitem [{\citenamefont {Castrill\'on~L\'opez}\ and\ \citenamefont {Garc\'\i~a
  P\'erez}(2001)}]{Lopez:2001}%
  \BibitemOpen
  \bibfield  {author} {\bibinfo {author} {\bibfnamefont {M.}~\bibnamefont
  {Castrill\'on~L\'opez}}\ and\ \bibinfo {author} {\bibfnamefont {P.~L.}\
  \bibnamefont {Garc\'\i~a P\'erez}},\ }\href@noop {} {\enquote {\bibinfo
  {title} {Multidimensional euler-poincar\'e equations.}}\ } (\bibinfo {year}
  {2001})\BibitemShut {NoStop}%
\bibitem [{\citenamefont {Demoures}\ \emph {et~al.}(2013)\citenamefont
  {Demoures}, \citenamefont {Gay-Balmaz},\ and\ \citenamefont
  {Ratiu}}]{demoures2013multisymplectic}%
  \BibitemOpen
  \bibfield  {author} {\bibinfo {author} {\bibfnamefont {F.}~\bibnamefont
  {Demoures}}, \bibinfo {author} {\bibfnamefont {F.}~\bibnamefont
  {Gay-Balmaz}}, \ and\ \bibinfo {author} {\bibfnamefont {T.~S.}\ \bibnamefont
  {Ratiu}},\ }\bibfield  {title} {\enquote {\bibinfo {title} {Multisymplectic
  variational integrators and space/time symplecticity},}\ }\href@noop {}
  {\bibfield  {journal} {\bibinfo  {journal} {arXiv preprint arXiv:1310.4772}\
  } (\bibinfo {year} {2013})}\BibitemShut {NoStop}%
\bibitem [{\citenamefont {Cartan}(1922)}]{Cartan:1922}%
  \BibitemOpen
  \bibfield  {author} {\bibinfo {author} {\bibfnamefont {E.}~\bibnamefont
  {Cartan}},\ }\href@noop {} {\emph {\bibinfo {title} {Leçons sur les
  invariants intégraux}}}\ (\bibinfo  {publisher} {Paris, Librairie
  Scientifique A. Hermann \& Fils},\ \bibinfo {year} {1922})\BibitemShut
  {NoStop}%
\bibitem [{\citenamefont {Cartan}(1933)}]{Cartan:1933}%
  \BibitemOpen
  \bibfield  {author} {\bibinfo {author} {\bibfnamefont {E.}~\bibnamefont
  {Cartan}},\ }\href@noop {} {\emph {\bibinfo {title} {Les espaces m\'etriques
  fond\'es sur la notion d’aire}}}\ (\bibinfo  {publisher} {Paris, Librairie
  Scientifique A. Hermann \& Fils, 6 rue de la Sorbonne},\ \bibinfo {year}
  {1933})\BibitemShut {NoStop}%
\bibitem [{\citenamefont {Echeverría-enríquez}\ \emph
  {et~al.}(2000)\citenamefont {Echeverría-enríquez}, \citenamefont
  {Muñoz-lec},\ and\ \citenamefont {Román-roy}}]{Echeverria00}%
  \BibitemOpen
  \bibfield  {author} {\bibinfo {author} {\bibfnamefont {A.}~\bibnamefont
  {Echeverría-enríquez}}, \bibinfo {author} {\bibfnamefont {M.~C.}\
  \bibnamefont {Muñoz-lec}}, \ and\ \bibinfo {author} {\bibfnamefont
  {N.}~\bibnamefont {Román-roy}},\ }\bibfield  {title} {\enquote {\bibinfo
  {title} {Geometry of multisymplectic hamiltonian first-order field
  theories},}\ }\href@noop {} {\bibfield  {journal} {\bibinfo  {journal} {J.
  Math. Phys}\ ,\ \bibinfo {pages} {7402--7444}} (\bibinfo {year}
  {2000})}\BibitemShut {NoStop}%
\bibitem [{\citenamefont {Marsden}\ and\ \citenamefont
  {Shkoller}(1999)}]{Marsden99}%
  \BibitemOpen
  \bibfield  {author} {\bibinfo {author} {\bibfnamefont {J.~E.}\ \bibnamefont
  {Marsden}}\ and\ \bibinfo {author} {\bibfnamefont {S.}~\bibnamefont
  {Shkoller}},\ }\bibfield  {title} {\enquote {\bibinfo {title}
  {Multisymplectic geometry, covariant hamiltonians and water waves},}\ }in\
  \href@noop {} {\emph {\bibinfo {booktitle} {Mathematical Proceedings of the
  Cambridge Philosophical Society 125}}}\ (\bibinfo {year} {1999})\ pp.\
  \bibinfo {pages} {553--575}\BibitemShut {NoStop}%
\bibitem [{\citenamefont {Echeverría-Enríquez}\ \emph
  {et~al.}(2007)\citenamefont {Echeverría-Enríquez}, \citenamefont
  {de~León}, \citenamefont {Muñoz-Lecanda},\ and\ \citenamefont
  {Román-Roy}}]{Echeverria2007}%
  \BibitemOpen
  \bibfield  {author} {\bibinfo {author} {\bibfnamefont {A.}~\bibnamefont
  {Echeverría-Enríquez}}, \bibinfo {author} {\bibfnamefont {M.}~\bibnamefont
  {de~León}}, \bibinfo {author} {\bibfnamefont {M.~C.}\ \bibnamefont
  {Muñoz-Lecanda}}, \ and\ \bibinfo {author} {\bibfnamefont {N.}~\bibnamefont
  {Román-Roy}},\ }\bibfield  {title} {\enquote {\bibinfo {title} {Extended
  hamiltonian systems in multisymplectic field theories},}\ }\href {\doibase
  http://dx.doi.org/10.1063/1.2801875} {\bibfield  {journal} {\bibinfo
  {journal} {Journal of Mathematical Physics}\ }\textbf {\bibinfo {volume}
  {48}},\ \bibinfo {eid} {112901} (\bibinfo {year} {2007})}\BibitemShut
  {NoStop}%
\bibitem [{\citenamefont {Marsden}\ and\ \citenamefont
  {SHKOLLER}(1999)}]{marsden_shkoller}%
  \BibitemOpen
  \bibfield  {author} {\bibinfo {author} {\bibfnamefont {J.}~\bibnamefont
  {Marsden}}\ and\ \bibinfo {author} {\bibfnamefont {S.}~\bibnamefont
  {SHKOLLER}},\ }\bibfield  {title} {\enquote {\bibinfo {title}
  {Multisymplectic geometry, covariant hamiltonians, and water waves},}\
  }\href@noop {} {\bibfield  {journal} {\bibinfo  {journal} {Math. Proc. Camb.
  Phil. Soc.}\ }\textbf {\bibinfo {volume} {125}},\ \bibinfo {pages} {553--575}
  (\bibinfo {year} {1999})}\BibitemShut {NoStop}%
\bibitem [{\citenamefont {H\'elein}(2004)}]{Helein2004c}%
  \BibitemOpen
  \bibfield  {author} {\bibinfo {author} {\bibfnamefont {F.}~\bibnamefont
  {H\'elein}},\ }\href@noop {} {\emph {\bibinfo {title} {Noncompact Problems at
  the Intersection of Geometry, Analysis and Topology Contemporary
  Mathematics}}},\ Vol.\ \bibinfo {volume} {350}\ (\bibinfo  {publisher} {A.
  Bahri, S. Klainerman et M. Vogelius, eds. AMS},\ \bibinfo {year} {2004})\
  pp.\ \bibinfo {pages} {127--147}\BibitemShut {NoStop}%
\bibitem [{\citenamefont {Jost}(2005)}]{jost2005riemannian}%
  \BibitemOpen
  \bibfield  {author} {\bibinfo {author} {\bibfnamefont {J.}~\bibnamefont
  {Jost}},\ }\href@noop {} {\emph {\bibinfo {title} {Riemannian Geometry and
  Geometric Analysis}}},\ Universitext (1979)\ (\bibinfo  {publisher}
  {Springer},\ \bibinfo {year} {2005})\BibitemShut {NoStop}%
\end{thebibliography}%
